\documentclass[10pt,reqno]{amsart}
\usepackage{amsfonts,amssymb,amsmath,amsopn,amsthm,graphicx,dsfont}
\usepackage{amsxtra,mathrsfs}
\usepackage[mathscr]{eucal}
\usepackage{colordvi}
\usepackage[usenames,dvipsnames]{color}
\usepackage{mathtools}
\usepackage[colorlinks=true]{hyperref}
\usepackage{setspace,scalerel}

\numberwithin{equation}{section}

\usepackage[parfill]{parskip}    

\theoremstyle{definition}

\newtheorem{thm}{Theorem}[section]
\newtheorem{cor}[thm]{Corollary}
\newtheorem{lem}[thm]{Lemma}
\newtheorem{prop}[thm]{Proposition}

\newtheorem{assumption}[thm]{Assumption}
\newtheorem{defin}[thm]{Definition}

\newtheorem{rem}[thm]{Remark}
\newtheorem{rems}[thm]{Remarks}

\theoremstyle{remark}
\newcommand{\B}{\mathscr{B}}
\newcommand{\s}{\mathbb{S}}

\newcommand{\HH}{{\rm H}}
\newcommand{\Lp}{\textsf{L}\!}
\newcommand{\Ln}{\mathcal{L}}
\newcommand{\eps}{\varepsilon}
\newcommand{\N}{\mathbb{N}}
\newcommand{\R}{\mathbb{R}}

\newcommand{\J}{\mathcal{J}}
\newcommand{\sn}{\mathcal{S}}

\newcommand{\Z}{\mathbb{Z}}

\newcommand{\rt}{{\rm curl}\, }
\newcommand{\C}{\mathbb{C}}
\newcommand{\cc}{\mathfrak{c}}

\newcommand{\D}{\mathscr{D}}
\newcommand{\G}{\mathcal{G}}
\newcommand{\K}{\mathscr{K}}
\newcommand{\X}{\mathcal{X}}
\newcommand{\A}{\mathpzc{A}}
\newcommand{\E}{\mathcal{E}}
\newcommand{\V}{{\rm U}}
\newcommand{\NN}{\mathcal{N}}
\newcommand{\po}{{\rm P_{0}}}
\newcommand{\pom}{{\rm P^{\scriptscriptstyle -}_{\!0}}}
\newcommand{\pop}{{\rm P^{\scriptscriptstyle +}_{\!0}}}
\newcommand{\g}{\mathfrak{g}}

\newcommand{\Q}{\mathcal{Q}}
\newcommand{\p}{\mathscr{R}\, }

\newcommand{\pd}{\partial}

\newcommand{\LL}{\mathcal{\big\langle}}
\newcommand{\RR}{\mathcal{\big\rangle}}
\newcommand{\id}{\mathds{1}}

\newcommand{\bb}{\mathfrak{b}}
\newcommand{\hk}{\mathfrak{h}}
\newcommand{\hs}{\rm{hs}}
\newcommand{\w}{{\rm w}}
\newcommand{\x}{\langle x \rangle }

\newcommand{\Rm}{R_{\scriptscriptstyle-}\!}
\newcommand{\Rp}{R_{\scriptscriptstyle+}\!}
\newcommand{\Pm}{P_{\!\scriptscriptstyle-}}
\newcommand{\Pp}{P_{\!\scriptscriptstyle+}}

\newcommand{\m}{{\scriptscriptstyle-}}
\newcommand{\pp}{{\scriptscriptstyle+}}
\newcommand{\ppm}{{\scriptscriptstyle\pm}}

\newcommand{\psim}{\psi^{\scriptscriptstyle-}}
\newcommand{\psip}{\psi^{\scriptscriptstyle+}}

\newcommand{\Psim}{\Psi^{\scriptscriptstyle-}}
\newcommand{\Psip}{\Psi^{\scriptscriptstyle+}}

\newcommand{\nup}{\nu^{\scriptscriptstyle+}}

\newcommand{\mup}{\mu^{\scriptscriptstyle+}}

\newcommand{\Phim}{\Phi^{\scriptscriptstyle-}}
\newcommand{\Phip}{\Phi^{\scriptscriptstyle+}}

\newcommand{\vp}{{\rm W}_{\!\scriptscriptstyle+}}
\newcommand{\vm}{{\rm W}_{\!\scriptscriptstyle-}}
\newcommand{\vpm}{{\rm W}_{\!\scriptscriptstyle\pm}}

\newcommand{\wth}{\widehat  }

\newcommand{\alp}{{\alpha'}}

\newcommand{\rig}{\big\rangle}
\newcommand{\lef}{\big\langle}

\DeclareMathAlphabet{\mathpzc}{OT1}{pzc}{m}{it}

\DeclareFontFamily{OT1}{pzc}{}
\DeclareFontShape{OT1}{pzc}{m}{it}{<-> s * [1.15] pzcmi7t}{}
\DeclareMathAlphabet{\mathpzc}{OT1}{pzc}{m}{it}

\title[Spectral properties of Pauli and Dirac operators  in dimension two]{Spectral properties and time decay of the wave functions of Pauli and Dirac operators in dimension two}

\author {Hynek Kova\v{r}\'{\i}k}

\address {Hynek Kova\v{r}\'{\i}k, DICATAM, Sezione di Matematica, Universit\`a degli studi di Brescia, Italy}

\email {hynek.kovarik@unibs.it}

\begin{document}

\begin{abstract}
We consider two-dimensional Pauli and Dirac operators with a polynomially vanishing magnetic field. The main results of the paper provide resolvent expansions of these operators 
in the vicinity of their thresholds. It is proved that the nature of these expansions is fully determined by the 
flux of the magnetic field.  The most important novelty of the proof is a comparison between the spatial asymptotics of the zero modes obtained in two different manners. 
The result of this matching allows to  compute explicitly all the singular terms in the associated resolvent expansions. 

As an application we show how the magnetic field influences the time decay of the associated wave-functions quantifying thereby the paramagnetic and diamagnetic effects of the spin.

\medskip


\medskip


\end{abstract}


\maketitle
{\hypersetup{linkcolor=black}
\setcounter{tocdepth}{1}
\tableofcontents}

\section{\bf Introduction and outline of the paper}
\label{sec-intro}
\subsection{The set up} Let $B:\R^2\to\R$ be a magnetic field and let $A:\R^2\to\R^2$  be the associated vector potential satisfying $\rt A=B$. We consider the  
Pauli operator 
\begin{equation} \label{pauli}
P(A)=  \begin{pmatrix}
P_\m(A)& 0 \\
 0  &  P_\pp(A)
\end{pmatrix}\,
\qquad  \text{with} \quad P_\ppm(A) : = (i\nabla +A)^2 \pm B,
\end{equation} 
and the Dirac operator
\begin{equation*} 
\qquad \  D_m(A) = \begin{pmatrix}
m& \D(A) \\
 \D(A)^*  &   -m
\end{pmatrix}  
\qquad \ \text{with} \quad
\D(A):= -i\pd_1 -A_1 -\pd_2 +iA_2 ,
\end{equation*}
where $m\geq 0$ is a constant.  The operators $P(A)$ and $D_m(A)$ act in $\Lp^2(\R^2;\C^2)$, and  are essentially  self-adjoint on $C_0^\infty(\R^2;\C^2)$ under mild regularity assumptions on $B$. In physics the Pauli and Dirac operators represent the non-relativistic and relativistic quantum mechanical Hamiltonians of a spin $1/2$ particle confined to a plane and interacting with the magnetic field $(0,0,B)$. The component $P_\m(A)$ in \eqref{pauli} denotes the restriction of the Pauli operator on the spin-up subspace, while $P_\pp(A)$ stands for the restriction on the spin-down subspace. We refer e.g.~to \cite{th2} for further background.

The main objects of our interest in the first part of the paper are the resolvents of $P(A)$ and $D_m(A)$ and their asymptotic expansions, in a suitable topology, near the respective thresholds of the essential spectrum. In the second part of the paper we study certain consequences of these expansions, such as local decay in time of the wave-functions.  We will primarily focus on the Pauli operator and subsequently apply the obtained results to the Dirac operator. 

It has been known since the landmark paper by Jensen and Kato  \cite{JK} that the asymptotic behavior of the weighted resolvent of a Schr\"odinger operator  depends on the spectral nature of the threshold, see also \cite{jen, jn,jn2, mu, wa}. Jensen and Kato showed that the structure of the expansion depends on whether the threshold is a {\em regular point}  (neither an eigenvalue nor a resonance), or an {\em exceptional point} (a resonance and/or an eigenvalue).  Accordingly, the resolvent expansion  is regular in the former case, and singular in the latter. It is also shown in \cite{JK} that generically the threshold of a Schr\"odinger operator is a regular point (in any dimension).  

The situation changes completely when a magnetic field is introduced and its interaction with spin is taken into account. For the threshold of the Pauli operator, which is zero, is {\em always an exceptional point} for $P(A)$. Indeed, the Aharonov-Casher theorem  \cite{ac, cfks,ev,rsh,shi} shows that if $B\in \Lp^1(\R^2)$, and if the (normalized) flux 
\begin{equation} \label{flux} 
\alpha = \frac{1}{2\pi} \int_{\R^2} B(x)\, dx
\end{equation}
satisfies $|\alpha|>1$, then zero is an eigenvalue of $P_\m(A)$ (for $\alpha>0$) or of $P_\pp(A)$ (for $\alpha <0$). Moreover, a closer inspection of the proof of the Aharonov-Casher theorem  reveals that zero is, for {\em any} $\alpha$, also  a resonance of $P(A)$.  Either of $P_\m(A)$ (when $\alpha>0$), or of $P_\pp(A)$ (when $\alpha <0$), or of both $P_\pp(A)$ and $P_\m(A)$ (when $\alpha=0$), see Lemma \ref{lem-ah-cash}. Note that zero is a resonance of $P_\ppm(A)$ if the equation $P_\ppm(A) u = 0$ admits a solution  $u\in \Lp^\infty(\R^2)\setminus \Lp^2(\R^2)$. In the sequel we call such solution a {\em zero resonant state}. 

All these facts are well-known and easily deducible from the Aharonov-Casher theorem and its proof. However, so far nothing is known about the resolvent behavior of Pauli  and magnetic Dirac operators at threshold, nor about the long time behavior of their wave-functions.

In this paper we prove an expansion of $(P(A)-\lambda)^{-1}$ as $\lambda\to 0$ (in a suitable topology). The latter reveals that the presence of a magnetic field causes new   phenomena with a continuous range of possible behaviors. In fact, 
 as soon as $\alpha\neq 0$, then, depending on the sign of $\alpha$,  one of the operators $(P_\ppm(A) -\lambda)^{-1}$ remains bounded whereas the other one is singular and its expansion contains negative powers of $\lambda$ which depend on $\alpha$, see e.g.~equation \eqref{r-exp-intro} below.  

In dynamical terms this means that when $\alpha>0$, then the wave-functions with initial data from the spin-down subspace of the Pauli operator decay faster than $t^{-1}$, while those with initial data from the spin-up subspace decay slower than $t^{-1},$ and vice versa for $\alpha<0$. Our main results then quantify these effects  in terms of $\alpha$, cf.~equation \eqref{dia-para}.

Analogous results are established also for the Dirac operator. 


\subsection{Main results} 
Let us describe our main results more in detail. We start by introducing some necessary notation. Since the Pauli operator is diagonal, it suffices to analyze the components $P_\ppm(A)$ separately.  Following \cite{JK} we will consider their resolvents $(P_\ppm(A)-\lambda)^{-1}$  as operators between  weighted Sobolev spaces $\HH^{k,s}$ equipped with the norm 
\begin{equation} \label{hms-norm}
\| u\|_{\HH^{k,s}} = \|\, \x^s (1- \Delta)^{k/2}\, u\|_{\Lp^2(\R^2)}.  \qquad k\in\Z, \ s\in\R,
\end{equation}
where $\x = (1+|x|^2)^{1/2}$. By 
$$
\B(k,s;k',s') = \B(\HH^{k,s}; \HH^{k',s'}) 
$$
we denote the space of bounded linear operators from $\HH^{k,s}$ into $\HH^{k',s'},$ and for $k=0$ we use the shorthands 
$$
 \| u\|_{\HH^{0,s}} = \| u\|_{2,s} \qquad  \Lp^{2,s}\,  = \HH^{0,s}.
$$
Throughout the paper we will work under the following conditions on the magnetic field:

\begin{assumption} \label{ass-B}
The function $B:\R^2\to \R$ is continuous and satisfies $|B| \, \lesssim\,  \langle\, \cdot\, \rangle^{-\rho}$  some some $\rho>2$.
\end{assumption}
Obviously, we then have $|\alpha| < \infty$.  It will be convenient to decompose $\alpha$ as follows;  
\begin{equation} \label{alpha}
\alpha = n + \alp,  \qquad \text{where} \quad n= {\rm sign}(\alpha)  \max\{ k \in\N: k < |\alpha| \}\, , \quad \text{and}\quad  \alp = \alpha-n.
\end{equation} 
With this notation zero is an eigenvalue of $P(A)$ of multiplicity $|n|$, cf.~Lemma \ref{lem-ah-cash}. Later we will built a basis of the zero eigenspace of $P(A)$ is such a way that exactly one eigenfunction  
does not belong to $\Lp^1(\R^2)$, see equations \eqref{ef-hat} and \eqref{ef-hat-plus}.   Another important quantity associated to the magnetic field is the function  
\begin{equation} \label{superp}
h(x) = -\frac{1}{2\pi} \int_{\R^2} B(y) \log |x-y|\, dy
\end{equation}
which satisfies $-\Delta h= B$.

\subsubsection*{\bf The Pauli operator} 

It is easily verified that any magnetic field satisfying Assumption \ref{ass-B} admits a vector potential $A\in \Lp^\infty(\R^2;\R^2)$. The components $P_\ppm(A)$ of the Pauli operator are then defined as unique self-adjoint operators generated by the closed quadratic forms
$$
\int_{\R^2} |(i\nabla +A) u|^2\, dx\,  \pm \int_{\R^2} B |u|^2\, dx , \qquad u\in W^{1,2}(\R^2)\, ,
$$
and satisfy  
\begin{equation*} 
\sigma(P_\ppm(A)) = \sigma_{\rm es}(P_\ppm(A))= [0,\infty). 
\end{equation*} 

Furthermore, it follows from \cite[Cor.~6.5]{ahk} that under Assumption \ref{ass-B} the operators $P_\ppm(A)$ have no positive eigenvalues. In view of \cite{jmp, mps} this in turn implies that the limits
\begin{equation*}  
R_\ppm(\lambda,A) := \lim_{\eps\to 0+} (P_\ppm(A) -\lambda -i\eps)^{-1} 
\end{equation*} 
in $\B(0,s;0,-s)$ exist and are finite for any $\lambda\neq 0$ and any $s>1/2$.

One of the main results of this paper states that if $0 <\alpha\not\in\Z$, and if $\rho$ is large enough, then $R_\pp(\lambda, A)=\mathcal{O}(1)$ as $\lambda\to 0$  in $\B(-1,s;1,-s), s>3,$ whereas
\begin{equation}  \label{r-exp-intro}
R_\m(\lambda, A) = -\lambda^{-1} \,  \pom  + \frac{ z_0\, \lambda^{\alp-1}}{1+ z_1 \, \lambda^{\alp}}\  \LL \,\cdot\, , \psim \RR\, \psim\, 
+ \frac{ z_2\,  \lambda^{-\alp}}{1 +z_3\, \lambda^{1-\alp}}\,  \LL \,\cdot\, , \varphi^\m \RR\, \varphi^\m\,  + \mathcal{O}(1). \\[3pt]
\end{equation}
Here $z_j\in\C$ are constants,  $\pom$  is an orthogonal projection on the zero eigenspace of $P_\m(A)$,  $\psim$ is a zero eigenfunction of $P_\m(A)$, and $\varphi^\m$ is a zero resonant state of $P_\m(A)$.  Notice that $\alp\in(0,1)$. Obviously, 
if $\alpha<1$, then $\pom=\psim=0$.  The explicit formulas for the functions $\psim$ and $\varphi^\m$ are too complicated to state in detail here. 
However, for radial $B$ these formulas simplify and we get
$$
\psim(x)  = \ {\rm const}\,   (x_1+ix_2)^{n-1} \, e^{h(x)}   \quad  \text{and} \quad \varphi^\m(x) ={\rm const}\,    (x_1+ix_2)^n \, e^{h(x)}  \hskip1.5cm [\, \text{if} \ B \ \text{is radial} \, ]\\[3pt]
$$
with $h$ defined in \eqref{superp}. From equation \eqref{h-asymp} we easily deduce that $\psim\in \Lp^2(\R^2)$ while $\varphi^\m\in\Lp^\infty(\R^2)\setminus \Lp^2(\R^2)$.

Let us comment on the second term on the right hand side of \eqref{r-exp-intro}. The asymptotic matching of zero modes (i.e.~zero eigenfunctions and resonant states) implies that only  slowly decaying  eigenfunctions contribute to this term, namely those which do not belong to $\Lp^1(\R^2)$. But the Aharonov-Casher theorem, Lemma \ref{lem-ah-cash}, shows that modulo $\Lp^1(\R^2)$ there exists {\em at most one} linearly independent eigenfunction which satisfies such condition, see e.g.~$\psim$ above. Hence the second term in \eqref{r-exp-intro}, if present, is of {\em rank one independently of} $n$. This is in contrast with resolvent expansions of non-magnetic Schr\"odinger operators where the rank of the corresponding term typically depends on the rank of the zero eigenspace,  \cite[Thm.~6.5 \& Rem.~6.6]{JK}. 

Another important difference with respect to non-magnetic Schr\"odinger operators is the presence of the denominators in \eqref{r-exp-intro}. This leads to higher order corrections of the leading terms $ \lambda^{\alp-1}$ and $ \lambda^{-\alp}$. 
Notice however, that in case of half-integer flux, when $\alp= 1/2$, these denominators do not contribute to the singular part of $R_\m(\lambda, A)$, see also Remark \ref{rem-rhs}.  Explicit form of equation \eqref{r-exp-intro} is stated in Theorems \ref{thm-pauli-res} and \ref{thm-pauli-res-m}.

The situation becomes more complicated when $\alpha\in\Z$, since the Pauli operator then has two zero resonant states. Nevertheless,  $R_\m(\lambda, A)$ resp.~$R_\pp(\lambda, A)$  satisfies an asymptotic equation similar to \eqref{r-exp-intro}  with fractional powers of $\lambda$ replaced by factors including $\log\lambda$. For more precise statements of the results, and for the expansions in the case $\alpha <0$ see Theorems \ref{thm-pauli-res-int}  and \ref{thm-pauli-int-m}. The special case $\alpha=0$ is treated in Corollary \ref{cor-pauli-res-0}.

\subsubsection*{\bf The Dirac operator} The operator $D_m(A)$ is, under Assumption \ref{ass-B}, self-adjoint on $W^{1,2}(\R^2)$, and its spectrum is given by 
$$
\sigma(D_m(A)) = \sigma_{\rm es}(D_m(A))= (-\infty, -m] \cup  [m,\infty). \\[4pt]
$$
Hence if $m>0$, then there are two thresholds of the essential spectrum; $m$ and $-m$. 
Since Dirac and Pauli operators are related through the identity
\begin{equation}  \label{pauli-dirac}
(D_m(A)-\lambda) (D_m(A) +\lambda) =  \begin{pmatrix}
P_\m(A) +m^2-\lambda^2 & 0\\
0 &    P_\pp(A) +m^2-\lambda^2 
\end{pmatrix}  ,
\end{equation}
the resolvent expansion of $D_m(A)$ at $\pm m$ can be derived from the resolvent expansion of $P(A)$ at zero.  In doing so it turns out that if $m>0$, then 
the  expansion of $(D_m(A)-\lambda)^{-1}$ for $\lambda\to \pm m$ is qualitatively the same as the expansion of $(P_\ppm(A)-\lambda)^{-1}$  for $\lambda\to 0$, see
Section \ref{ssec-dirac-time}.

However, an interesting effect occurs in the case of the massless Dirac operator; i.e.~for $m=0$. When $\alpha$ is small enough, then the resolvent expansion of $D_0(A)$ at zero is regular; 
\begin{equation*} 
\qquad \qquad  (D_0(A)-\lambda)^{-1} =  \mathscr{K} +  \mathcal{O}(\lambda^{1-2|\alpha|} )+ \mathcal{O}(\lambda^{|\alpha|} ), \qquad  \hskip1cm \text{if} \ \ |\alpha|\leq 1/2 \, ,
\end{equation*} 
as $\lambda\to 0$. See Theorem \ref{thm-dirac-2} for details. In dynamical terms this effect is reflected by faster decay of the wave-functions, see Corollary \ref{cor-dirac-time}.

\subsubsection*{\bf Applications} Resolvent expansions of Schr\"odinger operators obtained in \cite{JK} have been applied in various context such as low energy scattering theory, analysis of the wave operators or in the study of threshold resonances, see \cite{JK,jn3, jss, yaj} and references therein. Similar applications of equation \eqref{r-exp-intro} and its version for integer values of $\alpha$ can be obtained for Pauli and Dirac operators as well. This will be studied elsewhere. 

In this paper we discuss yet another, thought related, consequences of the resolvent expansions.
First, in Section \ref{sec-tdecay} we show how the magnetic field accelerates the local decay in time of the wave-functions of the Pauli operator in one of the spin subspaces, and slows down
the time decay in the other spin subspace (depending on the sign of $\alpha$). Suppose for example that $0<\alpha\not\in\Z$. If $u,v\in\Lp^{2,s}$ with $s>3,$ and if $u$ is 
orthogonal to all zero eigenfunctions of $P_\m(A)$, then 
\begin{equation} \label{dia-para} 
\| e^{-it P_\m(A)}\, u \|_{\Lp^{2,-s}}  \, \sim \, t^{-1+\alp} \qquad \text{while} \qquad \| e^{-it P_\pp(A)}\, v \|_{\Lp^{2,-s}} \, \sim \,  t^{-1-\min\{\alp, 1-\alp\}}  \\[2pt]
\end{equation}
as $ t\to\infty$. If $u$ is not orthogonal to the zero eigenspace of $P_\m(A)$, then of course $e^{-it P_\m(A)}\, u$ does not decay at all. 
Equation \eqref{dia-para}  thus quantifies the {\em paramagnetic effect} in the spin-up subspace, and the {\em diamagnetic effect} in the spin-down subspace (for $\alpha >0$). We refer to Theorems \ref{thm-pauli-time-1}, \ref{thm-pauli-time-2} for more complete statements, and to Theorem \ref{thm-dirac-time-1} and Corollary \ref{cor-dirac-time} for analogous results on the wave-functions of the Dirac operator. 

\smallskip

Our second application concerns discrete eigenvalues of the perturbed Pauli operator $P(A)  - \eps V$ where $V$ is an electric field vanishing at infinity. Asymptotic expansions of these eigenvalues for $\eps\to 0$ in the case of radial magnetic and electric field were established, with variational methods, in \cite{fmv}, see also \cite{bcez}. In Section \ref{sec-weak}, using Theorems \ref{thm-pauli-res} and  \ref{thm-pauli-res-int}, we  extend  some of these results to general $B$ and $V$, see Proposition \ref{prop-weak-1}.

\subsection{Key ingredients of the proof} 
\label{ssec-keys}
To explain the essential ideas of the method that we use to prove our main results, let us suppose for the moment that $\alpha>0$ and consider the resolvent of  $P_\m(A)$.

\subsubsection*{\bf The reference operator}  As usual, we make use of resolvent equation and write 
\begin{equation} \label{res-intro}
R_\m(\lambda, A)=  \big(1+(H_0-\lambda)^{-1} (P_\m(A)-H_0) \big)^{-1} (H_0-\lambda)^{-1}\, ,
\end{equation}
provided $1+(H_0-\lambda)^{-1} (P_\m(A)-H_0)$ is invertible. Here $H_0$ is a reference operator to be specified below. To use equation \eqref{res-intro} in weighted $\Lp^2-$spaces, we have to make sure that the coefficients of the difference $P_\m(A)-H_0$ vanish fast enough at infinity. For example, in the proof of \eqref{r-exp-intro} we need 
these coefficients to decay as $o(|x|^{-6})$. We will comment later on the origin of this restriction. This means, however,  that we cannot choose $H_0=-\Delta$ as it is usually done in the absence of magnetic field, cf.~\cite{jen,JK,jn,mu}. For then we would have $P_\m(A)-H_0=  2 i A\cdot \nabla  +i \nabla\cdot A + |A|^2-B$. But the Stokes theorem implies
that if $\alpha\neq 0$, then $A$ cannot decay faster than $\mathcal{O}(|x|^{-1})$.

Hence in view of the Stokes theorem  the only admissible choice of a reference operator is $H_0=(i\nabla+A_0)^2$, with $A_0$ generating a magnetic field of total flux $\alpha$, i.e.~the same as $B$.
We put
\begin{equation}  \label{a0}
A_0(x) = \alpha\ \frac{(-x_2\, ,\,  x_1)}{|x|^2} \ \min\{1, |x|\} .
\end{equation}
This vector potential $A_0$ was used already in \cite{ko} in the study of magnetic Schr\"odinger operators. If we now make a suitable choice of the gauge in the Pauli operator, see the remark below, then
the coefficients of 
\begin{equation} \label{difference} 
P_\m(A)-H_0 =  2 i (A-A_0)\cdot \nabla  +i \nabla\cdot (A-A_0) + |A|^2-|A_0|^2 - B
\end{equation} 
will decay as $\mathcal{O}(|x|^{-\rho+1})$, where $\rho$ is the decay rate of $B$ in Assumption \ref{ass-B}. The additional advantage of this choice of $H_0$  is that its resolvent, contrary to the resolvent of $-\Delta$, has a regular expansion at zero. The disadvantage, on the other hand, is that the explicit expression for $(H_0-\lambda)^{-1}$ is 
rather complicated, which leads to a series of elementary, but  sometimes quite tiresome calculations, see Section \ref{sec-resol-h0} and Appendices \ref{sec-app-b}, \ref{sec-app-c}. We also need a much more precise expansion of $(H_0-\lambda)^{-1}$  at zero than the one established in \cite{ko}, namely up to order $\mathcal{O}(\lambda^2)$. This is provided by Propositions \ref{prop-exp} and \ref{prop-exp-int}. Since expansions to higher order require higher values of $s$, in our case $s>3$,  this dictates, in view of \eqref{res-intro}, the $o(|x|^{-6})$ decay condition on the coefficients of $P_\m(A)-H_0$. 

Note also that, in agreement with the above discussion, the case $\alpha=0$ is the only  one where one can use $-\Delta$ as a reference operator, and which is not covered by our method. In fact, the resolvent expansion for $\alpha=0$  follows from results obtained in \cite{mu} and \cite{fmv}, see Section \ref{ssec-zero-flux} for details.

\subsubsection*{\bf The gauge} The canonical choice of the gauge related to the Pauli operator is provided by the vector potential
\begin{equation} \label{gauge-pauli}
A_{h} = \big(\partial_{2} h \, ,\,  -\partial_{1} h \big)\, .
\end{equation}
Indeed, we have $\rt A_h =-\Delta h= B$. However, it turns out that for a general $B$ the difference $|A_h(x)-A_0(x)|$ decays as  $\mathcal{O}(|x|^{-2})$ and not faster, see Section \ref{sec-pauli}. As explained above this is not sufficient for the decay requirements on the coefficients in \eqref{difference}. \newline  We therefore introduce a gauge transformation determined by a scalar field $\chi$ such that the {\em new gauge}  $\A= A_h +\nabla \chi$ remains sufficiently close to $A_0$ at infinity. This is done in Proposition \ref{prop-gauge}.
All the results of the paper will be formulated with this choice of the vector potential. Note that for radial magnetic field  we can put $\chi=0$, see Section \ref{ssec-radial}.

\subsubsection*{\bf Spatial asymptotic of the zero modes} Once the gauge with required properties is fixed, we may proceed with the analysis 
of  the zero modes of $P_\m(\A)$, i.e.~of the bounded solutions of the equation $P_\m(\A) u=0$. 
{\em The central idea of the proof is to compare the spatial asymptotics of these solutions 
obtained in the two different ways}. First, from the result on the behavior of solutions of $(H_0+W)u=0$ for general first order operator $W$, see Corollary \ref{cor-null}. Second, from the Aharonov-Casher theorem, cf.~Lemma \ref{lem-ah-cash}. Such a comparison, when combined with the expansion of $(H_0-\lambda)^{-1}$,  yields explicit expressions for the leading terms of the quantities
\begin{equation}  \label{uv-matrix}
\LL (P_\m(\A) -H_0) \, u, \, v +(H_0-\lambda)^{-1} (P_\m(\A)-H_0)\, v\RR \qquad \lambda\to 0,
\end{equation}
where $u,v$ are zero modes of $P_\m(\A)$. As we shall see shortly, these quantities play a fundamental role in finding the expansion of $R_\m(\lambda, A)$.

\subsubsection*{\bf Expanding equation  \eqref{res-intro}}  We have to show that the operator 
 $1+(H_0-\lambda)^{-1} (P_\m(\A)-H_0)$ in \eqref{res-intro}  is invertible for $\lambda$ small enough and to expand its inverse. To do so we apply the Grushin method of enlarged systems based on the analysis of the associated Schur complement, see e.g.~\cite{grushin,sz}. The latter reduces the problem, very roughly speaking, to finding the inverse of an $N\times N$ $\lambda-$dependent matrix. To this end use the Feshbach formula. These techniques, in various combinations, are frequently used in the analysis of resolvent expansions, e.g.~\cite{jn,jn2,wa}.
 
In our case we have $N=|n|+1$ if $\alpha\not\in\Z$, and 
 $N=|n|+2$ if $\alpha\in\Z$. At this point the knowledge of the spatial asymptotics of the zero modes becomes crucial. Indeed, using the expansion of \eqref{uv-matrix} we are able to calculate the inverse of the matrix in question  explicitly up to the relevant order $\mathcal{O}(1)$. This result in combination with  \eqref{res-intro} then allows us to compute all the singular terms in the expansion of $R_\m(\lambda, \A)$, see Theorem \ref{thm-pauli-res}. For $\alpha \in (0,1)$  we go one step further and calculate also the constant term. This is done in Theorem \ref{thm-pauli-res2}.

\subsection{Remarks} We conclude this introduction with several observations.
\label{ssec-com}

{\bf Perturbed Pauli operator.} 
If an electric field is added to $P(A)$, then the situation changes, of course. The threshold  becomes generically a regular point, in which case one can apply the results of \cite{ko} to both components $P_\m(A)$ and $P_\pp(A)$.  The method developed in the present paper can be used to obtain (probably less explicit) resolvent expansion  of $P(A)$ even if the threshold is an exceptional point.

{\bf Dimension three.}  There is no analogue of the Aharonov-Casher theorem in dimension three.  Nevertheless, it is known that there exist magnetic fields $B\in \Lp^{3/2}(\R^3;\R^3)$ for which zero is an eigenvalue of $P(A)$, see e.g.~\cite{be,bvb,elt, fl,ly}. However, as proved in \cite{be},  those magnetic fields for which zero is not an eigenvalue of $P(A)$ form a dense set in $\Lp^{3/2}(\R^3;\R^3)$.  One may thus expect that the resolvent expansion will be generically regular, and consequently that the wave-functions will decay, locally,  as $t^{-3/2}$. The same decay should be generically observed for the massive Dirac operator. 

On the other hand, the expected decay rate of the wave-functions of the three-dimensional massless Dirac operator, in the absence of zero modes, is $t^{-1}$. This was in fact proved in \cite{df} under certain smallness assumptions on $A$ and $B$. In particular, the assumptions of \cite{df}  imply that $\|B\|_{3/2}$ must be small enough, see \cite[Thm.~1.6]{df}. Such a restriction is compatible with the above cited results on the existence of zero modes. Indeed, it is explained in \cite{fl} that as soon as $\|B\|_{3/2} <  3(\pi/2)^{\frac 43}$, then the operators $P(A)$ and $D_0(A)$ in $\R^3$ cannot have any zero modes.

{\bf $\Lp^1-\Lp^\infty$ estimates.} It would be of course of great interest to extend the results of the present paper to the $\Lp^1-\Lp^\infty$ dispersive setting. 
Note that even for magnetic Schr\"odinger operators, i.e.~for spinless particles, these estimates are known so far only  in $\R^2$ and only for Aharonov-Bohm type magnetic fields, where one can exploit an explicit knowledge of the propagator, see \cite{f3p-1,f3p-2, fgk}.  This is in stark contrast with the $\Lp^1-\Lp^\infty$ dispersive estimates for non-magnetic Schr\"odinger operators which are  very well understood by now, see e.g.~\cite{eg,gsch,jss,RS,schlag, sch2}. Hence the results obtained here might understood as  the first step towards magnetic dispersive estimates in dimension two.

 \medskip

{\bf Convention:}  Since the treatment of the cases $\alpha>0$ and $\alpha<0$ is completely analogous (with the roles of $P_\m(\A)$ and $P_\pp(\A)$ interchanged), for the sake of definiteness we  focus primarily  on the case of positive flux.  We therefore give full details of the proofs of the main results only for $\alpha>0$. 
 In Section \ref{sec-alpha-neg} we explain how the results have to be adjusted when $\alpha < 0$. 
 
As already mentioned above, the case $\alpha=0$ is exceptional and requires a separate analysis, cf.~Section \ref{ssec-zero-flux} .

\section{\bf Preliminary results} 
\label{sec-prelim}

In this section we study  the behavior of solutions to 
the equation $(H_0 +W) u=0$,  where $W$ is a first order differential operator with rapidly vanishing coefficients.  In particular, we prove
 Corollary \ref{cor-null} which provides spatial asymptotic of these solutions. In the subsequent sections we will apply these results with the choice $W=P_\ppm(A)-H_0$. 
 We start by analyzing the inverse of $H_0$.

\subsection{The operators $G_0$ and $\G_0$}
\label{sec-h0}
\noindent  Recall that the reference operator $H_0$ is given by 
\begin{equation} \label{H0-def}
H_0 = (i \nabla +A_0)^2 \qquad \text{in} \quad \Lp^2(\R^2),
\end{equation}
with $A_0$ defined in \eqref{a0}. Since $|A_0|\in \Lp^\infty(\R^2)$, the form domain of $H_0$ coincides with the Sobolev space $W^{1,2}(\R^2)$. We use the shorthands
\begin{equation*}
 \HH^{\, k,s+0} = \bigcup_{s<r}  \HH^{\, k,r} , \qquad  \HH^{\, k,s-0} = \bigcap_{r<s}  \HH^{\, k,r}
\end{equation*}

 As mentioned in Section \ref{ssec-keys}, the resolvent of $H_0$ has a regular expansion at zero. We denote by
 $G_0$ respectively $\G_0$  the formal inverse of $H_0$ for $\alpha\not\in\Z$ respectively  $\alpha\in\Z$ . Using the partial wave decomposition it was shown in \cite[Sec.~5]{ko} that  in the case of non-integer flux we have 
\begin{equation} \label{G0}
\qquad \qquad   G_0(x,y) \, =  \sum_{m\in\Z} \, G_{m,0}(r,t) \ e^{im(\theta-\theta')} \, , \hskip2cm [\alpha\not\in\Z]\, ,
\end{equation} 
where $x = r e^{i\theta}$ and $y= t e^{i\theta'}$. 
The above formula is to be interpreted as follows; for a given $u\in H^{-1,s},\, s>1$ it holds
$$
( G_0\, u)(r,\theta) =  \sum_{m\in\Z} e^{im \theta}\, \int_0^\infty \int_0^{2\pi}  G_{m,0}(r,t) \,  e^{-im \theta'}\, u(t,\theta')\, t\, dt d\theta.
$$
In order to identify the kernels $G_{m,0}(r,t)$ we define  
\begin{align}
v_{m}(0, r) &=  e^{-|\alpha| r}\, (2 |\alpha|\, r)^{|m|}\, M\Big(\frac 12+|m|-m\, {\rm sign}(\alpha), 1+2 |m|, \, 2|\alpha|\, r\Big) \label{vm}\\
u_m(0, r) &=  e^{-|\alpha| r}\, (2 |\alpha|\, r)^{|m|}\, U\Big(\frac 12+|m|-m\, {\rm sign}(\alpha), 1+2 |m|, \, 2|\alpha|\, r\Big) \  \label{um},
\end{align}
with $M(a,b,z)$ and $U(a,b,z)$ being the Kummer's  confluent hypergeometric functions. Following the notation of \cite{ko} we write
\begin{equation} \label{abm}
a_m = v_{m}(0, 1), \qquad a'_m = v_{m}'(0, 1), \qquad b_m =  u_{m}(0, 1), \qquad b'_m =  u'_{m}(0, 1)\, ,
\end{equation}
and
\begin{align} \label{delta-m} 
\delta_m  & = \frac{a_m' -|m-\alpha| \, a_m}{a_m' +|m-\alpha|\,  a_m}, \qquad \beta_m  = \frac{b_m' +|m-\alpha|\,  b_m}{a_m' +|m-\alpha|\,  a_m}\,  , \qquad   \gamma_m = \frac{1}{a_m' +|m-\alpha|\,  a_m}\, .
 \end{align}
 
Moreover, we introduce the functions
\begin{align}
g_m(r) & = 
\left\{
\begin{array}{l@{\quad}l}
2\, |m-\alpha|\,  \gamma_m\, v_m(0,r) & \qquad r \leq 1  , \\
&\\
r^{|m-\alpha|} - \delta_m\, r^{-|m-\alpha|}  & \qquad 1  < r ,
\end{array}
\right.  \label{g_m} \\[4pt]
h_m(t) & = \frac{1}{4\pi |m-\alpha|} 
\left\{
\begin{array}{l@{\quad}l}
 \frac{\Gamma(\frac 12 +|m|-m)\,  }{\gamma_m\,  \Gamma(1+2 |m|)}\ \big(u_m(0,t)- \beta_m \, v_m(0,t)\big)  & t \leq 1  , \\
&\\
t^{-|m-\alpha|}  & 1  < t . \label{h_m}
\end{array}
\right. 
\end{align} 

From \cite[Eq.~(5.28)]{ko} we then deduce that\\
\begin{equation}\label{G02}
G_{m,0}(r,t)  
=  g_{m}(r) \, h_m(t) \quad \text{if} \ \  r <  t , \qquad \text{and} \qquad G_{m,0}(r,t)  
=  g_{m}(t) \, h_m(r) \quad \text{if} \ \  t \leq  r\, .
\end{equation}

When $\alpha\in\Z$, then one has to modify only the contribution to \eqref{G0}  which comes from the channel $m=\alpha$. 
For this value of $m$ the function $g_m$  has to be replaced by 
\begin{equation} \label{g-alpha}
 \g_{\alpha}(r) =
  \frac{1}{a'_\alpha}
\left\{
\begin{array}{l@{\quad}l}
\ v_{\alpha}(0,r) & r \leq 1  , \\[3pt]
a_\alpha +a'_\alpha\, \log r& 1  < r .
\end{array}
\right. 
\end{equation} 
Then
\begin{equation} \label{G0-a}
\G_0(x,y)=  \sum_{m\neq \alpha} G_{m,0}(r,t)\, e^{im(\theta-\theta')}  +\, \G_{\alpha,0} (r,t) \, e^{i \alpha(\theta-\theta')}  \hskip2cm [\alpha\in\Z]\, ,
\end{equation} 
where the kernel of $ \G_{\alpha,0}$ reads 
\begin{equation} \label{G-alpha}
\begin{aligned}
\G_{\alpha,0}(r,t) 
& = \frac{ \Gamma(1/2)\, v_{\alpha}(0,  r) }{2\pi\, \Gamma(1+2\alpha)}\ \big( u_{\alpha}(0, t)-
\frac{b'_\alpha}{ a'_\alpha}  \ v_{\alpha}(0,  t) \big) \qquad \text{if} \  r< t\leq 1,  \\
\G_{\alpha,0}(r,t) & = \frac{\g_{\alpha}(r)}{2\pi} \qquad \qquad  \qquad \qquad \qquad \qquad    \qquad \qquad \, \ \text{if} \ \ \max\{1,r\} <  t,  
\end{aligned}
\end{equation}
see \cite[Eq.~(6.6)]{ko}.

\begin{lem} \label{lem-g0-inv}
Let $f\in C_0^\infty(\R^2)$. Then $G_0 H_0 f =f$ respectively $\G_0 H_0 f =f$ holds pointwise on $\R^2$.
\end{lem}

\begin{proof}
We decompose $f$ into the Fourier series 
\begin{equation} \label{f-fourier} 
f(r,\theta) = \sum_{k\in\Z} \, e^{i k \theta}  f_k(r),
\end{equation} 
where
\begin{equation*}
f_k(r) = \frac{1}{2\pi}\,  \LL f, \, e^{ik(\cdot)}\RR_{\Lp^2(\s)} =  \frac{1}{2\pi} \int_0^{2\pi} e^{-ik\theta}\, f(r,\theta)\, d\theta.
\end{equation*}
Now let 
\begin{equation} \label{a0-polar}
a_0(r) = \alpha \quad \text{if}\quad r < 1, \qquad a_0(r) = \frac{\alpha}{r} \quad \text{if} \quad r\geq 1.
\end{equation}
Then, in polar coordinates we have $A_0(r,\theta) =  a_0(r)\, (-\sin\theta,\, \cos\theta)$, and an easy calculation shows that 
\begin{equation*} 
H_0 f (r,\theta) =  \sum_{k\in\Z} e^{i k \theta} \Big[ -f''_k(r)-\frac{1}{r} \, f_k'(r) +\Big(\frac kr-  a_0(r)\Big)^2\, f_k(r)\Big ]\, . 
\end{equation*} 
Let us first assume that $\alpha\not\in\Z$. Equation \eqref{G0} and integration by parts then give
\begin{align*}
G_0 H_0 f (r,\theta) &= 2\pi \sum_{m\in\Z} e^{i m\theta} \int_0^\infty \Big[ -t\, \pd_t^2 G_{m,0}(r,t)-\pd_t G_{m,0}(r,t) +t \Big(\frac m t-  a_0(t)\Big)^2\, G_{m,0}(r,t)\Big ]\, f_m(t) \,  dt \\
&\quad + 2\pi \sum_{m\in\Z} e^{i m\theta}\,  r \, \Delta_m(r)\, f_m(r),
\end{align*}
where
\begin{equation*} 
\Delta_m(r) := \lim_{t\to r-}  \pd_t G_{m,0}(r,t) - \lim_{t\to r+}  \pd_t G_{m,0}(r,t)\, .
\end{equation*}
From the definition of $G_{m,0}$ and from \cite[Eqs.~13.1.31-33]{as},  \cite[Eqs.~9.1.1]{as} we deduce that 
$$
 -t \, \pd_t^2 G_{m,0}(r,t)- \pd_t G_{m,0}(r,t) + \Big(\frac m t-  a_0(t)\Big)^2 t\,  G_{m,0}(r,t) = 0\qquad \forall\, r>0, \ \forall\, m\in\Z.
$$
Hence it suffices to show that 
\begin{equation} \label{sm-eq}
r\, \Delta_m(r) = \frac{1}{2\pi}  \qquad\forall\, r>0, \ \forall\, m\in\Z.
\end{equation}
If $r>1$, then \eqref{sm-eq} follows easily from \eqref{G02}. 
Suppose now that $r \in (0,1)$. Then, by \eqref{G02}, \eqref{vm} and \eqref{um},
\begin{align}
\Delta_m(r)  &=  \frac{\Gamma(\frac 12 +|m|-m)}{2\pi \Gamma(1+2 |m|)}\ \big(u_m(0,r)\, v'_m(0,r)- u'_m(0,r) \, v_m(0,t)\big) \label{sm-eq2} \\
& =  \frac{\Gamma(\frac 12 +|m|-m)\,  (2\alpha r)^{2|m|}}{2\pi e^{2\alpha r} \Gamma(1+2 |m|)}\, \mathcal{W}\!\left\{U\Big(\frac 12+|m|-m, 1+2 |m|, \, 2\alpha\, r\Big), M\Big(\frac 12+|m|-m, 1+2 |m|, \, 2\alpha\, r\Big)\right\}, \nonumber
\end{align}
where $\mathcal{W}\{ \cdot\, , \cdot\}$ denotes the Wronskian. To calculate the latter we use \cite[Eqs.~13.1.22]{as} which implies
\begin{equation} \label{wronski}
\mathcal{W}\Big\{U\Big(\frac 12+|m|-m, 1+2 |m|, \, 2\alpha\, r\Big), M\Big(\frac 12+|m|-m, 1+2 |m|, \, 2\alpha\, r\Big)\Big\} =   \frac{2\alpha\, e^{2\alpha r} \, \Gamma(1+2 |m|)}{(2\alpha r)^{2|m|+1}\Gamma(\frac 12+ |m|-m)},
\end{equation} 
and therefore also identity \eqref{sm-eq} for $0<r<1$. Finally, if $r=1$, then using \eqref{G02} and \eqref{delta-m} we verify that 
\begin{align*}
\Delta_m(1) = \frac{1}{2\pi}\big(\gamma_m\, v'_m(0,1) + |m-\alpha|\, \gamma_m\, v_m(0,1)\big) = \frac{1}{2\pi}\, .
\end{align*}
This proves \eqref{sm-eq} and the claim in the case of non-integer magnetic flux. The proof for integer magnetic flux follows the same lines; equation \eqref{G0-a} and the above calculations imply  that \eqref{sm-eq} holds also if $\alpha\in\Z$ provided $m\neq \alpha$. On the other hand, for $m=\alpha$ we deduce \eqref{sm-eq}  from \eqref{G-alpha} and \eqref{wronski}.
\end{proof}

The following results show how the functions from the image of $G_0$ behave at infinity. 

\begin{prop} \label{lem-null}
Let $0<\alpha\not\in\Z$. 
For any $f\in \Lp^{2, p}(\R^2)$ with $p>2$ one has $G_0 f \in \HH^{1,-1-0} \cap \Lp^\infty(\R^2)$. In particular, there exists $v\in \Lp^2(\R^2) \cap \Lp^\infty(\R^2)$ such that 
\begin{equation} \label{g0f}
G_0 f = u_0 +v, 
\end{equation}
where 
\begin{equation} \label{u0-0}
u_0(r,\theta)  = \frac{e^{i n\theta}\,  r^{-\alp} }{4\pi\alp} \ \big\langle f \, , \, g_n\,  e^{i n(\cdot)} \big\rangle + \, \frac{e^{i (n+1)\theta}\, r^{\alp-1}}{4\pi(1-\alp)} \  \big\langle f\, ,  \, g_{n+1}\,  e^{i (n+1)(\cdot)} \big\rangle
\end{equation}
for all $r>1$ and $\theta\in [0,2\pi]$.

Moreover, if $p>3$, then there exists $w\in \Lp^1(\R^2) \cap \Lp^\infty(\R^2)$ such that the function $v$ in \eqref{g0f} satisfies
\begin{equation} \label{v-asymp}
v(r,\theta) =  \frac{e^{i (n-1)\theta}\,  r^{-1-\alp} }{4\pi(1+\alp)} \ \big\langle f \, , \, g_{n-1}\,  e^{i (n-1)(\cdot)} \big\rangle+ \frac{e^{i (n+2)\theta}\,  r^{-2+\alp} }{4\pi(2-\alp)} \ \big\langle f \, , \, g_{n+2}\,  e^{i (n+2)(\cdot)} \big\rangle  + w(r,\theta)
\end{equation}
for all $r>1$ and $\theta\in [0,2\pi]$.
\end{prop}

\begin{proof}
We decompose $f$ into the Fourier series as in  \eqref{f-fourier}.  The Parseval identity gives
\begin{equation} \label{parseval-f}
\|f\|_{2,q}^2 = \sum_{m\in\Z} \int_0^\infty |f_m(t)|^2 (1+t^2)^{p} \, t\, dt\, ,
\end{equation}
By \cite[Eq.~(5.33)]{ko},
\begin{equation} \label{gm0-ub}
 | G_{m,0}(r,t)| \ \leq\  \frac{C_0}{|m-\alpha|}\, 
\left\{
\begin{array}{l@{\quad}l}
r^{|m|}\  t^{-|m|} &\quad r <t \leq 1  , \\
r^{|m-\alpha|}\  t^{-|m-\alpha|}  &  \quad 1 \leq r < t\\
r^{|m|}\  t^{-|m-\alpha|} 
 &\quad  r < 1  < t 
\end{array}
\right.
\end{equation}
holds  for all $0<r\leq t \leq 1$ and with a constant $C_0$ which depends on $\alpha$ but not on $m$. 
Hence $\langle\, f , g_n\, e^{in(\cdot)}\rangle$ and $\langle\, f, g_{n+1}\, e^{i(n+1)(\cdot)} \rangle$ are finite, see \eqref{g_m}.
We  split $G_0 f$ as
\begin{equation} \label{u123}
G_0 f =u_1+u_2 +u_3,
\end{equation}
where 
\begin{align*}
u_1(r,\theta)& = e^{i n\theta}\! \int_0^\infty G_{n,0}(r,t)\, f_n(t)\, t\, dt , \qquad  u_2(r,\theta)= e^{i(n+1)\theta}\! \int_0^\infty G_{n+1,0}(r,t)\, f_{n+1}(t)\, t\, dt ,
\end{align*}
and $u_3 = G_0 f-u_1-u_2$. Note that 
\begin{equation*}
\qquad u_3(r,\theta) = \sum_{m\neq n,n+1} e^{im \theta}\,  F_m(r)  \qquad \text{with} \quad
F_m(r) :=   \int_0^\infty G_{m,0}(r,t)\, f_m(t)\, t\, dt .
\end{equation*}
The Cauchy-Schwarz inequality thus gives
\begin{equation} \label{cs}
|F_m(r)|^2\  \leq \  z_m(r)  \int_0^\infty |f_m(t)|^2 (1+t^2)^{p} \, t\, dt ,
\end{equation}
where
$$
z_m(r) = \int_0^\infty G^2_{m,0}(r,t) (1+t^2)^{-p} \, t\, dt\, . 
$$
In view of \eqref{gm0-ub}  it follows that $z_m(r) \lesssim |m-\alpha|^{-2}$ for any $m\in\Z$ and any fixed $r$. This shows that 
\begin{equation} \label{g0f-infty}
u_j \in \Lp^\infty_{\rm\, loc}(\R^2), \qquad j=1,2,3.
\end{equation}
For $r >1$  we estimate $z_m(r)$ as follows 
\begin{align} \label{km-ub}
z_m(r) & =  \int_0^1 G^2_{m,0}(r,t) (1+t^2)^{-p} \, t\, dt\ +  \int_1^r G^2_{m,0}(r,t) (1+t^2)^{-p} \, t\, dt\ + \int_r^\infty G^2_{m,0}(r,t) (1+t^2)^{-p} \, t\, dt \nonumber\\
&\   \leq \frac{2\, C_0^2}{|m-\alpha|^2}\, (  r^{-2|m-\alpha|} + r^{-2p+2} )
\end{align}
where we have used \eqref{gm0-ub} again. 
Since $\inf_{m\neq n,n+1} |m-\alpha| >1$, we have
$$
\sup_{m\in \Z\setminus\{n,n+1\}} \int_0^\infty z_m(r)\, r\, dr < \infty \, ,
$$
and by  equations \eqref{cs} and \eqref{parseval-f},
\begin{equation} \label{u3-L2}
\|u_3\|_2^2  = \sum_{m\neq n,n+1} \int_0^\infty |F_m(r)|^2\, r\, dr \, \lesssim \,  \sum_{m\neq n,n+1} \int_0^\infty |f_m(t)|^2 (1+t^2)^{p} \, t\, dt\ \lesssim \ \|f\|_{2, p}^2. 
\end{equation}
By \eqref{g0f-infty} this implies that $u_3\in \Lp^2(\R^2)\cap\Lp^\infty(\R^2)$. Now consider the function $u_1(r,\theta)$. Using \eqref{G02} and \eqref{g_m} we get  
\begin{align*}
 u_1(r,\theta) & =  e^{in\theta} \ \frac{r^{-\alp}}{4\pi\alp} \int_0^\infty g_n(t) f_n(t) \, t\, dt   +e^{in\theta} \, v_1(r) , \qquad r>1,
\end{align*}
with 
$$
v_1(r) = \frac{r^{\alp}}{4\pi\alp} \int_r^\infty (t^{-\alp}-\delta_n\, r^{-2\alp} \, t^{-\alp})\, f_n(t)\, t\, dt  -  \frac{r^{-\alp}}{4\pi\alp} \int_r^\infty (t^{\alp}-\delta_n \, t^{-\alp})\, f_n(t)\, t\, dt\, .
$$
The Cauchy-Schwarz inequality gives 
\begin{align*}
|v_1(r) |^2 & \lesssim\  \Big( r^{-2\alp}\int_r^\infty t^{2\alp-2p+1}\, dt  +  r^{2\alp} \int_r^\infty t^{-2\alp-2p+1}\, dt \Big) \int_0^\infty |f_n(t)|^2 (1+t^2)^{p} \, t\, dt \, \lesssim\   r^{-2p +2} 
\end{align*}
for all $r>1$. In the same way we obtain 
$$
u_2(r,\theta) = e^{i (n+1)\theta} \, \frac{r^{\alp-1}}{4\pi(1-\alp)} \int_0^\infty g_{n+1}(t) f_{n+1}(t) \, t\, dt + \mathcal{O}(r^{-2p +2})  \qquad \forall\, r>1, 
$$
In view of \eqref{u3-L2} this proves \eqref{g0f} and \eqref{u0-0}.  To prove \eqref{v-asymp} we decompose  $u_3$ further as follows; 
\begin{equation} \label{u3-split}
u_3(r,\theta) =  e^{i(n-1)\theta}\! \int_0^\infty G_{n-1,0}(r,t)\, f_{n-1}(t)\, t\, dt  + e^{i(n+2)\theta}\! \int_0^\infty G_{n+2,0}(r,t)\, f_{n+2}(t)\, t\, dt  + u_4(r,\theta), 
\end{equation}
where
$$
u_4(r,\theta) =  \sum_{m\in \mathbb{I}_n} e^{im \theta}\,  F_m(r), \qquad  \mathbb{I}_n:= \Z\setminus\{n-1,n,n+1,n+2\}.
$$
Mimicking the above analysis of $u_1$ and $u_2$ shows that 
\begin{equation} \label{u3-split2}
\begin{aligned}
 \int_0^\infty G_{n-1,0}(r,t)\, f_{n-1}(t)\, t\, dt & =  \frac{  r^{-1-\alp} }{4\pi(1+\alp)} \ \big\langle f \, , \, g_{n-1}\,  e^{i (n-1)(\cdot)} \big\rangle + \mathcal{O}(r^{-2p +2})  \\
\int_0^\infty G_{n+2,0}(r,t)\, f_{n+2}(t)\, t\, dt  & =  \frac{  r^{-2+\alp} }{4\pi(2-\alp)} \ \big\langle f \, , \, g_{n+2}\,  e^{i (n+2)(\cdot)}  \big\rangle  +\mathcal{O}(r^{-2p +2})
\end{aligned}
\end{equation}
hold for all $r>1$ and $\theta\in[0,2\pi]$. As for $u_4$ we note that 
$$
\inf_{m\in\mathbb{I}_n} |m-\alpha| >2 .
$$
Hence from \eqref{km-ub} and \eqref{cs}  we deduce that 
\begin{align*}
\|u_4\|_{\Lp^1(\R^2)} & \lesssim \sum_{m\in \mathbb{I}_n} \int_0^\infty\! \sqrt{z_m(r)}\ r\, dr\,  \Big( \int_0^\infty |f_m(t)|^2 (1+t^2)^{p} \, t\, dt \Big)^{1/2} \ \lesssim \ \|f\|_{2,p}^2 +
\sum_{m\in \mathbb{I}_n} \big\|\sqrt{z_m}\, \big\|_{\Lp^1(\R^2)}^2 < \infty,
\end{align*}
since $p>3$. Equation \eqref{v-asymp} now follows from \eqref{u3-split2}.
\end{proof}

\begin{prop} \label{lem-nul-int}
Let $0\leq \alpha\in\Z$. 
For any $f\in \Lp^{2, p}(\R^2)$ with $p>2$ one has $G_0 f \in \HH^{1,-1-0} \cap \Lp^\infty(\R^2)$. Moreover,  
\begin{equation} \label{g0f-int}
G_0 f = \tilde u_0 +\tilde v,  
\end{equation}
where $\tilde v\in \Lp^2(\R^2) \cap \Lp^\infty(\R^2),$ and where 
\begin{equation} \label{u0-int}
\tilde u_0(r,\theta) = \frac{e^{i \alpha\theta} }{2\pi} \ \big\langle f \, , \, \g_\alpha\,  e^{i \alpha(\cdot)} \big\rangle\, + \, \frac{1}{4\pi\, r}\,  \Big[\,  e^{i (\alpha+1)\theta}\, \big\langle f\, ,  \, g_{\alpha+1}\,  e^{i (\alpha+1)(\cdot)} \big\rangle + e^{i (\alpha-1)\theta}  \big\langle f\, ,  \, g_{\alpha-1}\,  e^{i (\alpha-1)(\cdot)} \big\rangle \Big] 
\end{equation}
holds for all $r>1$ and $\theta\in [0,2\pi]$.

Moreover, if $p>3$, then there exists $\widetilde w\in \Lp^1(\R^2) \cap \Lp^\infty(\R^2)$ such that the function $\tilde v$ in \eqref{g0f-int} satisfies
\begin{equation} \label{tilde-v-asymp}
\tilde v(r,\theta) =  \frac{e^{i (\alpha-2)\theta}}{8\pi\,  r^2} \ \big\langle f \, , \, g_{\alpha-2}\,  e^{i (\alpha-2)(\cdot)} \big\rangle+ \frac{e^{i (\alpha+2)\theta}}{8\pi\,  r^2} \ \big\langle f \, , \, g_{\alpha+2}\,  e^{i (\alpha+2)(\cdot)} \big\rangle  + \widetilde w(r,\theta)
\end{equation}
for all $r>1$ and $\theta\in [0,2\pi]$.

\end{prop}

\begin{proof}
If $\alpha\in\Z$, then we split
$$
G_0 f =  u_+ +u_- +  u_\alpha+ \tilde u_3 ,
$$
where 
\begin{align*}
u_\pm(r,\theta)& = e^{i(\alpha\pm 1) \theta}\! \int_0^\infty G_{\alpha\pm 1,0}(r,t)\, f_{\alpha\pm 1}(t)\, t\, dt , \qquad  u_\alpha(r,\theta)= e^{i\alpha \theta}\! \int_0^\infty \G_{\alpha,0}(r,t)\, f_{\alpha}(t)\, t\, dt ,
\end{align*}
and
\begin{equation}\label{u3-int}
\tilde u_3(r,\theta) =  \sum_{|m-\alpha| >1} e^{im \theta}\! \int_0^\infty G_{m,0}(r,t)\, f_m(t)\, t\, dt \, .
\end{equation}
Here we use the decomposition \eqref{f-fourier}. 
From the  analysis of the case $\alpha\not\in\Z$ it follows that $\tilde u_3\in \Lp^2(\R^2)\cap\Lp^\infty(\R^2),$ and that 
$$
u_\pm(r,\theta) = \frac{ \min\{1,r^{-1}\}}{4\pi}\, e^{i (\alpha\pm 1)\theta}\, \big\langle f\, ,  \, g_{\alpha \pm 1}\,  e^{i (\alpha\pm 1)(\cdot)} \big\rangle +  \mathcal{O}(r^{-2p +2})  \qquad \forall\, r>1, 
$$
Similarly, using equation \eqref{G-alpha} we find
 $$
 u_\alpha(r,\theta) =  \frac{e^{i \alpha\theta} }{2\pi} \ \big\langle f \, , \, \g_\alpha\,  e^{i\alpha (\cdot)} \big\rangle\,  +  \mathcal{O}(r^{-2p +2})  \qquad \forall\, r>1, 
 $$
 This proves \eqref{u0-int}. 
 To prove \eqref{tilde-v-asymp} we split $\tilde u_3$, with a slight abuse of notation,  as 
 \begin{align}	\label{u3tilde-split}
\tilde u_3(r,\theta) &= \sum_\pm  e^{i(\alpha\pm 2) \theta}\! \int_0^\infty G_{\alpha\pm 2,0}(r,t)\, f_{\alpha\pm 2}(t)\, t\, dt  +   \sum_{|m-\alpha|>2} e^{im \theta}\,  F_m(r) .
 \end{align}
 In the same way as in the proof of Proposition \ref{lem-null} it follows that
 $$
  \sum_{|m-\alpha|>2} e^{im \theta}\,  F_m(r) \in \Lp^1(\R^2), 
 $$
 and that
 \begin{equation} \label{u3tilde-split2}
 \int_0^\infty G_{\alpha\pm 2,0}(r,t)\, f_{\alpha\pm 2}(t)\, t\, dt =  \frac{1}{8 \pi r^2} \ \big\langle f \, , \, g_{\alpha\pm 2}\,  e^{i (\alpha\pm 2)(\cdot)} \big\rangle + \mathcal{O}(r^{-2p +2})  \, .
\end{equation}
Hence the claim.
\end{proof}


\subsection{Zero modes of the operator $H_0+W$}  Our next aim is to use Propositions \ref{lem-null} and \ref{lem-nul-int} 
to analyze the asymptotic behavior of solutions to the equation 
\begin{equation*}
(H_0+W)\, u = 0 .
\end{equation*}
Here $W$ is a first order differential operator of the type 
\begin{equation}\label{eq-W} 
W = 2i Z\cdot  \nabla +i \nabla \cdot Z  + U
\end{equation}
with the real-valued coefficients $Z: \R^2\to \R^2$ and $U:\R^2\to \R$. We suppose that the coefficients of $W$ satisfy

\begin{assumption} \label{ass-W}
Let  $W$  be given by \eqref{eq-W}, and let $|Z|  + |U| \lesssim \x^{-\tau}$  for some some $\tau>1$. 
\end{assumption}

Note that if $W$ satisfies Assumption \ref{ass-W} for $\tau> 2$, then   
\begin{equation} \label{W-op}
W: \HH^{1,-s} \to \Lp^{2,\tau-s} \,   
\end{equation}

\begin{defin} 
Let $\mathscr{H}$ be a self-adjoint operator in $\Lp^2(\R^2)$. We define
\begin{equation} \label{def-null}
\NN(\mathscr{H}) = \left\{u\in \HH^{1,-1-0} \cap \Lp^\infty(\R^2): \  \mathscr{H} u=0 \right\}. 
\end{equation}
\end{defin}

We call  elements of $\NN(\mathscr{H})$  {\em zero modes} of $\mathscr{H}$.

For the proof of our next result we recall  a Hardy-type inequality for the operator $H_0$. Namely,
\begin{equation} \label{hardy}
\int_{\R^2} \frac{|u|^2}{1+|x|^2\, \log^2 |x|}\, dx \ \lesssim\ 
\| (i \nabla +A_0) u\| ^2 \, \qquad  u\in W^{1,2}(\R^2),
\end{equation}
see \cite{lw,timo}. The logarithmic factor on the left hand side is needed only if $\alpha\in\Z$, \cite[Lem.~6.1]{kov2}.

\smallskip

\begin{lem} \label{lem-uwu}
Let $W$ satisfy Assumption \ref{ass-W} with $\tau> 3$. Then for any $u\in \NN(H_0+W),\, u\neq 0,$  we have 
\begin{equation} \label{uwu} 
\LL u, - W u\RR > 0 .
\end{equation}
\end{lem}

\begin{proof} 
Assume that $u\in \NN(H_0+W)$. By \eqref{W-op} we have $W u\in \Lp^{2,2+0}(\R^2)$ and therefore the left hand side of \eqref{uwu} is well-defined. Since $\nabla \cdot A_0=0,$ it follows that
\begin{equation} \label{h0-uW}
H_0 u = -\Delta u  +2 i A_0 \cdot \nabla u + |A_0|^2\, u= - W u.
\end{equation}
Now define a cut-off function $\xi_n: \R^2\to\R, \, 2\leq n\in\N,$ by 
$$
\xi_n(x)= 1 \quad \text{if} \ \ |x| \leq 1, \qquad \xi_n(x)= \Big(1- \frac{\log |x|}{\log n}\Big)_+ \quad \text{otherwise}. 
$$
Then $f_n:= \xi_n u \in W^{1,2}(\R^2)$. Moreover, since $x\cdot A_0=0$, and since 
$$
\nabla \xi_n(x)= -\frac{x}{|x|^2 \log n}\ \qquad  \text{if } \quad 1\leq |x| \leq n, \qquad \nabla \xi_n(x) = 0  \qquad  \text{otherwise} ,
$$
we observe that 
\begin{align*}
\int_{\R^2} |(i \nabla +A_0) f_n|^2 \, dx & = \int_{\R^2} |\nabla(\xi_n u)| ^2\, dx +  i \int_{\R^2} \xi_n^2\,  \big(  \bar u\, A_0\cdot \nabla u  - u\, A_0\cdot \nabla \bar u + |A_0|^2\, |u|^2\big) \, dx. 
\end{align*}
An integration by parts then shows that, as as $n\to\infty$, 
\begin{align*}
 \int_{\R^2} |\nabla(\xi_n u)| ^2\, dx & = \int_{\R^2} |\nabla \xi_n|^2 |u|^2\, dx - \frac 12 \int_{\R^2} \xi_n^2\, \big(u\, \Delta \bar u+ \bar u\, \Delta u\big)\, dx  =   - \frac 12 \int_{\R^2} \xi_n^2\, \big(u\, \Delta \bar u+ \bar u\, \Delta u\big)\, dx + \mathcal{O}\Big(\frac{1}{\log n} \Big) ,
\end{align*}
The last two equations in combination with \eqref{h0-uW} yield
\begin{align*}
\int_{\R^2} |(i \nabla +A_0) f_n|^2 \, dx & =   -\frac 12 \int_{\R^2} \xi_n^2\, ( \bar u\, W  u   +u\, W\bar u) \, dx+  \mathcal{O}\Big(\frac{1}{\log n} \Big) .
\end{align*}
Thus by the Hardy inequality \eqref{hardy},
\begin{equation*} 
\int_{\R^2} \frac{  |u|^2\, \xi_n^2}{1+|x|^2 \log^2 |x|}\ dx \ \lesssim\ -\frac 12 \int_{\R^2} \xi_n^2\, ( \bar u\, W  u   +u\, W\bar u) \, dx+  \mathcal{O}\Big(\frac{1}{\log n} \Big).
\end{equation*}
Note that $\xi_n\leq 1,$ and $\xi_n \to 1$ pointwise in $\R^2$. Passing to the limit  $n\to \infty$ in the above inequality we thus get
$$
\int_{\R^2} \frac{  |u|^2}{1+|x|^2 \log^2 |x|}\ dx \ \lesssim\  \LL u, -W u\RR\, ,
$$
where we have used the dominated convergence. 
\end{proof}

\begin{lem} \label{lem-nh0}
$\NN(H_0) = \{0\}$.
\end{lem}

\begin{proof} 
Let $u\in \NN(H_0)$. An application of Lemma \ref{lem-uwu} with $W=0$ implies that $u=0$. 
\end{proof}

To be able to apply Propositions \ref{lem-null} and \ref{lem-nul-int} to the solutions of $(H_0 +W) u=0$, we have to determine in which sense the operator $G_0$ is  inverse to $H_0$. This is provided by the following lemma.

\begin{lem} \label{lem-2a}
We have 
\begin{enumerate}
\item $H_0 G_0\, v = v$ for any $v\in \Lp^{2, 2+0}(\R^2)$.  

\item $G_0 H_0\, u=  u$ for any $u\in \HH^{1,-1-0}\cap \Lp^\infty(\R^2)$ such that $H_0\, u \in \Lp^{2, 2+0}(\R^2)$.

\end{enumerate}
\end{lem}

\begin{proof}
Let $f\in C_0^\infty(\R^2)$ and let $v\in \Lp^{2, 2+0}(\R^2)$. Then $\LL f, H_0  G_0\, v \RR = \LL G_0 H_0\, f, v\RR= \LL f, v \RR$, see Lemma \ref{lem-g0-inv}. This proves $(1)$. To prove $(2)$ assume that 
$u\in \HH^{1,-1-0}\cap \Lp^\infty(\R^2)$ and $H_0\, u \in \Lp^{2, 2+0}(\R^2)$. Then by $(1)$ one has $H_0(G_0 H_0u -u)=0$. On the other hand,  Propositions \ref{lem-null},  \ref{lem-nul-int}   imply that $G_0 H_0 u \in \HH^{1,-1-0} \cap \Lp^\infty(\R^2)$, and therefore 
$(G_0 H_0u -u)\in \NN(H_0)$. Hence $G_0 H_0\, u=  u$ by Lemma \ref{lem-nh0}.
\end{proof}

With these preliminaries at hand we can state the main result of this section.

\begin{cor} \label{cor-null}
Let $W$ satisfy Assumption \ref{ass-W} with  $\tau> 4$. Then  any $u\in\NN(H_0+W), u\neq 0$ admits a decomposition
\begin{equation*} 
u = u_0 + u_1+\tilde u, \\[2pt]
\end{equation*}
where $\tilde u\in \Lp^1(\R^2) \cap \Lp^\infty(\R^2),$ and $u_0, u_1$ satisfy for all $r>1, \ \theta\in[0,2\pi]$ the following equations: 
\begin{align*} 
u_0(r,\theta)  & = -\frac{e^{in\theta}\, r^{-\alp} }{4\pi\alp} \ \big\langle Wu \, , g_n\,  e^{i n(\cdot)} \big\rangle\, - \, \frac{e^{i(n+1)\theta}\, r^{\alp-1}}{4\pi(1-\alp)} \  \big\langle Wu\, ,  g_{n+1}\,  e^{i (n+1)(\cdot)} \big\rangle, \\[2pt]
u_1(r,\theta)  & = - \frac{e^{i (n-1)\theta}\,  r^{-1-\alp} }{4\pi(1+\alp)} \ \big\langle Wu \, , \, g_{n-1}\,  e^{i (n-1)(\cdot)} \big\rangle - \frac{e^{i (n+2)\theta}\,  r^{-2+\alp} }{4\pi(2-\alp)} \ \big\langle W u\, , \, g_{n+2}\,  e^{i (n+2)(\cdot)} \big\rangle  
\end{align*}
if $0< \alpha\not\in\Z$, and
\begin{align*} 
u_0(r,\theta)  & = -\frac{e^{i\alpha\theta}}{2\pi} \ \big\langle Wu \, , \g_\alpha\,  e^{i\alpha (\cdot)} \big\rangle\, - \, 
 \frac{ r^{-1}}{4\pi}\,  \Big[\,  e^{i (\alpha+1)\theta}\, \big\langle Wu\, ,  \, g_{+}\,  e^{i (\alpha+1)(\cdot)} \big\rangle + e^{i (\alpha-1)\theta}  \big\langle Wu\, ,  \, g_{-}\,  e^{i (\alpha-1)(\cdot)} \big\rangle \Big] \\[4pt]
u_1(r,\theta)  & =  -\frac{e^{i (\alpha-2)\theta}}{8\pi\,  r^2} \ \big\langle Wu\, , \, g_{\alpha-2}\,  e^{i (\alpha-2)(\cdot)} \big\rangle - \frac{e^{i (\alpha+2)\theta}}{8\pi\,  r^2} \ \big\langle Wu \, , \, g_{\alpha+2}\,  e^{i (\alpha+2)(\cdot)} \big\rangle 
\end{align*}
if $0\leq \alpha\in\Z$.  Here we have used the shorthands $g_\pm = g_{\alpha\pm 1}$.
\end{cor}

Note that $u_1 \in \Lp^2(\R^2) \cap \Lp^\infty(\R^2)$, but $u_1\not\in \Lp^1(\R^2)$.

\begin{proof}
Let $u\in\NN(H_0+W)\setminus\{0\}$. 
By \eqref{W-op} we then have $W u \in \Lp^{2,3+0}(\R^2)$. Hence Lemma \ref{lem-2a} implies $u =- G_0 W u$. The claim now follows by setting $f = -Wu$ in Propositions \ref{lem-null} and \ref{lem-nul-int}. 
\end{proof}

\begin{rem} \label{rem-match}
By Proposition \ref{prop-gauge} below one can find $A$ such that the operator $W= P_\ppm(A)-H_0$ satisfies Assumption \ref{ass-W} with $\tau> 4$ provided $\rho >5$. The asymptotics obtained in Corollary \ref{cor-null} with this choice of $W$ then will be compared with the  expressions for the zero modes of $P_\ppm(A)$ obtained in Lemma \ref{lem-ah-cash}, see 
Section \ref{ssec-E-matrix}.
\end{rem}

\subsection{Eigenfunctions and resonant states} 
In this short subsection we analyze the the structure of zero modes of $H_0+W$ more accurately.
The spaces of zero eigenvalues and zero resonant states of $H_0+W$ are defined as follows; 
\begin{equation*} 
 \NN_e(H_0+W)  := \NN(H_0+W)\cap   \Lp^2(\R^2) \quad \text{and} \quad  \NN_r(H_0+W) := \NN(H_0+W)\setminus\NN_e(H_0+W) 
\end{equation*} 
respectively. Let 
$$
N_e: = {\rm dim} (\NN_e(H_0+W))  , \qquad  N_r := {\rm dim} (\NN_r(H_0+W)).
$$

\subsubsection*{The normalization} In view of \eqref{uwu} we can construct a basis of $\NN(H_0+W)$ in such a way that the following conditions be satisfied: 
\begin{align}
\LL \psi_k, - W \psi_\ell\RR & = \delta_{k, \ell} , \qquad \qquad\qquad\quad \qquad\ \,  k,\ell \in\{1,\dots,N_e\}  \label{norm1}  \\
\LL\phi_i, - W\phi_j\RR & =\delta_{i, j}, \quad \LL\phi_j,  W\psi_k\RR = 0 \, \qquad i,j\in \{1,...,N_r\}, \ \  k \in\{1,\dots,N_e\} \label{norm2}
\end{align}

Now let $s> 1,$ and let $Q: \HH^{1, -s} \to \NN(H_0+W)$ be the projection operator defined by 
\begin{equation} \label{Q}
Q u =  \sum_{j=1}^{N_r} \LL u, - W\phi_j \RR \, \phi_j \, +\, \sum_{k=1}^{N_e}\LL u, - W\psi_k\RR\,\psi_k \, .
\end{equation}
Put
$$
Q_0 = 1-Q. 
$$
If $\tau >2$, then in view of \eqref{W-op} 
\begin{equation*} 
1+G_0\, W : \HH^{1,-s}\,  \to \HH^{1,-s}\,  \qquad  \forall\, s: \   \tau/2  <s <\tau -1 . 
\end{equation*} 

\begin{lem} \label{lem-5}   Let $W$ satisfy Assumption \ref{ass-W} with $\tau > 2 $.  Assume that $\alpha\not\in\Z$.
Let $\tau/2  < s < \tau -1$. Then 
\smallskip

\begin{enumerate}

\item $  \NN(H_0+W) = \{u\in \HH^{1,-s}: \,  (1+G_0 W)u=0\}$. 
\medskip

\item One has 
\begin{equation} \label{Q-0}
Q(1+G_0 W) = (1+G_0 W) Q=0\quad  \text{on} \ \ \HH^{1,-s}\, .
\end{equation} 

\medskip

\item The operator $Q_0(1+G_0\, W) Q_0$ is invertible on $ \HH^{1,-s}$,  
and 
\begin{equation} \label{Omega-0}
\Omega_0 := (Q_0(1+G_0\, W) Q_0)^{-1}\, Q_0 
\end{equation} 
is bounded from $ \HH^{1,-s}$ to $ Q_0\HH^{1,-s}$.
\end{enumerate}
The operator $G_0$ may be replaced by $\G_0$ throughout when $\alpha\in\Z$.
\end{lem}

\begin{proof} 
Suppose that $\alpha\not\in\Z$.
Let $u\in   \HH^{1,-s}$ be such that $(1+G_0 W)u=0$. By \eqref{W-op} we have $u\in  \Lp^{2,2+0}$, and therefore Lemma \ref{lem-2a} implies $H_0 u =-W u$. Moreover, in view of Propositions \ref{lem-null}, \ref{lem-nul-int} we have $u = -G_0 Wu \in  \HH^{1,-1-0} \cap \Lp^\infty(\R^2)$. Hence $u\in \NN(H_0+W)$. Vice-versa, if $u\in  \NN(H_0+W)$, then $H_0 u = -Wu \in \Lp^{2,2+0}$, and Lemma \ref{lem-2a} gives $-G_0 W u= G_0 H_0u =u$. This proves $(1)$. 
 The identities in $(2)$ follow from part $(1)$ of the lemma and from the definition of $Q$.
 \smallskip 

\noindent To prove $(3)$ we first note that $G_0 W $ is compact on $\HH^{1,-s}$. The latter follows from the compactness of the operators $\x^{-1-0}\, G_0\, \x^{-1-0}$ and $\x^{-1-0}\, \nabla\, G_0\, \x^{-1-0}$ on $\Lp^2(\R^2)$, see \cite[Lemma~5.2]{ko}, and from the assumptions on $Z$ and $U$. Now let $f\in Q_0  \HH^{1,-s}$ be such that $(1+G_0\, W) f= u \in \NN(H_0+W)$. Then $(1+G_0\, W) u=0$, by $(1)$, and therefore $(1+WG_0 )Wu =W(1+G_0\, W)u=0$. This implies
$$
\LL u, W u\RR = \LL (1+G_0 W) f, W u\RR = \LL f, (1+W G_0) W u\RR = 0, 
$$
and by Lemma \ref{lem-uwu} we must have $u=0$, which in turn implies that $Q_0 f = f =0$. Hence the equation $Q_0(1+G_0\, W) f =0$ has no non-trivial solution in $Q_0 \HH^{1,-s}$. Consequently, the operator $Q_0(1+G_0\, W) Q_0$ is injective and surjective on in $Q_0 \HH^{1,-s}$, and therefore invertible with an inverse in $\B(Q_0 \HH^{1,-s})$. The proof for $\alpha\in\Z$ follows the same lines.
\end{proof}

\begin{prop}  \label{prop-Omega0}
Let $W$ satisfy Assumption \ref{ass-W} with $\rho > 2$. Assume that $\alpha\not\in\Z$. Then
\begin{align} 
(1+G_0 W)\, \Omega_0 & = \Omega_0 \, (1+G_0 W) = Q_0 \label{om-01} \\
W\,  \Omega_0 & = \Omega_0^* \, W\label{om-02}
\end{align} 
The operator $G_0$ may be replaced by $\G_0$ throughout when $\alpha\in\Z$.
\end{prop}

\begin{proof}
By Lemma \ref{lem-5} we have $Q \Omega_0 = \Omega_0 Q=0$. 
Equation \eqref{om-01} now follows from \eqref{Q-0}. To prove 	\eqref{om-02} we proceed as in \cite[Sec.~3]{JK}.  From the  definition of $Q$ we easily deduce that $W Q = Q^* W$. Hence in view of \eqref{om-01} 
$$
\Omega_0^* W =  \Omega_0^* W(1-Q) =  \Omega_0^* W (1+G_0 W) \Omega_0 =  \Omega_0^* (1+W G_0) W \Omega_0 = (1-Q^*) W \Omega_0 = W \Omega_0.
$$
\end{proof}


\section{\bf  Pauli operator; the gauge transformation and zero modes}  
\label{sec-pauli}
In order to apply the results of the  previous section to the Pauli operator, we have to ensure that the coefficients of  $P_\m(A) -H_0$ vanish fast enough at infinity. As explained in Section \ref{ssec-keys}, this is achieved by an appropriate choice of $A$. Let us first have a closer look at the vector potential $A_h$ introduced in \eqref{gauge-pauli}.

From \eqref{superp}  and \eqref{gauge-pauli} we deduce that as $|x|\to\infty$, 
\begin{equation} \label{h-asymp}
\qquad h(x) = -\alpha \log |x| + \frac{\vartheta_1\, x_1+\vartheta_2\, x_2}{ |x|^2}  + \mathcal{O}(|x|^{-2}),  \qquad  \text{where} \ \ \vartheta_j :=  \frac{1}{2\pi} \int_{\R^2} B(y)\,  y_j \, dy ,
\end{equation}
and
\begin{equation} \label{Ah-asymp}
A_{h}(x)   =  \alpha\, \frac{(-x_2, x_1) }{|x|^2}\ -
|x|^{-4} \big ( \vartheta_2 (x_1^2-x_2^2) -2\, \vartheta_1\, x_1 x_2\, , \,  \vartheta_1 (x_1^2-x_2^2) +2\, \vartheta_2\, x_1 x_2 \, \big) + \mathcal{O}(|x|^{-3}). \\[4pt]
\end{equation}
Hence unless $\vartheta_1=\vartheta_2=0$, the difference $A_h-A_0$ decays, for $|x|\to\infty$,  as
\begin{equation*}  
|A_h(x) - A_0(x)| = \mathcal{O}(|x|^{-2}) 
\end{equation*}
and {\em not faster}.  In Proposition \ref{prop-gauge} we therefore modify $A_h$ by adding a gradient of a suitable  scalar field.

\begin{rem}
Note that if $B$ is radial, then $\vartheta_1=\vartheta_2=0$. 
\end{rem}

\subsection{The gauge transformation} 

\begin{prop} \label{prop-gauge}
Let $B$ satisfy Assumption \ref{ass-B} with $\rho >4$, and let $ \int_{\R^2} B\, dx = 2\pi\alpha$. Then there exists a function $\chi\in C^2(\R^2;\R)$ such that 
\begin{equation} \label{vp}
 | \Delta \chi(x)|  \ \lesssim \  \x^{-\rho+1} ,  \qquad \text{and} \qquad 
|  A_h(x)  +\nabla\chi(x) - A_0(x) | \ \lesssim \  \x^{-\rho+1} .\\[2pt]
\end{equation}
Moreover, in the region $\{x\in\R^2:\, |x| >1\}$ the function $\chi$ is determined uniquely up to an additive constant, and can be chosen so as to satisfy 
\begin{equation} \label{phi-infty}
\chi(x) = \frac{\vartheta_2\, x_1 -\vartheta_1\, x_2}{ |x|^2}   +  \mathcal{O}(|x|^{-2}) \qquad \text{as}\ \  |x| \to \infty. 
\end{equation}
\end{prop}

\begin{proof} 
To prove the claim we introduce yet another gauge associated to the magnetic field $B$, namely the Poincar\'e gauge which gives, in polar coordinates, 
\begin{align}
  A_p(r,\theta) &=  \frac{(-\sin\theta, \cos\theta) }{r}\, \int_0^r B(s,\theta)\, s\, ds \label{hat-1} \, .
\end{align}
A short calculation then shows that $\rt A_p =B$. Using the notation 
\begin{equation} \label{psi}
z(\theta) := \int_0^\infty \! B(s,\theta)\, s\, ds
\end{equation}
we can thus decompose $A_p$ into two singular parts as follows;
\begin{equation} \label{ap-split}
 A_p  =A^{ab}_z+ \wth A_{p} ,
\end{equation}
where we have denoted
\begin{equation} \label{ap-hat} 
A^{ab}_z(r, \theta)  =z(\theta)\, \frac{(-\sin\theta, \cos\theta) }{r}  \qquad \text{and} \qquad  \wth A_p(r,\theta) = \frac{(\sin\theta, -\cos\theta) }{r}\  \int_r^\infty B(s,\theta)\, s\, ds.
\end{equation}
Similarly, we split 
\begin{equation} \label{ah-split}
 A_h  =A^{ab} + \wth A_h ,
\end{equation}
with
\begin{equation*}
A^{ab}(x)  = \alpha\, \frac{(-x_2, x_1) }{|x|^2}\, ,  \qquad \text{and} \qquad  \wth A_h(x) = \big(\partial_{2} (h(x)+\alpha\log |x|) \, ,\,  -\partial_{1} (h(x)+\alpha\log |x|)\big ) ,
\end{equation*}
see \eqref{gauge-pauli}. Notice that $A^{ab}$ and  $A^{ab}_z$ are Aharonov-Bohm type vector potentials which generate a singular magnetic field of total flux $2\pi\alpha$. The latter follows from the identity
\begin{equation} \label{psi2}
\int_0^{2\pi} z(\theta)\, d\theta = 2\pi \alpha .
\end{equation}
We claim that the vector field $\wth A_p- \wth A_h$ is conservative in the punctured plane $\R^2\!\setminus\!\{0\}$.  Indeed, let $\gamma\subset\R^2\!\setminus\!\{0\}$
 be a piece-wise regular closed curve. Note that $\rt \, \wth A_p =\rt\, \wth A_h=B$ in $\R^2\!\setminus\! \{0\}$.  Hence if $\gamma$ does not encircle the origin, then by the Stokes theorem
$$
\oint_\gamma 	\wth A_p  = \oint_\gamma \wth A_h  = \int_{E_\gamma}\!\! B(x)\, dx,
$$
where $E_\gamma$ is the region enclosed by $\gamma$.
On the other hand, if $\gamma$ does encircle the origin, then again by the Stokes theorem and equations \eqref{ap-split}, \eqref{ah-split} and \eqref{psi2}
$$
\oint_\gamma \wth A_p =  \oint_\gamma  A_p- \oint_\gamma  A^{ab}_z = \int_{E_\gamma} \!\! B(x)\, dx  -2\pi \alpha=   \oint_\gamma  A_h  - \oint_\gamma  A^{ab}= \oint_\gamma \wth A_h \, .
$$
So, in either case 
$
\oint_\gamma (\wth A_p-\wth A_h)  =0
$.
This shows that there exits a scalar field $\chi_0: \R^2\!\setminus\!\{0\}\to \R$,  determined uniquely up to an additive constant, such that 
\begin{equation} \label{eq-cons} 
\wth A_p = \wth A_h +\nabla \chi_0 \qquad \text{in} \quad \R^2\!\setminus\! \{0\}\, .
\end{equation}
Moreover, since  $\wth A_p, \wth A_h\in C^1(\R^2\!\setminus\!\{0\}; \R^2)$, by assumptions on $B$, we have $\chi_0 \in C^2(\R^2\!\setminus\!\{0\}; \R)$.
By \cite[Thm.1]{wi}, it is thus possible to find a function $\chi\in C^2(\R^2;\R)$ such that 
\begin{equation*} 
\chi(x) = \chi_0(x) \qquad \forall\, x:\  |x|>1.
\end{equation*}
Now, in view of \eqref{ap-hat} and assumptions on $B$,
$$
 | \nabla \cdot \wth A_p(x)|  \ \lesssim \  \x^{-\rho+1} \ , \qquad | \, \wth A_p(x)  | \ \lesssim \  \x^{-\rho+1} \, .  \\[3pt]
$$
Since $\nabla \cdot A_h =0$ and $A^{ab}(x) =A_0(x)$ if $|x|>1$, this in combination with \eqref{eq-cons} implies
\eqref{vp}.

To prove  \eqref{phi-infty} we note that in view of \eqref{Ah-asymp}, \eqref{eq-cons} and the above estimate on $\wth A_p$, 
$$
\nabla \chi (x) = \Big (\, \frac{\vartheta_2 (x_1^2-x_2^2) -2\, \vartheta_1\, x_1 x_2}{|x|^4}\, , \, \frac{\vartheta_1 (x_1^2-x_2^2) +2\, \vartheta_2\, x_1 x_2}{|x|^4}\, \Big) + \mathcal{O}(|x|^{-3}) \qquad \text{as}\ \  |x| \to \infty. 
$$
Integrating the above equation gives \eqref{phi-infty}. 
\end{proof}

\begin{rem} \label{rem-gauge} 
The proof shows that instead of setting $\chi=\chi_0$ for $|x|>1$ one could take  $\chi=\chi_0$ for $|x|> \eps$ with an arbitrary $\eps>0$. The function  $\chi$ would then be determined uniquely, up to an additive constant, outside the disc of radius $\eps$.
\end{rem}

 Once the choice of $\chi$ is made following Proposition \ref{prop-gauge}, we set
\begin{equation}  \label{gauge-transf}
\mathpzc{A} = A_{h}  +\nabla\chi .
\end{equation}
Accordingly, we denote 
\begin{equation} \label{w-pauli}
\vpm = P_\pm(\A) - H_0 = 2 i (\A-A_0)\cdot \nabla  +i \nabla\cdot \A + |\A|^2-|A_0|^2\pm B \, ,
\end{equation}
where we have taken into account the fact that $\nabla\cdot A_0=0$. 

A combination of Proposition \ref{prop-gauge} and the last two equations gives

\begin{cor} \label{cor-w0}
Let $B$ satisfy Assumption \ref{ass-B}. Then $\vpm$ satisfies Assumption \ref{ass-W} with $\tau = \rho-1$.
\end{cor}

\subsection{\bf Zero modes of $P_\pm(\A)$}  
\label{ssec-ah-cash} 
The following result is known  as the Aharonov-Casher theorem. It describes the zero eigenfunctions and resonant states of the Pauli operator. 

\begin{lem} \label{lem-ah-cash} 
Let $B$ satisfy Assumption \ref{ass-B}. The following assertions hold: 
\begin{itemize} 
\item[(i)] If $\alpha >0$, then $\NN(P_\pp(\A)) =\{0\}$, and any $ u\in \NN(\Pm(\A))$ is of the form
\begin{equation} \label{resonance}
\begin{aligned}
u & = \sum_{j=0}^{[\alpha]} \, c_j\,  (x_1+i x_2)^j\, e^{i\chi +h}   \, ,
\end{aligned}
\end{equation}
where $[\alpha] = \max\{ k \in\Z: k \leq \alpha \}$ is the integer part of $\alpha$, and $c_j\in\C $ are constants.

\item[(ii)] If $\alpha <0$, then $\NN(\Pm(\A)) =\{0\}$, and any $ u\in \NN(P_\pp(\A))$ is of the form
\begin{equation*}
u  =  \sum_{j=0}^{[-\alpha]}\,  c_j\,   (x_1-i x_2)^j\, e^{i\chi-h}  \, .
\end{equation*}

\item[(ii)] If $\alpha =0$, then 
\begin{equation*}
\qquad \NN(P_\ppm(\A)) = \big\{c\,  e^{i\chi\mp h}\, : \, c\in\C \big\}\, . \\[4pt]
\end{equation*}
\end{itemize}
\end{lem}

\begin{rem} 
The part of  Lemma  \ref{lem-ah-cash} concerning eigenfunctions of $P_\pm(\A)$ is very well known, see e.g.~\cite{ac, cfks, ev,rsh}. The part concerning the resonant states is less known, but follows easily from the arguments given in the above cited literature. For the sake of completeness we give a short proof which follows the lines of \cite[Thm.~6.5]{cfks}.
\end{rem}

\begin{proof}[Proof of Lemma  \ref{lem-ah-cash}]
Let us  consider the zero modes of $P_\pm(A_h)$.  A direct calculation shows that 
\begin{equation*} 
\LL e^{{\scriptscriptstyle\mp} h}\, v , \, P_\pm(A_h) \, e^{{\scriptscriptstyle\mp} h}\, v \RR = \int_{\R^2} e^{{\scriptscriptstyle\mp} 2h}\, \big | (\pd_1 \mp i \pd_2 ) v\big |^2\, dx
\end{equation*} 
Hence the zero modes of $P_\pm(A_h)$ have the form $u_\pm = e^{\mp h} v_\pm$, where $v_\pm$ is an entire function in the variable $z= x_1 \mp i x_2$. 

\smallskip

\noindent Assume first that $\alpha >0$. Since there is no entire function which vanishes at infinity in all directions as $|x| \to \infty$, it follows from \eqref{h-asymp} that $\NN(P_\pp(A_h)) =\{0\}$. On the other hand, \eqref{h-asymp} also implies  that for $u_ -=e^{h} v_-$ to be bounded, $v_-$ must be an entire function not increasing faster $|x|^{[\alpha]}$, and hence a polynomial of degree not larger than $[\alpha]$. This shows that 
$$
\NN(\Pm(A_h)) = \Big\{  \sum_{j=0}^{[\alpha]} \, c_j \,  (x_1+i x_2)^j \, e^h , \ c_j\in\C \Big\}.
$$
Since
\begin{equation} \label{a-to-ah}
P_\pm(A_h) = e^{-i \chi}\, P_\pm(\A)\, e^{i\chi}\, , \\[2pt]
\end{equation}
see \eqref{gauge-transf}, this proves (i). Assertions (ii) and (iii) follow in the same way.
\end{proof}

\subsection{Zero eigenfunctions of $\Pm(\A)$}
\label{ssec-zero-ef}
Lemma \ref{lem-ah-cash}  implies that $ \NN_e(\Pm(\A)) \neq \{0\}$ if and only if $\alpha>1$. Let us assume that this is the case and 
construct the basis of $ \NN_e(\Pm(\A))$ satisfying conditions \eqref{norm1}. This is done 
in the following way:  

\vskip0.2cm

\noindent We put $\ \psim_1 = d^\m_{1,1} \, e^{h+i \chi} \ $ with a constant $d^\m_{1,1}$ chosen such that  $\LL \psim_1, \, \vm\, \psim_1\RR =-1$. Notice that by  Lemmas \ref{lem-uwu} and \ref{lem-ah-cash} we have  $\LL e^{h+i \chi} , \vm\, e^{h+i \chi}\RR <0$. 
Next we put 
$$
\psim_2(x) = d^\m_{2,1} \, e^{h(x) + i \chi(x)}\, + d^\m_{2,2}\,  \, e^{h(x) + i \chi(x)}\, (x_1+i x_2) ,
$$ 
with $d^\m_{2,1}$ and $d^\m_{2,2}$ chosen so that $\LL \psim_2,  \vm\, \psim_2\RR =-1$ and $\LL \psim_1, \vm\, \psim_2\RR =0$. 
Continuing in this way we obtain a family of linearly independent zero eigenfunctions of $\Pm(\A)$ such that
\begin{equation} \label{ef-hat} 
\psim_j(x)=  \sum_{k=1}^j  d^\m_{j,k} \,  (x_1+i x_2)^{k-1}\, e^{h(x)+i\chi(x)}, \qquad j=1,\dots, n, \qquad \LL \psim_j,  \vm\, \psim_i\RR = -\delta_{i,j}\, .
\end{equation}

\subsection{ Radial magnetic fields} 
\label{ssec-radial}
If $B$ is radial,  i.e. $B(x) = b(|x|)$ for some $b:\R_+\to \R$, then the function $h$ defined by \eqref{superp} is radial as well. Hence, by the Newton's theorem, cf.~the proof of \cite[Thm.~9.7]{LL},
\begin{equation} \label{h-radial}
h(r) = -\log r \!\int_0^r b(t)\, t\, dt -\int_r^\infty b(t)\, t\, \log t\, dt\, .
\end{equation}

We then have

\begin{lem} \label{lem-radial}
Let $B$ satisfy Assumptions of Proposition \ref{prop-gauge}, and suppose moreover that $B$ is radial. Then $\A = A_p=A_h$. Moreover, the zero eigenfunctions of $\Pm(\A)=\Pm(A_h)$ satisfy
\begin{equation} \label{psij-radial}
 \psim_j = d^\m_j (x_1+i x_2)^{j-1}\, e^{h} \, , \qquad d^\m_j := d^\m_{j,j}, \quad j=1,\dots, n-1,
\end{equation}
where $d^\m_{j,k}$ are given in \eqref{ef-hat}.

\end{lem}

\begin{proof} 
By \eqref{hat-1} we have 
$$
A_p(x) = \frac{(-x_2, x_1)}{|x|^2} \, \int_0^{|x|}\, b(z) \, z\, dz =  \frac{ \alpha\, (-x_2, x_1)}{|x|^2} -  \frac{(-x_2, x_1)}{|x|^2} \, \int_{|x|}^\infty\, b(z) \, z\, dz
$$
From \eqref{h-radial} and \eqref{gauge-pauli} it thus follows that $A_h =A_p$. Moreover,  $\upsilon(\theta)=\alpha$, see \eqref{psi}. 
Hence in view of Proposition \ref{prop-gauge} we can take $\chi=0$, and therefore $\A=A_h$. To prove \eqref{psij-radial} we note that, in polar coordinates,
\begin{equation} \label{W-rad}
 \vm (r)= w(r)\,  \partial_\theta + v(r), 
\end{equation}
where the functions $w,v: \R_+\to \R$ depend only on $A_0$ and $B$. Since $h$ is radial, we conclude that 
\begin{equation*} 
\LL  \Phim_j ,\, \vm\, \Phim_k\RR\ = 0 \qquad \text{if} \quad j\neq k, 
\end{equation*}
see equation \eqref{phi-eq}. By construction of $\psim_j$ this implies \eqref{psij-radial}. 
\end{proof}


\section{\bf Resolvent expansion of the reference operator}
\label{sec-resol-h0}
In this section we state the asymptotic expansion of $R_0(\lambda)$, see Propositions \ref{prop-exp} and \ref{prop-exp-int}.  Recall that 
\begin{equation*}  
R_0(\lambda) = \lim_{\eps\to 0+} (H_0 -\lambda -i\eps)^{-1} \,
\end{equation*} 
where the limit is taken in the sense of norm in $\B(0,s;0,-s)$.
By the partial wave decomposition the integral kernel of $R_0(\lambda)$ can be written in the form 
\begin{equation} \label{r0-eq}
R_0(\lambda; x,y) =   \sum_{m\in\Z}\, R_{m,0}(\lambda; r,t )\, e^{im (\theta-\theta')}\, .
\end{equation}
We thus have to expand $R_{m,0}(\lambda)$, the integral operator in $\Lp^2(\R_+)$ with kernel $R_{m,0}(\lambda; r,t )$,  for each $m\in\Z$. The explicit expression for $R_{m,0}(\lambda; r,t )$ was obtained in \cite{ko}. Since we want to find all the singular terms in the expansion of $R_\m(\lambda,\A)$, 
we need to expand  $R_0(\lambda)$, and hence  $R_{m,0}(\lambda)$, up to a remainder term of order $\mathcal{O}(\lambda^2)$. 

 The goal of this section is to determine the coefficients in this expansion.  
Here we limit ourselves to the explicit calculation of the latter for non-integer flux and for $(r,t)\in (1,\infty)^2$. The computations in the remaining cases are postponed to Appendices \ref{sec-app-b} and \ref{sec-app-c}. As we shall see, the structure of these coefficients has important consequences for the calculation 
 of the matrix elements $\LL \vm u, (1+R_0(\lambda) \vm) \, v \RR$, with $u,v\in \NN(\Pm(\A))$, and therefore for deriving the expansion of $R_\m(\lambda,A)$. This will be done in the subsequent sections.

\subsection{Non integer flux}
\label{ssec-G-nint}
We show below that the coefficients in the expansion of $R_{0}(\lambda; r,t )$ corresponding to fractional powers of $\lambda$ less than two are finite rank operators with contributions only from the angular momenta close enough to $\alpha$, namely from those which satisfy $n-1\leq m\leq n+2$, cf.~equation \eqref{rm0-full}.
In order to simplify the notation we define the function $\zeta: \R_+\!\!\setminus\Z\to \C$ by
\begin{equation} \label{zeta} 
 \zeta(t) =  -\frac{4^{t-1}\, \Gamma(t)\, e^{i\pi t}}{\pi \Gamma(1-t)}\, \, ,\qquad  0<t\not\in\Z\, . \\[2pt]
\end{equation}
As indicated above, we assume that $1<r<t$. By \cite[Sec.5]{ko} we then have 
\begin{equation} \label{r-kernel}
R_{m,0}(\lambda; r,t ) = \frac{ \big ( A_m(\lambda)\,  J_{|m-\alpha|}(\sqrt{\lambda}\, r) +B_m(\lambda)\, Y_{|m-\alpha|}(\sqrt{\lambda}\, r)\big)\, \big ( J_{|m-\alpha|}(\sqrt{\lambda}\, t ) +i \,Y_{|m-\alpha|}(\sqrt{\lambda}\, t )\big)}{4 (B_m(\lambda)\ -i A_m(\lambda))}\, ,
\end{equation} 
where 
\begin{equation} \label{eq-ABm}
\begin{aligned} 
A_m(\lambda) & = (v_m(\lambda,1)\, |m-\alpha| -v'_m(\lambda,1))\, Y_{|m-\alpha|}(\sqrt{\lambda}) -\sqrt{\lambda}\ v_m(\lambda,1)\,  Y_{|m-\alpha|+1}(\sqrt{\lambda})  \\
B_m(\lambda) & = (v'_m(\lambda,1) -|m-\alpha| \, v_m(\lambda,1))\, J_{|m-\alpha|}(\sqrt{\lambda}) +\sqrt{\lambda}\ v_m(\lambda,1)\, J_{|m-\alpha|+1}(\sqrt{\lambda}) ,
\end{aligned}
\end{equation} 
and where $v_m(\lambda,r)$ is given by
\begin{align}
v_{m}(\lambda, r) &=  e^{-\sqrt{\alpha^2-\lambda}\, r} \big(2 r\sqrt{\alpha^2-\lambda}\ \big)^{|m|}\, M\Big(\frac 12+|m|-\frac{m\alpha}{ \sqrt{\alpha^2-\lambda}}, 1+2 |m|,\,  2r \sqrt{\alpha^2-\lambda} \, \Big) \label{vm-tot} .
\end{align}
Here we use a slightly different notation with respect to \cite{ko}; the coefficients $A_m(\lambda)$ and $B_m(\lambda)$ are rescaled by factor $\frac 2\pi$. Therefore some of the formulas from \cite{ko}; which we refer to below have been adjusted accordingly.  
 In order to expand $R_{m,0}(\lambda; r,t )$ into a sum of (integer and non-integer) powers of $\lambda$ it is convenient to introduce the function
\begin{equation}  \label{f-nu} 
f_\nu(z)  = \sum_{k=0}^\infty \, \frac{(-1)^k\, z^{k}}{4^k\, k!\ \Gamma(\nu+k+1)} \, .
\end{equation}
Then, by \cite[Sec.~9.1]{as}, 
\begin{equation} \label{eq-jy}
J_\nu(z) = (z/2)^\nu\, f_\nu(z^2), \qquad  Y_\nu(z) = \frac{J_\nu(z) \, \cos(\nu \pi) -J_{-\nu}(z)}{\sin(\nu\pi)}
\end{equation}
hold for all $ \nu\neq -1,-2,\dots$, where $J_\nu$ and $Y_\nu$ denote the 
Bessel functions.  With the notation 
\begin{align*}
\mathscr{P}_{m,1}(z) & =  2 v_m(z,1) f_{-|m-\alpha|-1}(z) + (v_m(z,1)\, |m-\alpha| -v'_m(z,1))\,  f_{-|m-\alpha|}(z) \\[2pt]
\mathscr{P}_{m,2}(z) & =  2 \big(v_m(z,1)\, |m-\alpha| -v'_m(z,1)\big)\,  f_{|m-\alpha|}(z)-  z v_m(z,1) \,  f_{|m-\alpha|+1}(z) 
\end{align*}
we can further rewrite the coefficients $A_m(\lambda)$ and $B_m(\lambda)$  as follows
\begin{equation}
\begin{aligned}  \label{am-1}
A_m(\lambda) & = -\sec(|m-\alpha|\pi)\,  \lambda^{-|m-\alpha|/2}\Big [\, 2^{|m-\alpha|} \, \mathscr{P}_{m,1}(\lambda) -2^{-1-|m-\alpha|} \,\cos(|m-\alpha| \pi) \, \lambda^{|m-\alpha|}   \ \mathscr{P}_{m,2}(\lambda) \Big] \\[2pt]
B_m(\lambda) & = - 2^{-1-|m-\alpha|}\, \mathscr{P}_{m,2}(\lambda)\, \lambda^{|m-\alpha|/2} \, .
\end{aligned} 
\end{equation}
By Lemma \ref{lem-pm12}, 
\begin{equation} \label{pm12-estim}
\frac{\mathscr{P}_{m,2}(\lambda)}{\mathscr{P}_{m,1}(\lambda)} = -2\sigma_m(\lambda) -2^{1+2|m-\alpha|}\, q_m\, \lambda +\mathcal{O}(\epsilon_m \, \lambda^2),  \qquad \epsilon_m = \frac{\Gamma(1-|m-\alpha|)}{\Gamma(1+|m-\alpha|)}\, ,
\end{equation}
where we have abbreviated
\begin{equation} \label{c-m-eq}
\sigma_m(\lambda) : = \frac{v'_m(\lambda,1)-v_m(\lambda,1)\, |m-\alpha| }{\mathscr{P}_{m,1}(\lambda)}\, f_{|m-\alpha|}(\lambda) \qquad \text{and} \qquad q_m = \frac{4^{-|m-\alpha|}\,
v_m(0,1)}{2 \Gamma(2+|m-\alpha|)\, \mathscr{P}_{m,1}(0)} \, .
\end{equation}
An elementary calculation now shows that as $\lambda\to 0$, 
\begin{align} \label{b-ia} 
B_m(\lambda) -iA_m(\lambda) & = \frac{i \, 2^{|m-\alpha|} \, \mathscr{P}_{m,1}(\lambda) }{\sin(|m-\alpha|\pi)}\, \lambda^{-\frac{|m-\alpha|}{2}}\, \Big[1+4^{-|m-\alpha|} \sigma_m(\lambda)\, e^{-i \pi |m-\alpha|}\, \lambda^{|m-\alpha|} +q_m e^{-i \pi |m-\alpha|}\, \lambda^{1+|m-\alpha|} \nonumber \\
& \qquad \qquad \qquad \qquad \qquad \qquad  \quad+\mathcal{O}\big(\epsilon_m\,\lambda^{2+|m-\alpha|}\big) \Big]\, , 
\end{align} 
and
\begin{equation} \label{abm-w} 
\begin{aligned}
\frac{A_m}{B_m -i A_m} (\lambda) & =  \frac{ i+  i \cos(\pi  |m-\alpha|)\big[  4^{-|m-\alpha|} \sigma_m(\lambda)\,  \lambda^{|m-\alpha|}+  q_m \,  \lambda^{|m-\alpha|+1} +
\mathcal{O}(\epsilon_m\, \lambda^{2+|m-\alpha|}) \big]}{1+ 4^{-|m-\alpha|} \sigma_m(\lambda) \, \lambda^{|m-\alpha|}\, e^{-i \pi |m-\alpha|} +q_m \lambda^{1+|m-\alpha|}\, e^{-i \pi |m-\alpha|}+\mathcal{O}(\epsilon_m\,\lambda^{2+|m-\alpha|}) } \\[8pt]
\frac{B_m}{B_m-i A_m}  (\lambda)& =  \frac{ -i \sin(\pi  |m-\alpha|)  \big[ 4^{-|m-\alpha|}  \sigma_m(\lambda)\,  \lambda^{|m-\alpha|}    +  q_m  \, \lambda^{|m-\alpha|+1}+\mathcal{O}(\epsilon_m\, \lambda^{2+|m-\alpha|}) \big]}{1+4^{-|m-\alpha|}\sigma_m(\lambda)   \lambda^{|m-\alpha|}\, e^{-i \pi |m-\alpha|} +q_m \lambda^{1+|m-\alpha|}\, e^{-i \pi |m-\alpha|} +\mathcal{O}(\epsilon_m\,\lambda^{2+|m-\alpha|}) }  
\end{aligned}
\end{equation}
With the above expressions at hand we return to equation \eqref{r-kernel} and expand its right hand side up to order $\mathcal{O}(\lambda^2)$. To do so, it will be convenient to denote
\begin{equation} \label{Q-nu}
Q_\nu(\lambda, r ) =  r^\nu f_\nu(\lambda\, r^2), \qquad \kappa_m = - 4^{-|m-\alpha|}\,  e^{-i \pi |m-\alpha|} \, , \qquad p_m(\lambda) = \kappa_m\, \sigma_m(\lambda)\, ,
\end{equation}
and
 \begin{equation} \label{gm-new}
g_m(\lambda,r) =   Q_{ |m-\alpha|}(\lambda, r) +\sigma_m(\lambda)\, Q_{ -|m-\alpha|}(\lambda, r)\, . \\[3pt]
\end{equation}
With this notation we get
\begin{align} \label{AJ-BY}
 \frac{A_m(\lambda)\,  J_{|m-\alpha|}(\sqrt{\lambda}\, r) +B_m(\lambda)\, Y_{|m-\alpha|}(\sqrt{\lambda}\, r)}{ B_m(\lambda)-i A_m(\lambda)} &=  \frac{i \lambda^{\frac{|m-\alpha|}{2}}}{4 \mathscr{V}_m(\lambda)}\, \Big( 2^{-|m-\alpha|}\, g_m(\lambda,r) +\lambda\, q_m\, Q_{ -|m-\alpha|}(\lambda, r)\Big) \nonumber \\
 & \qquad  + \mathcal{O}(\epsilon_m\, \lambda^{2+|m-\alpha|})\,  \big( J_{|m-\alpha|}(\sqrt{\lambda}\, r ) +Y_{|m-\alpha|}(\sqrt{\lambda}\, r )\big), 
\end{align}
where  
$$
\mathscr{V}_m(\lambda) = 1-p_m(\lambda) \, \lambda^{|m-\alpha|} +q_m \lambda^{1+|m-\alpha|}\, e^{-i \pi |m-\alpha|}\, . \\[4pt]
$$
Notice that the terms proportional to $\lambda^{1+\frac 32|m-\alpha|}$ and to $\lambda^{\frac 32 |m-\alpha|}$  canceled out in \eqref{AJ-BY}. 
By inserting \eqref{AJ-BY} into \eqref{r-kernel} we find that, as operators 
\begin{align} \label{rm0-new} 
R_{m,0}(\lambda) & = \frac{\sec(|m-\alpha|\, \pi)}{4 \mathscr{V}_m(\lambda)}\, \Big [ \Sigma_{m,0}(\lambda)+  \lambda^{|m-\alpha|} \Sigma_{m,1}(\lambda)+\lambda \Sigma_{m,2} (\lambda)
- \lambda^{1+ |m-\alpha|} \Sigma_{m,3}(\lambda) \Big]  
 + \mathcal{O}(\epsilon_m\, \lambda^{2+|m-\alpha|})\,  \widehat{R}_m(\lambda)  \, ,
\end{align}
where  $\Sigma_{m,j}(\lambda)$ are integral operators with kernels 
\begin{equation} \label{Sm-def}
\begin{aligned}
\Sigma_{m,0}(\lambda; r,t ) & :=  g_m(\lambda,r)\, Q_{- |m-\alpha|} (\lambda,t), 
\ \ \ \quad \Sigma_{m,2}(\lambda; r,t )  := q_m\, 4^{|m-\alpha|}\, Q_{- |m-\alpha|} (\lambda,r) \, Q_{- |m-\alpha|} (\lambda,t) \\[3pt]
\Sigma_{m,1}(\lambda; r,t )  & := \kappa_m\,  g_m(\lambda,r)\, Q_{ |m-\alpha|} (\lambda,t) , \quad 
\Sigma_{m,3}(\lambda; r,t )  := q_m \, e^{-i |m-\alpha| \pi} \, Q_{- |m-\alpha|} (\lambda,r)\,  Q_{|m-\alpha|} (\lambda,t) \, ,
\end{aligned}
\end{equation}
and where 
\begin{equation} \label{rm-hat}
 \widehat{R}_m(\lambda; r,t ) =   \big (J_{|m-\alpha|}(\sqrt{\lambda}\, r) +Y_{|m-\alpha|}(\sqrt{\lambda}\, r)\big)\, \big ( J_{|m-\alpha|}(\sqrt{\lambda}\, t ) +i \,Y_{|m-\alpha|}(\sqrt{\lambda}\, t )\big)\, . \\[4pt]
  \end{equation} 
The advantage of equation \eqref{rm0-new} with respect to  \eqref{r-kernel} is that the kernels of $\Sigma_{m,j}(\lambda)$ contain only integer powers of $\lambda$. Hence to single out the terms with fractional powers of $\lambda$ less than two it remains to expand the inverse of $\mathscr{V}_m(\lambda)$. We find
\begin{align}
 \frac{1}{\mathscr{V}_m(\lambda)}  \label{wm-exp}
 & = \frac{1}{1-p_m(\lambda) \, \lambda^{|m-\alpha|}} -\frac{q_m \, e^{-i |m-\alpha| \pi} \, \lambda^{1+|m-\alpha|} }{1-p_m(\lambda) \, \lambda^{|m-\alpha|}}
  -\frac{q_m \, e^{-i |m-\alpha| \pi} \, p_m(\lambda)\, \lambda^{1+2|m-\alpha|} }{(1-p_m(\lambda) \, \lambda^{|m-\alpha|})^2}+  \mathcal{O}\big (q_m^2\, \lambda^{2+2{|m-\alpha|}}\big) .
\end{align}
Now the crucial observation is that when we insert equation \eqref{wm-exp} into \eqref{rm0-new}, then the terms with fractional powers of $\lambda$ can be matched up  to form integral kernels given by multiples of 
$
g_m(\lambda,r)\, g_m(\lambda,t)$  or of  $g_m(\lambda,r) \, Q_{- |m-\alpha|} (\lambda,t) +g_m(\lambda,t) \, Q_{- |m-\alpha|} (\lambda,r).
$
Indeed, we have
\begin{align*}
\frac{\Sigma_{m,0}(\lambda; r,t )  +  \lambda^{|m-\alpha|}\, \Sigma_{m,1}(\lambda; r,t )}{1-p_m(\lambda)\,  \lambda^{|m-\alpha|}} & =   \Sigma_{m,0}(\lambda; r,t )  + \frac{\kappa_m\,   \lambda^{|m-\alpha|}\, g_m(\lambda,r) \, g_m(\lambda, t)}{1-p_m(\lambda)\,  \lambda^{|m-\alpha|}} \, ,
\end{align*}
and similarly,
\begin{align*}
& \frac{\lambda\, \Sigma_{m,2}(\lambda; r,t )}{ 1-p_m(\lambda) \, \lambda^{|m-\alpha|}}  - \frac{q_m \, e^{-i |m-\alpha| \pi}\,  \lambda^{1+|m-\alpha|} }{1-p_m(\lambda) \, \lambda^{|m-\alpha|}}\  \Sigma_{m,0}(\lambda; r,t ) =   \\[4pt]
& \qquad  \qquad\qquad  \ = \lambda\, \Sigma_{m,2}(\lambda; r,t ) - \frac{ q_m\, e^{-i\pi |m-\alpha|}\,    \lambda^{1+|m-\alpha|}}{1-p_m(\lambda)\,  \lambda^{|m-\alpha|}}\, \Big[  g_m(\lambda,r) \, Q_{- |m-\alpha|} (\lambda,t) +g_m(\lambda,t) \, Q_{- |m-\alpha|} (\lambda,r) \Big] \, .
\end{align*}
Summing up we obtain 
\begin{align} \label{rm0-full} 
R_{m,0}(\lambda) & =  \frac{ \Sigma_{m,0}(\lambda) + \lambda\,  \Sigma_{m,2}(\lambda) }{4 \sin(|m-\alpha|\, \pi)} +   \frac{\kappa_m\,   \lambda^{|m-\alpha|}\, g_m(\lambda,r) \, g_m(\lambda, t)}{4 \sin(|m-\alpha|\, \pi)\big(1-p_m(\lambda)\,  \lambda^{|m-\alpha|}\big)} \, \left[1- \frac{ q_m\, e^{-i\pi |m-\alpha|}\,  \lambda^{1+|m-\alpha|}}{1-p_m(\lambda)\,  \lambda^{|m-\alpha|} } \right ]
\\[4pt]
  & \quad - \frac{ q_m\, e^{-i\pi |m-\alpha|}\,    \lambda^{1+|m-\alpha|}}{4 \sin(|m-\alpha|\, \pi)\big(1-p_m(\lambda)\,  \lambda^{|m-\alpha|}\big)}\, \Big[  g_m(\lambda,r) \, Q_{- |m-\alpha|} (\lambda,t) +g_m(\lambda,t) \, Q_{- |m-\alpha|} (\lambda,r) \Big] + T_m(\lambda)  , \nonumber
\end{align}
where the integral kernel of the remainder term satisfies
\begin{align} \label{tm}
 T_m(\lambda;r,t)  & =  \mathcal{O}\big (q_m^2\, \lambda^{2+{|m-\alpha|}}\big)\, \sum_{j=0}^3\,  \Sigma_{m,j}(\lambda;r,t)  + \mathcal{O}(\epsilon_m\, \lambda^{2+|m-\alpha|})\,  \widehat{R}_m(\lambda;r,t)   \, .
\end{align}
By Lemma \ref{lem-Tm}, the operator $T_m(\lambda)$ is, in suitable topology,  of order $\mathcal{O}(\lambda^2)$. Note also that $\Sigma_{m,j}(\lambda)$ an d $p_m(\lambda)$ are power series in $\lambda$. Therefore the only contributions with fractional powers of $\lambda$  come from the second and third term on the right hand side in  \eqref{rm0-full}.  In order to calculate these contributions  we note that 
\begin{equation} \label{g-z}
\Gamma(1+z) = z \Gamma(z), \qquad \Gamma(1-z) \Gamma(z) = \frac{\pi}{\sin(\pi z)}\, ,
\end{equation}
which implies 
\begin{equation} \label{sigma-m}
\sigma_m(0) = -\delta_m\, \frac{\Gamma(1-|m-\alpha|)}{\Gamma(1+|m-\alpha|)}\  .
\end{equation}
Hence in view of \eqref{f-nu} and  \eqref{Q-nu}, 
\begin{equation} \label{g_m(0,r)}
g_m(0,r) = \frac{r^{|m-\alpha|} -\delta_m\, r^{-|m-\alpha|}}{\Gamma(1+|m-\alpha|)} =  \frac{ g_m(r)\, }{\Gamma(1+|m-\alpha|)} \ ,
\end{equation}
and consequently 
$$
  \frac{ \Sigma_{m,0}(0; r,t )  }{4 \sin(|m-\alpha|\, \pi)} = \frac{g_m(r)\, t^{- |m-\alpha|}}{4 \pi |m-\alpha|} = G_{m,0}(r,t)\, . \\[3pt]
$$ 
So for $1<r<t$ we can read the coefficients related to low order fractional powers of $\lambda$ in the expansion of $R_{m,0}(\lambda;r,t)$ directly from the right hand side of \eqref{rm0-full}. 
For example, from \eqref{zeta},  \eqref{g-z} and \eqref{g_m(0,r)} we deduce that the term proportional to $\lambda^\alp$ in \eqref{rm0-full} corresponds to $m=n$ and reads
$$
  \frac{\kappa_n\,   \lambda^{\alp}\, g_n(0,r) \, g_n(0, t)}{4 \sin(\alp \pi)\big(1-p_n(\lambda)\,  \lambda^\alp\big)} =  \frac{\lambda^{\alp}}{1-p_n(\lambda)\,  \lambda^{\alp} }  \ \frac{ g_n(r)\, g_n(t)}{(4\pi \alp)^2\, \zeta(\alp)}.  \\[2pt]
$$ 
In the same way we find the coefficients of all the other sub-quadratic powers of $\lambda$ in \eqref{rm0-full}.
It then remains to extend equation \eqref{rm0-full} to $r,\, t \leq 1$,  and control the remainder terms. This is done 
in  Appendix \ref{sec-app-b}. 
We  arrive at

\begin{prop} \label{prop-exp} 
Assume that $0<\alpha\not\in\Z$. Let $G_3$ be given by \eqref{g3-kernel}.
Then, as $\lambda\to 0$, 
\begin{align}
 \label{B0-eq-1}
R_0(\lambda) & = G_0 + \lambda^{\alp} G_1^\lambda + \lambda^{1-\alp} G_2^\lambda 
 + \lambda G_3+   \lambda^{1+\alp} G_4^\lambda +  \lambda^{1+2\alp}G_5^\lambda  +\lambda^{2-\alp} G_6^\lambda +\lambda^{3-2\alp} G_7^\lambda + \mathcal{O}(\lambda^2)\, ,
\end{align}
holds in $\B(-1,s;1,-s)$ for all $s>3$. Here $G_j^\lambda$ denote finite rank integral operators with kernels 
\begin{align}
G_1^\lambda(r,t;\theta, \theta') & =\frac{e^{i n (\theta-\theta')} }{1-p_n(\lambda)\,  \lambda^{\alp} }  \ \frac{ g_n(r)\, g_n(t)}{(4\pi \alp)^2\, \zeta(\alp)}   \label{G1} \\[6pt]
G_2^\lambda(r,t;\theta, \theta') & = \frac{e^{i (n+1) (\theta-\theta')}  }{1-p_{n+1}(\lambda)\,  \lambda^{1-\alp} }   \     \frac{g_{n+1}(r) g_{n+1}(t)}{16\pi^2(1- \alp)^2\, \zeta(1-\alp)} \label{G2} \\[6pt]
G_4^\lambda(r,t, \theta, \theta') &= \frac{ e^{i(n-1)(\theta-\theta')}\, g_{n-1}(r) \, g_{n-1}( t)}{16\pi^2(1+\alp)^2\, \zeta(1+\alp)}  - \frac{e^{-i\pi\alp} \, e^{in(\theta-\theta')}}{1-p_n(0)\,  \lambda^{\alp}}\, \big[  g_n(r) \, \varrho_n(t) +g_n(t) \, \varrho_n(r)\big]   \label{g4}  \\[5pt]
G_5^\lambda( r,t, \theta, \theta') &=  \frac{ q_n\, e^{-i\pi \alp}\, e^{in(\theta-\theta')}}{\big(1-p_n(0)\,  \lambda^{\alp}\big)^2}\  \frac{ g_n(r)\, g_n(t)}{(4\pi \alp)^2\, \zeta(\alp)}  \label{g5} 
\end{align} 
and
\begin{align}
G_6^\lambda( r,t, \theta, \theta') &=\frac{ e^{i\pi\alp}\, e^{i(n+1)(\theta-\theta')} }{1-p_{n+1}(0)\,  \lambda^{1-\alp}}\,\big [ g_{n+1}(r) \, \varrho_{n+1}(t) +g_{n+1}(t) \, \varrho_{n+1}(r) \big]- 
 \frac{ e^{i(n+2)(\theta-\theta')}\, g_{n+2}(r) \, g_{n+2}( t)}{16\pi^2(2-\alp)^2\, \zeta(2-\alp)} 
\label{g6} \\[6pt]
 G_7^\lambda( r,t, \theta, \theta') &=  \frac{q_{n+1}\, e^{i\pi\alp}\, e^{i(n+1)(\theta-\theta')}  }{\big(1-p_{n+1}(0)\,  \lambda^{1-\alp}\big)^2}\    \frac{g_{n+1}(r) g_{n+1}(t)}{16\pi^2(1- \alp)^2\, \zeta(1-\alp)}   \label{g7}\, . 
\end{align}
The function $\varrho_m$ is defined in  \eqref{rho-q}. Recall also that $q_m\in\R$ for all $m\in\Z$, cf.~\eqref{c-m-eq}.
\end{prop}
The proof of Proposition  \ref{prop-exp} is given in Appendix \ref{sec-app-b}.  

\begin{rem}
Given $\alp\in (0,1)$ we notice that not all of the terms on the right hand side of \eqref{B0-eq-1} are of order less than $\mathcal{O}(\lambda^2)$. Moreover, it will be shown in Section  \ref{ssec-E-matrix} that only the operators $G_1^\lambda, G_3$ and $G_4^\lambda$ contribute to the singular part of  $R_\m(\lambda,\A)$, see also Corollary \ref{cor-G12} below.
\end{rem}

\begin{rem} \label{rem-g3}
The operator $G_3$ is self-adjoint. To see this it suffices to observe that if  $\lambda= -z$ with $z>0$,  then all the other terms in \eqref{B0-eq-1} are self-adjoint operators multiplied by real numbers. The latter follows from the definitions of $\zeta$ and $p_m$. Since $R_0(-z)$ is self-adjoint, $G_3$ must be self-adjoint as well. 
\end{rem} 

\smallskip

\begin{cor} \label{cor-G12}
Assume that $|Z|  + |U| \lesssim \x^{-\tau}$ holds  for some $\tau> 4$. Then 
\begin{equation}  \label{G12W}
G_1^\lambda\, \vm \psi  = G^\lambda_2\,  \vm \psi =G^\lambda_5\,  \vm \psi = G^\lambda_7\,  \vm \psi =0 
\end{equation}
for all $\psi\in \NN_e(\Pm(\A))$.
\end{cor}

\begin{proof}
The claim follows from Corollary \ref{cor-null} and equations \eqref{G1}, \eqref{G2}, \eqref{g5}, and \eqref{g7}.
\end{proof}

\subsection{Integer flux} In  case of integer flux the expansion of $R_0(\lambda)$ contains no fractional powers of $\lambda$. Instead, there are logarithmic correction factors. The latter arise again from contributions relative to angular momenta $m$ close enough to $\alpha$, more precisely for $m\in \{\alpha-2, \dots, \alpha+2\}$. In order to simplify the notation we will use the shorthands 
\begin{equation} \label{short-pm}
g_{\pm} = g_{\alpha\pm 1}, \quad \delta_{\pm}= \delta_{\alpha \pm 1}, \quad \gamma_{\pm} = \gamma_{\alpha \pm 1} .
\end{equation}

\begin{prop} \label{prop-exp-int} 
Assume that $0<\alpha\in\Z$ and let $s>3$. Then 
\begin{equation} \label{B0-eq-2-int}
R_0(\lambda) = \G_0 + \Big(\frac{1}{ \log\lambda -z_\alpha}\, +\frac{E_4\, \lambda}{(\log\lambda -z_\alpha)^2}\, \Big) \, \G_1
 -  \frac{\lambda(\log\lambda-i\pi)}{64 \pi} \ \G^\lambda_2\,
 + \lambda\, \G_3+ \frac{ \lambda\, \G_4}{\log\lambda -z_\alpha}  -\frac{\lambda^2( \log\lambda-i\pi)}{64\pi} \ \G_5 
+ \mathcal{O}(\lambda^2)
\end{equation}
holds in $\B(-1,s;1,-s)$ as $\lambda\to 0$. Here $z_\alpha\in\C$ and $E_4\in \R$ are  given by equations \eqref{z-alpha} and \eqref{e-4}, the operator
$\G_3$ is defined in \eqref{g3-kernel-int}, and the remaining coefficients in \eqref{B0-eq-2-int} are finite rank operators with integral kernels 
\begin{align}
\G_1(r,t;\theta, \theta') & = \frac 1\pi\, \g_{\alpha}(r)\, \g_{\alpha}(t)\,  e^{i \alpha (\theta-\theta')}  \label{G1-int}  \\
\G_2^\lambda(r,t;\theta, \theta') & = g_{+}(r)\, g_{+}(t) \, e^{i (\alpha+1) (\theta-\theta')} \big(4+\delta_+ \lambda (\log\lambda-i\pi)\big)+ g_{-}(r)\, g_{-}(t) \, e^{i (\alpha-1) (\theta-\theta')}  \big(4+\delta_-\lambda( \log\lambda -i\pi)\big) \label{G2-int}  \\[3pt]
\G_4(r,t, \theta, \theta') &=  \frac 1\pi \ e^{i \alpha (\theta-\theta')} \,  \big [\g_\alpha (t)\, \mathfrak{j}_\alpha(r) +\g_\alpha (r)\, \mathfrak{j}_\alpha(t) \big]   \label{g4-int} \\[4pt]
\G_5(r,t, \theta, \theta') &=  e^{i (\alpha+1) (\theta-\theta')}   \big[ g_{+}(r)\, \mathfrak{j}_+(t)  + g_{+}(t)\, \mathfrak{j}_+(r)\big]  +
e^{i (\alpha-1) (\theta-\theta')} \big[ g_{-}(r)\, \mathfrak{j}_-(t)  + g_{-}(t)\, \mathfrak{j}_-(r)\big] \nonumber \\
& \quad +\frac 14 \Big[  g_{\alpha+2}(r)\, g_{\alpha+2}(t) \, e^{i (\alpha+2) (\theta-\theta')} +  g_{\alpha-2}(r)\, g_{\alpha-2}(t) \, e^{i (\alpha-2) (\theta-\theta')} \Big]
\label{g5-int} ,
\end{align}
where the functions $\mathfrak{j}_\pm$ and $\mathfrak{j}_\alpha$ are defined in equations \eqref{j-frak-pm} and \eqref{j-frak-alpha}.
\end{prop}

The proof of Proposition  \ref{prop-exp-int} is given in Appendix \ref{sec-app-c}.  

\begin{rem} \label{rem-hs}
The proofs of Propositions  \ref{prop-exp} and  \ref{prop-exp-int} given in Appendices \ref{sec-app-b} and \ref{sec-app-c} show that the operators $G_3, \G_3$ as well as the remainder terms in equations \eqref{B0-eq-1} and \eqref{B0-eq-2-int} are Hilbert-Schmidt as operators $\Lp^{2,s}(\R^2) \to \Lp^{2,-s}(\R^2)$ for $s>3$. Hence $R_0(\lambda)$ is also  Hilbert-Schmidt as operator $\Lp^{2,s}(\R^2) \to \Lp^{2,-s}(\R^2)$.
\end{rem}

\begin{rem} \label{rem-G3-int}
Recall that  $R_0(\lambda)$ is defined as a (weighted) limit of $R_0(\lambda +i\eps)$ for $\eps \to 0+$. Hence if $\lambda = -z$ with $z>0$, then $\log \lambda = \log z +i\pi$. Since $z_\alpha = i\pi +\R$, see equation \eqref{z-alpha}, the same reasoning as in Remark \ref{rem-g3} shows that 
also the operator $\G_3$ is self-adjoint.  
\end{rem}

\section{\bf Resolvent expansion of the Pauli operator; non-integer flux} 
\label{sec-exp-non-int}

Throughout this section we assume that $0<\alpha\not\in\Z$.  In order to simplify the notation  it is convenient to define the functions
\begin{equation} \label{phi-eq}
\Phim_j(x) = (x_1+i x_2)^{[\alpha]+1-j}\, e^{h (x)+i\chi(x)}\ , \qquad j=1,\dots , [\alpha]+1. 
\end{equation}
Note that in this case $[\alpha]=n$. By Lemma \ref{lem-ah-cash} we then have 
$$
{\rm dim} (\NN_r(\Pm(\A)))  =1, \qquad {\rm dim} (\NN_e(\Pm(\A)))  = n, \qquad {\rm dim} (\NN(\Pm(\A))) =n+1 =: N.
$$
First have to normalize the resonant state $\phi$ in such a way that conditions \eqref{norm2} be satisfied with $W= \vm$. In view of Lemma \ref{lem-ah-cash} this is achieved upon setting 
\begin{equation} \label{phi-2}
\phi(x) = c_0 \,   \Psim (x), \qquad \text{with} \quad \Psim = \Phim_1  + \sum_{j=1}^n\, \LL  \Phim_1\, ,\,  \vm\,\psim_j\RR\, \psim_j\, ,
\end{equation}
and choosing $c_0$ such that $ \LL  \phi\, ,\,  \vm\,  \phi\, \RR\ = -1$.

If $1+R_0(\lambda)  \vpm$ is invertible in $\B(1,-s;1,-s)$, then by the resolvent equation 
\begin{equation} \label{eq-res1}
R_\pm(\lambda,\A) = \big(1+R_0(\lambda)  \vpm\big)^{-1}\, R_0(\lambda).
\end{equation} 
Recall that in oder to expand $R_0(\lambda)$ to the needed precision we need to assume $s>3$, see Section \ref{sec-resol-h0}. 
Let us therefore assume that $ \vm$ satisfies Assumption \ref{ass-W} for some $\tau >6$. It follows from \eqref{W-op} that  
\begin{equation}  \label{op-W}
M(\lambda) := 1+R_0(\lambda)  \vm : \ \HH^{1,-s}\, \to \HH^{1,-s}\, , \qquad \tau/2 < s < \tau -3 .
\end{equation}

\subsection{\bf The Schur complement formula}  \label{ssec-grushin}
As already mentioned in the introduction, to comute an approximate expression for the inverse of $M(\lambda)$ we apply a version of the Grushin method of enlarged systems 
based on the  the so-called Schur complement formula, \cite{grushin, sz}.

Similarly as in \cite{wa} we  define the operators  $J: \C^N \to \HH^{1,-s}\, $ and $J^*: \HH^{1,-s}\, \to \C^N$ by 
\begin{align*}
J z &= z_0 \,\phi+ \sum_{k=1}^n z_ k\, \psim_k, \qquad z =(z_0,z_1,\dots,z_n)\in \C^N, \\
J^* u &=  \Big(\LL u, -  \vm\, \phi \RR ,\,   \LL u, -  \vm \, \psim_1\RR , \dots ,  \LL u, -  \vm\, \psim_n\RR \Big)\, .
\end{align*}
In view of equations \eqref{norm1} and \eqref{norm2} we then have
\begin{equation} \label{jj}
J J^* u  = Qu \qquad \forall\, u\in \HH^{1,-s}\, , \qquad \text{and}  \qquad J^*\!J z = z \qquad \forall\, z\in \C^N.
\end{equation}
Moreover, Proposition \ref{prop-gauge} and Lemma \ref{lem-5} imply that the operator $Q_0(1+R_0(\lambda)  \vm) Q_0$ is invertible on $Q_0\, \HH^{1,-s}\, $ for $ \tau/2 < s < \tau-3$ and $\lambda$ close enough to $0$.  Let 
\begin{equation} \label{D-lambda}
\Omega(\lambda) = \big[ Q_0\big(1+R_0(\lambda)\,  \vm\big) Q_0\big]^{-1} Q_0.
\end{equation}
By \eqref{B0-eq-1},
\begin{equation} \label{D-exp} 
\Omega(\lambda) = \Omega_0 + \mathcal{O}(|\lambda|^\mu) \qquad \lambda \to 0
\end{equation}
in $\B(1,-s;1,-s)$, where 
\begin{equation} \label{mu} 
\mu= \mu(\alpha) = \min\{\alp, 1-\alp\} = \min_{m\in\Z} |m-\alpha| \in [0,1/2] \\[4pt]
\end{equation}
is the distance between $\alpha$ and the set of integers, and $ \Omega_0$ is given by \eqref{Omega-0}. Now, let 
\begin {equation} \label{a-lam}
a(\lambda) = 
\left( \begin{array}{lr}
M(\lambda) &J \\[2pt]
J^* & 0
\end{array} \right) 
\end{equation}
be an operator matrix on $\HH^{1,-s}\,  \oplus\C^N$.  With the help of \eqref{jj} and \eqref{D-lambda}  one verifies that 
\begin {equation*}
a^{-1}(\lambda) = 
\left( \begin{array}{lr}
\Omega(\lambda) & a_{12}(\lambda) \\[3pt]
a_{21}(\lambda) & E(\lambda)
\end{array} \right) \\[3pt]
\end{equation*}
on $\HH^{1,-s}\,  \oplus\C^N$, where $a_{12}: \C^N\to \HH^{1,-s}\, , \ a_{21}: \HH^{1,-s}\, \to \C^N$ and $E :\C^N\to\C^N$ are given by  
\begin{align}
a_{12}(\lambda)& = J- \Omega(\lambda)Q_0 M(\lambda) J, \qquad a_{21}(\lambda)  = J^*-J^*M(\lambda)Q_0 \Omega(\lambda)\, , \label{a12}
\end{align}
and
\begin{align}
E(\lambda) &= -J^*M(\lambda) J +J^*M(\lambda)Q_0\Omega(\lambda)Q_0 M(\lambda)J\, . \label{E}
\end{align}
Th Schur complement formula now says that the operator $M(\lambda)$ is invertible in $\HH^{1,-s}\, $ if and only if the $N\times N$ matrix $E(\lambda)$ is invertible in $\C^N$, and 
in that case 
\begin{equation} \label{grushin}
M(\lambda)^{-1}  = \big(1+R_0(\lambda)  \vm\big)^{-1} =\Omega(\lambda) - a_{12}(\lambda)\, E^{-1}(\lambda)\, a_{21}(\lambda) .
\end{equation}


\subsection{\bf The matrix $E(\lambda)$} 
\label{ssec-E-matrix}
Next we have to show that $E(\lambda)$ is invertible for $\lambda$ small enough and to expand its inverse into several leading terms. To do so we 
make use of Corollaries  \ref{cor-null}, \ref{cor-w0} and Lemma \ref{lem-ah-cash}. When combined with Proposition \ref{prop-exp}, these results allow us to obtain 
important information on the matrix elements of $E(\lambda)$. This is done below in a series auxiliary lemmas.

\begin{lem} \label{lem-vg3v-0}
Let $B$ satisfy Assumption \ref{ass-B} with $\rho >7$, and let $u\in \Lp^{2,2+0}(\R^2)$ be such that $G_1 u=G_2 u=0$. Then
\begin{equation*}
\LL  \vm\, \psim_j, \, G_3\, u\RR \ =\  - \LL \psim_j, G_0\, u\RR\, , \qquad\forall\, j=1,\dots, n\, .
\end{equation*}
\end{lem} 

\begin{proof} 
To start with we note that the left side makes sense in view of Proposition \ref{prop-gauge} and the fact that $|(G_ 3 u) (x)| \lesssim \x^2$, see \eqref{gm3-1}. The right side also makes sense, because $G_0 u\in \Lp^2(\R^2)$ by Proposition \ref{lem-null}.  From Proposition \ref{prop-exp} we then deduce that, as $\lambda\to 0$, 
\begin{align*}
\LL \psim_j, u\RR &= \LL \psim_j, (H_0-\lambda) \, R_0(\lambda) u \RR = \LL H_0\, \psim_j, (G_0 +\lambda G_3) u\RR -\lambda \LL \psim_j, G_0 u\RR + o(\lambda) \\
& = \LL H_0\, \psim_j, \, G_0 u\RR - \lambda\, \LL  \vm \, \psim_j, G_3 u\RR -\lambda \LL \psim_j, G_0 u\RR + o(\lambda) \\
&= \LL \psim_j, u\RR - \lambda\, \LL  \vm \, \psim_j, G_3 u\RR -\lambda \LL \psim_j, G_0 u\RR + o(\lambda).
\end{align*}
Diving by $\lambda$ and letting $\lambda\to 0$ now implies the claim.
\end{proof}

\begin{lem}  \label{lem-jost}
Let $B$ satisfy Assumption \ref{ass-B} with some $\rho>7$. Let $1\leq i\leq 7,\, i\neq 3$.  Then 
\begin{equation} \label{i4}
\LL  \vm \,\psim_k, G_i^\lambda\   \vm\, \psim_j \RR  \neq 0 \qquad \text{\rm if and only if} \qquad i=4 \ \wedge j=k=n\, .
\end{equation}
Furthermore, 
\begin{equation} \label{pn}
\LL  \vm \,\psim_n, G_4^\lambda\,   \vm\, \psim_n \RR  = \frac{ |d^\m_{n}|^2\,}{\zeta(1+\alp)}  \, ,
\end{equation}
where $d^\m_n:=d^\m_{n,n}\in\C$ is given by \eqref{ef-hat}.
\end{lem}

\begin{proof}
Corollary \ref{cor-G12} implies  
\begin{equation}  \label{gwi}
G_i^\lambda\,   \vm\, \psim_j  =0 \qquad 	 i = 1,2,5,7 \qquad  1\leq j\leq n.
\end{equation} 
Moreover, since $\psim_j\in \Lp^1(\R^2)$ for all $j\leq n-1$, a combination of Corollary \ref{cor-null} with equations \eqref{ef-hat} and \eqref{h-asymp} gives  
\begin{equation} \label{w0gm}
\begin{aligned}
\LL  \vm \,\psim_j,  \, g_m \, e^{im\theta} \RR & = 0 \qquad 1 \leq j\leq n-1, \qquad   n-1 \leq m \leq n+2  \\
 \LL  \vm \,\psim_n,  \, g_m \, e^{im\theta} \RR & = 0   \qquad   n \leq m \leq n+2  \\
 \LL  \vm \,\psim_n,  \, g_{n-1} \, e^{i (n-1)\theta} \RR & = -4\pi (1+\alp) \, d^\m_{n}\, .
 \end{aligned}
 \end{equation}

On the other hand, using \eqref{g4}, \eqref{g6} and \eqref{w0gm}  we find that there exist constants $\mathcal{C}_k(n,j,\lambda)$ such that
\begin{equation} \label{gw46}
\begin{aligned}
G_6^\lambda\,   \vm\, \psim_j  & = \mathcal{C}_6(n,j,\lambda) \, e^{i(n+1)\theta}\, g_{n+1}(r) \qquad 1\leq j\leq n  \\
G_4^\lambda\,   \vm\, \psim_j  & = \mathcal{C}_4(n,j,\lambda) \, e^{in\theta}\, g_{n}(r)  \ \ \qquad  \qquad  1\leq j\leq n-1 
\end{aligned}
\end{equation} 
By \eqref{w0gm} this completes the proof of \eqref{i4}. Finally, for $j=n$ and $i=4$ we obtain, by \eqref{w0gm} and \eqref{g4}, 
\begin{align*}
G_4^\lambda\,   \vm\, \psim_n  &  = \mathcal{C}_4(n,n,\lambda) \, e^{i n\theta}\, g_{n}(r) + \frac{e^{i(n-1)\theta}\, g_{n-1}(r) }{16\pi^2 (1+\alp)^2\, \zeta(1+\alp)} \, \LL   \vm\, \psim_n , \, e^{i(n-1)\theta}\, g_{n-1}(r) \RR \   \\[5pt]
&= \mathcal{C}_4(n,n,\lambda) \, e^{i n\theta}\, g_{n}(r) -  \frac{e^{i(n-1)\theta}\, g_{n-1}(r) \, d^\m_{n}}{4\pi (1+\alp)\, \zeta(1+\alp)} \, .
\end{align*}
Applying \eqref{w0gm} once again yields \eqref{pn}.
\end{proof}

\begin{lem} \label{lem-ef-mat}
Let $B$ satisfy assumption \ref{ass-B} with $\rho>7$. Then 
\begin{align*} 
\LL  \vm \,\psim_n, \, M(\lambda)\,\psim_n\RR &\  =\  \lambda\, \LL\psim_n ,  \psim_n \RR  + \lambda^{1+\alp}\,  \frac{ |d^\m_{n}|^2}{\zeta(1+\alp)} +  \mathcal{O}(\lambda^2),  \\
\LL  \vm \,\psim_k, \, M(\lambda)\,\psim_j\RR &\  =\  \lambda\, \LL\psim_k ,  \psim_{j} \RR  + \mathcal{O}(\lambda^2), \qquad\qquad   (j,k) \neq (n,n) \, .
\end{align*}
\end{lem}

\begin{proof}
This is a combination of Lemma \ref{lem-5}, Lemma \ref{lem-jost} and Proposition \ref{prop-exp}.
\end{proof}

\begin{lem} \label{lem-vg3v}
Let $B$ satisfy assumption \ref{ass-B} with $\rho>7$. Then for all $j,k=1,\dots, n$,
\begin{align} 
\LL  \vm\,\phi , \, M(\lambda)\, \phi\,  \RR & \ = \ \zeta^{-1}(\alp)\,  |c_0|^2\, \lambda^\alp+|c_0|^2\, \eta_\m\,  \lambda  +   \mathcal{O}(\lambda^{2\alp}), \label{phi-phi} \\
\LL  \vm\, \phi , \, M(\lambda)\,\psim_j\RR &\ =\  \lambda\, \LL  \phi, \,    \vm\, G_3  \vm\, \psim_{j} \RR  + \mathcal{O}(\lambda^{1+\alp}),   \label{M-phi}
\end{align}
where 
\begin{equation} \label{xi-beta}
 \eta_\m  := \LL \Phim_1, \,  \vm\, G_3  \vm\, \Phim_1\, \RR\, \in \R.
\end{equation}
\end{lem}

\begin{proof} By \eqref{ef-hat}, \eqref{phi-2}, \eqref{h-asymp} and \eqref{phi-infty} it follows that as $|x|\to\infty$, 
\begin{equation} \label{fi-asymp}
\phi(x)  =  c_0 \,  r^{-\alp} e^{i n\theta} + c_0\, r^{-1-\alp} e^{i (n-1)\theta}\, \big(\vartheta_1+i\vartheta_2 +\LL  \Phim_1\, ,\,  \vm\, \psim_n\RR\,\big) 
 +\mathcal{O}\big(r^{-2-\alp}\big)\, .
\end{equation} 
A comparison with Corollary \ref{cor-null} thus gives 
\begin{equation} \label{eq-crutial}
 \big\langle \vm\, \phi\,  , \,  \, g_n\,  e^{i n\theta} \big\rangle  = - 4 \pi \alp\, c_0 \, , \qquad 
 \big\langle   \vm\, \phi, \, g_{n+1}\,  e^{i (n+1)\theta}\, \big\rangle =  \big\langle   \vm\, \phi, \, g_{n+2}\,  e^{i (n+2)\theta}\, \big\rangle =
 0. 
\end{equation}
This implies 
\begin{equation} \label{eq-cr-2}
 G_2^\lambda\,  \vm\, \phi=G_7^\lambda\,  \vm\, \phi= 0, \qquad 	\text{and} 	\qquad
G_6^\lambda\,  \vm\, \phi= \mathcal{C}_6(n,\lambda) \, e^{i(n+1)\theta}\, g_{n+1}(r) 
\end{equation}
for some  $\mathcal{C}_6(n,\lambda) \in\C$. Hence by \eqref{eq-crutial} 
$$
\LL  \vm\,\phi, \, G_6^\lambda\,  \vm\, \phi\,  \RR =0, 
$$
Proposition \ref{prop-exp} and Lemma \ref{lem-5} thus yield
$$
\LL  \vm\,\phi, \, M(\lambda)\, \phi\,  \RR = \ \zeta^{-1}(\alp)\,  |c_0|^2\, \lambda^\alp+|c_0|^2\, \eta_\m\,  \lambda  +   \mathcal{O}(\lambda^{2\alp}).
$$
Finally, equation \eqref{M-phi} follows from \eqref{gwi}, \eqref{gw46}, \eqref{eq-crutial} and Proposition \ref{prop-exp}.
\end{proof}

With the above Lemmas at hand we can calculate the entries of $E(\lambda)$ with the needed precision. 
We have 
\begin{equation}  \label{e11}
E_{11}(\lambda) =  \LL  \vm\, \phi, \, M(\lambda)\, \phi\,  \RR - \LL  \vm\, \phi, \, M(\lambda)Q_0\Omega(\lambda)Q_0 M(\lambda)\, \phi\,   \RR . \\[2pt]
\end{equation} 
Using equations \eqref{eq-crutial}, \eqref{B0-eq-1} and Lemma \ref{lem-5}(1), we find that
\begin{equation*}
\LL  \vm\, \phi, \, M(\lambda)Q_0\Omega(\lambda)Q_0 M(\lambda)\,\phi\,   \RR = \mathcal{O}\big(|\lambda|^{2\alp}\big).
\end{equation*} 
Equation \eqref{phi-phi} then gives 
\begin{equation}  \label{e11-2}
E_{11}(\lambda) =   \zeta^{-1}(\alp)\,   |c_0|^2\, \lambda^\alp  +  \eta_\m\,  |c_0|^2\, \lambda  + \mathcal{O}\big(|\lambda|^{2\alp}\big)\,  .
\end{equation} 
Similarly, by Corollary \ref{cor-G12} and equation \eqref{M-phi} 
\begin{align}  
E_{1j}(\lambda) & =  \LL  \vm\, \phi, \, M(\lambda)\, \psim_{j-1}\RR - \LL  \vm\, \phi, \, M(\lambda)\, Q_0\, \Omega(\lambda)\, Q_0 \, M(\lambda)\, \psim_{j-1} \RR  =  \lambda \, L_{j-1} + \mathcal{O}(|\lambda|^{1+\alp}) ,  \label{e1j}
 \end{align} 
where $j=2,\dots,n+1$, and
\begin{equation} \label{L-m}
L_k :=  \LL \phi , \,  \vm\, G_3  \vm\,  \psim_k \RR , \quad k=1,\dots, n.
\end{equation}
In the same way it follows that 
\begin{equation*} 
E_{j1}(\lambda) = \lambda \, \overline{L}_{j-1} + \mathcal{O}(|\lambda|^{1+\alp}), \qquad j=2,\dots, n+1 \, .
\end{equation*} 
Finally, by Lemma \ref{lem-ef-mat}, for $\ell, m\in\{2,\dots, n+1\}$ we have 
\begin{align}  \label{ejk}
E_{\ell m}(\lambda) &= \LL\,   \vm\, \psim_{\ell-1}\, , \, M(\lambda)\, \psim_{m-1}\RR - \LL  \vm\, \psim_{\ell-1}\, , \, M(\lambda)Q_0\Omega(\lambda)Q_0 M(\lambda)\, \psim_{m-1} \RR  \nonumber \\
&= \lambda \, T_{\ell-1,m-1} + \lambda^{1+\alp}\, \delta_{\ell,n+1} \delta_{m,n+1} + \mathcal{O}(\lambda^2) ,  
\end{align}
with
\begin{equation} \label{tlm}
T_{jk} :=  \LL\,  \psim_j\, , \,  \psim_k \RR , \quad j, k=1,\dots, n.
\end{equation}

Summing up the above calculations we conclude with
\begin{equation} \label{E-matrix}
E(\lambda) =  \left( \begin{array}{lr}
E_{11}(\lambda) &  \lambda\, L +  \mathcal{O}(\lambda^{1+\alp}) \\
\lambda\,  L^*+\mathcal{O}(\lambda^{1+\alp})  &  S(\lambda) + \mathcal{O}(\lambda^2)  
\end{array} \right) \, ,
\end{equation}
where the entries of the $n\times n$ matrix $S(\lambda)$ are given by 
\begin{equation} \label{s-lambda}
S_{jk}(\lambda) = \lambda\, T_{jk} + \lambda^{1+\alp} \,  \frac{|d^\m_{n}|^2}{\zeta(1+\alp)} \ \delta_{j,n} \delta_{k,n}\, .
\end{equation}
Since  $\{\psim_k\}_{k=1}^n$ are linearly independent, the matrix $T$ defined by \eqref{tlm} is invertible. We can thus define the operator $\pom : \Lp^2(\R^2) \to \NN_e(\Pm(\A))$ by
\begin{equation} \label{poh}
\pom  = \sum_{j,k=1}^n\psim_j\,  (T^{-1})_{jk}\,  \LL\,  \cdot\, ,\psim_k \RR .
\end{equation} 
$\pom $ is clearly self-adjoint and a short calculation shows that $(\pom )^2=\pom $. Hence $\pom $ is an orthogonal projection  on the zero eigenspace of $\Pm(\A)$.

\subsection*{The Feshbach formula} To calculate the inverse of $E(\lambda)$ we will apply the Feshbach formula. Let us recall its version for matrices: Let $\mathfrak B$ be a complex-valued $N\times N$ matrix of the form 
$$
 \mathfrak B =  \left( \begin{array}{lr}
\bb_{11} & \bb_{12}  \\
\bb_{21} &  \bb_{22} 
\end{array} \right) \, ,  
\qquad 
\bb_{22} : \C^{N-j} \to \C^{N-j}, \qquad   1\leq j\leq N-1,
$$
and suppose hat $\bb_{22}$ has a bounded inverse. Then $\mathfrak B$ is invertible with bounded inverse if and only if 
\begin{equation} \label{b-feshbach}
\bb := \big (\bb_{11} - \bb_{12}\, \bb_{22}^{-1}\, \bb_{21}\big)^{-1}  
\end{equation}
exists and is bounded on $\C^j$. If this is the case, then 
\begin{equation} \label{eq-feshbach}
\mathfrak B^{-1} = 
\left( \begin{array}{cc}
\bb & -\bb \, \bb_{12}\, \bb_{22}^{-1} \\[4pt] 
- \bb_{22}^{-1} \bb_{21}\, \bb &  \bb_{22}^{-1} \bb_{21}\, \bb \, \bb_{12}\, \bb_{22}^{-1} + \bb_{22}^{-1} 
\end{array} \right) \, . \\[7pt]
\end{equation}

\begin{rem}
Formula \eqref{eq-feshbach} holds in much more general settings, see e.g.~\cite[Lem.~2.3]{jn}.
\end{rem}

\subsection{Expansion of $\Rm(\lambda, \A)$} Before stating the main result of this section we introduce some additional notation. 
Recall that  $\Psim$ and $\psim_n$ are defined in  \eqref{phi-2} and \eqref{ef-hat}.
Let  ${\rm X}_\m : \HH^{1,-s} \to \NN_e(P_\m(\A))$  and $\varphi^\m \in \NN_r(P_\m(\A))$ be given by 
\begin{equation} \label{varphi-minus}
{\rm X}_\m= \pom\,   \vm G_3  \vm \qquad \text{and} \qquad \varphi^\m  = (1- {\rm X}_\m) \, \Phi_1^\m\, . \\[3pt]
\end{equation}
Furthermore, we define the zero eigenfunction $\psim$ by 
\begin{equation} \label{psi-minus}
\psi^\m  =  \pom\,   \vm\, \psim_n \, .
\end{equation}
Finally, we denote 
\begin{equation}  \label{omega-m1}
 \upsilon_\m: = \LL\, \Psim ,\,  \vm G_3  \vm\, \pom  \, \vm G_3  \vm\,  \Psim\, \RR , 
\end{equation}
and  
\begin{equation} \label{constants-m}
\mathfrak{c}_\m:= \eta_\m-\upsilon_\m, \qquad   \nu_\m :=   \frac{ |d^\m_{n}|^2}{\zeta(1+\alp)}  \, , \qquad  \omega_\m =  \nu_\m\, (T^{-1})_{nn}\, , \\[4pt]
\end{equation}
with $\eta_\m$ defined in \eqref{xi-beta}. 
Note that while $\cc_\m\in \R$, the coefficients $\nu_\m$ and $\omega_\m$ are complex. We have 

\begin{thm} \label{thm-pauli-res}
Let $1 <\alpha\not\in\Z$. 
Suppose that $B$ satisfies  \ref{ass-B} with $\rho > 7$. 
Then, as $\lambda\to 0$, 
\begin{equation} \label{res-pauli}
\Rm(\lambda,\A)   = -\lambda^{-1} \, \pom \, + \frac{\nu_\m\, \lambda^{\alp-1}}{1+\omega_\m\, \lambda^{\alp}}\  \psim  \LL   \ \cdot\ , \, \psim \RR \, 
- \frac{ \zeta(\alp)\, \lambda^{-\alp}}{1 +\cc_\m\, \zeta(\alp)\, \lambda^{1-\alp}}\ \varphi^\m \,  \LL   \ \cdot\ ,\,  \varphi^\m \RR \,  +\mathcal{O}\big(1 \big) \\[3pt]
\end{equation} 
holds in $\B(-1,s;1,-s)$ for all $s>3$ with $\pom $ given by \eqref{poh}.
\end{thm}

\begin{proof} 
In view of Corollary \ref{cor-w0} and equation \eqref{op-W} it suffices to prove the claim for any $s\in (3, \rho -4)$. We apply equations \eqref{eq-res1} and \eqref{grushin}.
For the calculation of $E^{-1}(\lambda)$ we use the Feshbach formula \eqref{eq-feshbach} with $j=1, \bb_{11} = E_{11}(\lambda), \bb_{12} = \lambda\, L +  \mathcal{O}(\lambda^{1+\alp}), \, 
 \bb_{21} = \lambda\, L^* +  \mathcal{O}(\lambda^{1+\alp})$, and $\bb_{22}= S(\lambda) + \mathcal{O}(\lambda^2)$. First we note that 
$$
\bb_{22}^{-1} = S^{-1}(\lambda) +\mathcal{O}(1) = \lambda^{-1} T^{-1} + \mathcal{O}(\lambda^{\alp-1}). 
$$
Hence in view of \eqref{e11-2}, \eqref{b-feshbach} and  \eqref{omega-m1}, 
\begin{align*}
\bb = \big(E_{11}(\lambda)  - \lambda L T^{-1} L^* +  \mathcal{O}(\lambda^{1+\alp}) \big)^{-1}  &=  \big(\zeta^{-1}(\alp)\,   |c_0|^2\, \lambda^\alp  +  \eta_\m\,  |c_0|^2\, \lambda  
-\upsilon_\m\,  |c_0|^2\, \lambda + \mathcal{O}\big(|\lambda|^{2\alp}\big)^{-1} \\[3pt]
&=  \frac{ |c_0|^{-2}  \zeta(\alp)\, \lambda^{-\alp}}{1 +\mathfrak{c}_\m\, \zeta(\alp)\, \lambda^{1-\alp}}\,  + \mathcal{O}(1)\, . 
\end{align*}
The Feshbach formula thus implies that 
 $E(\lambda)$ is invertible for $\lambda$ small enough and that 
\begin{equation} \label{E-inv-1} \\[3pt]
E^{-1}(\lambda)= 
 \frac{ |c_0|^{-2} \,  \zeta(\alp)\, \lambda^{-\alp}}{1 +\mathfrak{c}_\m\, \zeta(\alp)\, \lambda^{1-\alp}}
\left( \begin{array}{cl} 
1  & - L \, T^{-1} \\
- T^{-1}\,  L^*   &  T^{-1} \,  L^* L \, T^{-1}
\end{array} \right) +
\left( \begin{array}{cc}
0  & 0 \\
0 & S^{-1}(\lambda) 
\end{array} \right)
+ \mathcal{O}(1) .
\end{equation}
Next we expand $S^{-1}(\lambda)$. Let $T^{-1}_n$ denote the column vector with components $(T^{-1})_{1n}, \dots ,(T^{-1})_{nn}$. By \eqref{s-lambda} we get, with obvious abuse of notation, 
$$
\lambda^{-1}\,  T^{-1}\, S(\lambda)  = \id + \lambda^{\alp} N , \qquad N= \nu_\m \big(0, 0,\dots ,T^{-1}_n\big)\, .
$$
The Neumann series then gives, 
\begin{align} \label{S-inverse}
 S^{-1}(\lambda) T & = \lambda^{-1} \id  +\lambda^{-1} \sum_{k=1}^\infty (-\lambda^{\alp})^k\, N ^k =  \lambda^{-1} T^{-1} +\lambda^{-1} \sum_{k=1}^\infty (-\lambda^{\alp})^k\, (\nu_\m)^k \, (T^{-1})_{nn}^{k-1}\, \big(0, 0,\dots ,T^{-1}_n\big) \nonumber \\
& =  \lambda^{-1} \id - \frac{\nu_\m\, \lambda^{\alp-1}}{1+\omega_\m\, \lambda^{\alp}}\ \big(0, 0,\dots ,T^{-1}_n\big) .
\end{align}
Recall that $M(\lambda)\phi= \mathcal{O}(\lambda^\alp)$ and $M(\lambda)\psim_k = \mathcal{O}(\lambda)$ for all  $k=1,\dots,n$. Hence by \eqref{D-exp}, \eqref{a12} and \eqref{grushin} 
\begin{equation*} 
M(\lambda)^{-1}  = \big(1+R_0(\lambda)  \vm\big)^{-1} =  -J\, E^{-1}(\lambda)\, J^* + \mathcal{O}(1)
\end{equation*} 
on $\HH^{-1,s}$.
Equations  \eqref{eq-res1}, \eqref{eq-cr-2},  and  Corollary \ref{cor-G12} then yield
\begin{equation} \label{R-grushin} 
\Rm(\lambda,\A)   =  -J\, E^{-1}(\lambda)\, J^*\, G_0 + \mathcal{O}(1).
\end{equation}  
From equation \eqref{E-inv-1} and from the definitions of $T^{-1}$ and $L$ we now deduce that 
\begin{align} 
J\, E^{-1}(\lambda)\, J^*\, G_0  &=  \lambda^{-1}\, \pom \,   +  \frac{ |c_0|^{-2} \,  \zeta(\alp)\, \lambda^{-\alp}}{1 +\mathfrak{c}_\m\, \zeta(\alp)\, \lambda^{1-\alp}} \Big[ \, \phi\,    \LL \cdot \, , \phi\,  \RR\, -  \phi
\sum_{j,k=1}^n  L_j \, (T^{-1})_{jk}\,  \LL\,  \cdot \, ,\psim_k \RR   \nonumber   \\
& \quad -\sum_{j,k=1}^n (T^{-1})_{jk}\, L^*_k\,  \psim_ j \LL\,  \cdot \, ,\phi\,   \RR  +\!\!\! \sum_{i, \ell, j,k=1}^n   \psim_ \ell \, (T^{-1})_{i\ell}\ L^*_i\,  L_j \,(T^{-1})_{jk}\,  \LL\,  \cdot \, ,  \psim_ k\,  \RR \Big]  \label{jej} \\
& \quad  - \frac{\nu_\m\, \lambda^{\alp-1}}{1+\omega_\m\, \lambda^{\alp}}\  \sum_{j,k=1}^n   \LL\,  \cdot \, ,\psim_k\,  \RR \, \psim_j\,    (T^{-1})_{jn} \, (T^{-1})_{nk}\nonumber
+ \mathcal{O}(1)\, .
\end{align}
To simplify the notation  let us abbreviate 
\begin{equation}\label{phi-pauli}
\Pi_0^\m = \LL \,\cdot\, , \Psim \RR\, \Psim \, .  \\[3pt]
\end{equation}
Since $ \phi =   c_0\,  \Psim $, equations \eqref{L-m} and  \eqref{poh} then imply
$$
|c_0|^{-2} \,  \phi
  \sum_{j,k=1}^n  L_j \, (T^{-1})_{jk}\,  \LL\,  \cdot \, ,\psim_k \RR = \Pi_0^\m \, {\rm X}_\m^*\, , \qquad |c_0|^{-2}  \sum_{j,k=1}^n (T^{-1})_{jk}\, L^*_k\,  \psim_ j \LL\,  \cdot \, ,\phi\,   \RR =  {\rm X}_\m\, \Pi_0^\m \, 
$$
and
$$
|c_0|^{-2}   \sum_{i, \ell, j,k=1}^n   \psim_ \ell \, (T^{-1})_{i\ell}\ L^*_i\,  L_j \,(T^{-1})_{jk}\,  \LL\,  \cdot \, ,  \psim_ k\,  \RR = {\rm X}_\m\, \Pi_0^\m \, {\rm X}_\m^*\ . \\[3pt]
$$
However, in view of Lemma \ref{lem-vg3v-0}  and \eqref{ef-hat}  one has ${\rm X}_\m \, \psi_j^\m=\psi_j^\m$ for all $j =1,\dots, n$. Hence by \eqref{phi-2} and \eqref{varphi-minus}
$$
|c_0|^{-2} \,  \, \phi\,    \LL \cdot \, , \phi\,  \RR\ - \Pi_0^\m \, {\rm X}_\m^* - {\rm X}_\m\, \Pi_0^\m  + {\rm X}_\m\, \Pi_0^\m \, {\rm X}_\m^*  = (1-  {\rm X}_\m)\, \Psi_1^\m \, \LL \, \cdot\, , (1-  {\rm X}_\m)\, \Psi_1^\m \RR = \varphi^\m \,  \LL   \ \cdot\ ,\,  \varphi^\m \RR \, .
$$
As for the last term on the right hand side of \eqref{jej} we note that 
$ \sum_{j=1}^n   \, \psim_j\,    (T^{-1})_{jn} = - \pom \,   \vm\, \psim_n $, see \eqref{poh}  and \eqref{ef-hat} .
In combination with \eqref{psi-minus} this implies  
\begin{equation} \label{poh-w0}
 \sum_{j,k=1}^n   \LL\,  \cdot \, ,\psim_k\,  \RR \, \psim_j\,    (T^{-1})_{jn} \, (T^{-1})_{nk} =   \pom\,   \vm\, \psim_n  \LL  \ \cdot \ ,    \pom\,  \vm\, \psim_n \RR  =   \psim  \LL  \ \cdot \ ,   \psim \RR 
  \end{equation}
and the claim  follows from equation  \eqref{R-grushin}.
\end{proof}

\begin{rem} \label{rem-rhs}
Let us comment on the operators on the right hand side of \eqref{res-pauli}. The third term arises from the resonant state $\varphi^\m$. We will see below that if $B$ is radial, then 
$\varphi^\m = \Phi_1^\m$,  cf.~Corollary \ref{cor-radial}. The second term on the other hand side arises from the only eigenfunction of $P_\m(\A)$ which is not in $\Lp^1(\R^2)$, namely from $\psim_n$.

Note also that at most one of the denominators on the right hand side of \eqref{res-pauli} contributes to the singular part of  $\Rm(\lambda,\A)$. In particular, when $\alp=1/2$, then in view of \eqref{zeta}
$$
\Rm(\lambda,\A)   = -\lambda^{-1} \, \pom + i\, \lambda^{-\frac 12}\,  \Big(8  \pi\,  |d_n^\m|^2\,  \psim  \LL  \ \cdot\ , \psim \RR \,  +\frac{1}{2\pi}\,  \varphi^\m \LL \ \cdot\ , \varphi^\m \RR \Big)  +\mathcal{O}\big(1 \big) . \\[3pt]
$$
This is reminiscent of  the resolvent expansion of a Schr\"odinger operator in $\R^3$ when zero is an eigenvalue and a resonance, see \cite[Thm.~6.5]{JK}.
\end{rem}

In the absence of zero eigenfunctions we have

\begin{thm} \label{thm-pauli-res2}
Suppose that $B$ satisfies Assumption \ref{ass-B} with $\rho > 7$.  Let $0<\alpha<1$.  Then as $\lambda\to 0$, 
\begin{align} \label{res-pauli-2}
\Rm(\lambda,\A)  & = 
- \frac{ \zeta(\alpha)\, \lambda^{-\alpha}}{1 +\eta_\m\, \zeta(\alpha)\, \lambda^{1-\alpha}}\  \Phim_1 \,  \LL \,\cdot\, , \Phim_1 \RR\ + K_0^\m + \mathcal{O}\big(|\lambda|^\mu \big) 
\end{align} 
holds in $\B(-1,s;1,-s)$ for all $s>3$, where 
\begin{equation} \label{k-hat} 
K_0^\m  = \Omega_0 G_0-\frac{1}{4\pi \alpha} \big (\, \Phim_1\, \lef\, \cdot \, , \, \Omega_0 \, g_0\rig +\Omega_0\,  g_0\, \lef\, \cdot\, , \,  \Phim_1\rig\, \big ) - \frac{\lef g_0,  \vm\, \Omega_0\, g_0 \rig +4\pi\alpha \delta_0+16\pi^2 \alpha^2}{16\pi^2 \alpha^2} \  \LL \,\cdot\, , \Phim_1 \RR\, \Phim_1 \, .
\end{equation}
Recall that $\Omega_0$ and $\delta_0$ are given by \eqref{D-lambda} and \eqref{delta-m} respectivelly. 
\end{thm}

\begin{proof} 
Because $\alpha <1$, we have $\alpha=\alp$ and  $n=0$. Consequently, $\phi = c_0 \Phim = c_0 \Psim$. 
We are going to use the identity 
\begin{equation} \label{mg01}
R_\m(\lambda,\A) = M^{-1}(\lambda)\,  R_0(\lambda) = M^{-1}(\lambda)\, G_0 + \lambda^\alpha M^{-1}(\lambda)\, G_1^0 + \mathcal{O}\big(|\lambda|^{\mu} \big) , \\[3pt]
\end{equation}
see equations \eqref{grushin}, \eqref{eq-crutial} and \eqref{D-exp}. To prove the claim we have to expand $E_{11}(\lambda)$ to a higher order with respect to expansion \eqref{e11-2}.  Recall that $\Omega_0 Q_0 = Q_0\Omega_0= \Omega_0$, see Lemma \ref{lem-5}. From Lemma \ref{lem-vg3v} and equations \eqref{eq-crutial}, \eqref{e11} we thus conclude that 
\begin{equation*} 
E_{11}(\lambda) = \frac{|c_0|^{2}}{\zeta(\alpha)}\,  \big(1+\eta_\m \zeta(\alpha)\,  \lambda^{-\alpha} \big)  \lambda^{\alpha}   - |c_0|^{2} \, |\lambda|^{2\alp}\,  \frac{\lef g_0,  \vm\, \Omega_0\, g_0 \rig +4\pi\alpha \delta_0}{\,16\pi^2 \alpha^2\, |\zeta(\alpha)|^2}  + \mathcal{O}\big(|\lambda|^{3\alp} \big)  + \mathcal{O}\big(|\lambda|^{1+\alp} \big)  \, ,
\end{equation*}
where we have used the fact that $\upsilon_\m=0$.
This further implies 
\begin{equation}  \label{e11-inv}
 |c_0|^{2} \, E_{11}^{-1}(\lambda) = \frac{ \zeta(\alp)\, \lambda^{-\alp}}{1 +\eta_\m\, \zeta(\alp)\, \lambda^{1-\alp}}\, +\frac{\lef g_0,  \vm\, \Omega_0\, g_0 \rig +4\pi\alpha \delta_0}{16\pi^2 \alpha^2} +  \mathcal{O}\big(|\lambda|^{\mu} \big)  \, . \\[4pt]
\end{equation}
On the other hand, by Lemma \ref{lem-5} and equations \eqref{jj}, \eqref{D-exp}, \eqref{eq-cr-2} 
$$
a_{12}(\lambda) a_{21}(\lambda) = Q -\lambda^\alpha Q\, G_1^0\,  \vm\, \Omega_0 -\lambda^\alpha \Omega_0\,  G_1^0\,  \vm\, Q +  \mathcal{O}\big(|\lambda|^{\alpha+\mu} \big)\, .
$$
Hence using \eqref{grushin} and \eqref{eq-crutial} we obtain
\begin{align*}
M^{-1}(\lambda) G_0 &= \Omega_0 G_0 -\frac{ \zeta(\alp)\, \lambda^{-\alp}\,  \LL \,\cdot\, , \Phim_1 \RR\, \Phim_1}{1 +\eta_\m\, \zeta(\alp)\, \lambda^{1-\alp}}   - \frac{\lef g_0,  \vm\, \Omega_0\, g_0 \rig +4\pi\alpha \delta_0}{16\pi^2 \alpha^2} \  \LL \,\cdot\, , \Phim_1 \RR\, \Phim_1  
\\& \quad 
+  \frac{1}{4\pi\alpha} \Big( \, \Phim_1\, \LL \, \cdot \, ,  G_0\, \Omega_0^*\,  \vm\, g_0\RR - \Omega_0 g_0 \LL \, \cdot\, , \Phim_1\RR \Big)+ \mathcal{O}\big(|\lambda|^{\mu} \big)\, , \\[6pt]
 \lambda^\alpha\, M^{-1}(\lambda) \, G_1^0& = - \frac{Q\, g_0 \, \LL\, \cdot\, , \, g_0\RR}{16 \pi^2 \alpha^2} + \mathcal{O}\big(|\lambda|^{\mu} \big) = -\frac{\Phim}{4\pi\alpha}\ \, \LL\, \cdot\, , \, g_0\RR + \mathcal{O}\big(|\lambda|^{\mu} \big)\, .
\end{align*}
Now, by Proposition \ref{prop-Omega0},
$$
G_0\, \Omega_0^*\,  \vm = G_0\,  \vm\, \Omega_0 = (1+G_0   \vm)\, \Omega_0 -\Omega_0 = Q_0-\Omega_0\, .
$$
This implies 
\begin{align*}
M^{-1}(\lambda) G_0  +  \lambda^\alpha\, M^{-1}(\lambda) \, G_1^0 &= \Omega_0 G_0 - \frac{ \zeta(\alp)\, \lambda^{-\alp}}{1 +\eta_\m\, \zeta(\alp)\, \lambda^{1-\alp}}\,  \LL \,\cdot\, , \Phim_1 \RR\, \Phim_1 - \frac{\lef g_0,  \vm\, \Omega_0\, g_0 \rig +4\pi\alpha \delta_0}{16\pi^2 \alpha^2} \  \LL \,\cdot\, , \Phim_1 \RR\, \Phim_1  
\\ &\quad    - \frac{1}{4\pi\alpha} \big( \, \Phim_1\, \LL \, \cdot \, ,   \Omega_0\, g_0\RR  +\Omega_0 g_0 \LL \, \cdot\, , \Phim_1\RR \big)  -\frac{1}{4\pi\alpha}\  \Phim_1\, \LL\, \cdot\, , \, Q g_0\RR + \mathcal{O}\big(|\lambda|^{\mu} \big)\, \\[4pt]
& = - \frac{ \zeta(\alp)\, \lambda^{-\alp}}{1 +\eta_\m\, \zeta(\alp)\, \lambda^{1-\alp}}\,  \LL \,\cdot\, , \Phim_1 \RR\, \Phim_1 + K_0^\m + \mathcal{O}\big(|\lambda|^\mu \big) \, ,
\end{align*}
where we have used \eqref{eq-crutial} again. The claim follows from \eqref{mg01}.
\end{proof}

Formula \eqref{res-pauli} simplifies considerably when the magnetic field is radial because of the absence of interaction between the zero modes.

\begin{cor} \label{cor-radial}
Let $0 <\alpha\not\in\Z$. 
Suppose that $B$ satisfies Assumption \ref{ass-B} with $\rho > 7$. If $B$ is radial, then 
\begin{equation} \label{R-eq-radial}
\Rm(\lambda, \A)  = -\lambda^{-1} \, \pom +\frac{ \lambda^{\alp-1}}{1+\omega_\m\, \lambda^{\alp}}\,  \frac{ \LL \,\cdot\, , \Phim_2\RR\, \Phim_2 }{ \zeta(1+\alp) \big\|\Phim_2\big\|_2^4} \, - \frac{\zeta(\alp)\, \lambda^{-\alp}}{1 +\eta_\m \, \zeta(\alp)\, \lambda^{1-\alp}}\,  \LL \,\cdot\, , \Phim_1 \RR\, \Phim_1  +\mathcal{O}(1)\,  \\[3pt]
\end{equation} 
as $\lambda\to 0$. Recall that the functions $\Phim_j$ are defined in \eqref{phi-eq}. If $\alpha <1$, then equation \eqref{R-eq-radial} holds with $\pom=\Phi_2^\m=0$.
\end{cor}

\begin{proof}
By \eqref{g3-kernel}
$$ 
G_3(x,y) = \sum_{m\in\Z} G_{m,3}(r,t )\, e^{ im (\theta-\theta')}\, .
$$
From equation \eqref{W-rad} we thus infer that ${\rm X}_\m\, \Pi_0^\m = \Pi_0^\m {\rm X^*} =0$, which implies $\upsilon_\m=0$ and $\varphi^\m=\Phim_1$, cf.~\eqref{varphi-minus}. Moreover, in view of \eqref{psij-radial} and of the normalization of $\psim_n$ we deduce that $\psi_n^\m = d_n\, \Phim_2$. Hence by \eqref{psi-minus}, \eqref{constants-m} and \eqref{norm1}
\begin{equation} \label{pw0-radial}
\nu_\m \,   \psim\, \LL \,  \cdot\ , \psim \RR=  \frac{ \LL \,\cdot\, , \Phim_2\RR\, \Phim_2 }{ \zeta(1+\alp) \big\|\Phim_2\big\|_2^4}\, .
\end{equation}
Now equation \eqref{res-pauli} yields the claim.
\end{proof}

\subsection{Expansion of $\Rp(\lambda, \A)$}  In this section we state the threshold expansion of $\Rp(\lambda, \A)$. As already mentioned, this expansion is regular and could be deduced from \cite[Thm.~2.3]{ko}. In order to complete the picture we sketch the proof.

\begin{prop} \label{prop-pauli-regular}
Suppose that $B$ satisfies Assumption \ref{ass-B} for some $\rho > 7$, and that $0 <\alpha\not\in\Z$.  Then as $\lambda\to 0$, 
\begin{align} \label{res-pauli-reg}
\Rp(\lambda,\A)  & = (1+ G_0 \vp )^{-1}\, G_0  + (1+ G_0 \vp)^{-1} F_1(\alpha)\,  (1+ \vp G_0)^{-1}\, \lambda^{\mu} 
+o( \lambda^{\mu} )
\end{align} 
holds in $\B(-1,s;1,-s)$ for all $s>3$.  Here $\mu$ is given by \eqref{mu}, and  
\begin{equation*}
 F_1(\alpha) = G_1^0 \quad \text{if}\quad \alp <\frac 12, \qquad  F_1(\alpha) = G_2^0 \quad \text{if} \quad \alp >\frac 12, \qquad 
 F_1(\alpha) = G_1^0 +G_2^0 \quad \text{if}\quad \alp =\frac 12\, .
\end{equation*}
\end{prop}

\begin{proof}
Let $s>3$. By Proposition \ref{prop-exp}  
$$
R_0(\lambda) = G_0 +  F_1(\alpha) \,\lambda^{\mu} +o( \lambda^{\mu} ) \qquad \lambda\to 0.
$$
Since $\NN(P_\pp(\A)) =\{0\}$, see Lemma \ref{lem-ah-cash}, we deduce from Lemma \ref{lem-5} that $1+ G_0 \vp$ is invertible in $\HH^{1,-s}$. Expanding $(1+ R_0(\lambda) \vp )^{-1}$ into the Neumann series then gives
$$
(1+ R_0(\lambda) \vp )^{-1} = (1+ G_0 \vp )^{-1} - (1+ G_0 \vp )^{-1} F_1(\alpha)\ \vp\,  (1+ G_0 \vp )^{-1} \, \lambda^{\mu} +o( \lambda^{\mu} ) 
$$
Moreover, by duality $1+ \vp G_0$ is invertible in $\HH^{-1,s}$, and 
$$
\vp  (1+ G_0 \vp )^{-1} \, G_0 =  (1+ \vp G_0  )^{-1} \, \vp G_0\, .
$$
holds in $\HH^{-1,s}$.
Hence
\begin{align*}
 (1+ R_0(\lambda) \vp )^{-1} R_0(\lambda) &= (1+ G_0 \vp )^{-1}\, G_0   +(1+ G_0 \vp)^{-1} F_1(\alpha)\big[1- (1+ \vp G_0  )^{-1} \, \vp G_0\big]\, \lambda^{\mu}+o( \lambda^{\mu} ) \\
& =  (1+ G_0 \vp )^{-1}\, G_0  + (1+ G_0 \vp)^{-1} F_1(\alpha)\,  (1+ \vp G_0)^{-1}\, \lambda^{\mu} \, 
+o( \lambda^{\mu} ), 
\end{align*} 
and  the claim follows from \eqref{eq-res1}.
\end{proof}

\begin{rem} 
The conditions on $\rho$ and $s$ in Proposition \ref{prop-pauli-regular} could be somewhat relaxed, see \cite[Thm.~2.3]{ko}. However, since we consider $\Rp(\lambda,\A)$ a part of the resolvent of the Pauli operator it is natural to use the same conditions as in Theorem \ref{thm-pauli-res}.
\end{rem}

\section{\bf Resolvent expansion of the Pauli operator; integer flux} 
\label{sec-exp-int}
\subsection{Resonant states}

Throughout this section we assume that $0< \alpha\in \Z$.  Lemma \ref{lem-ah-cash} then gives
$$
 {\rm dim} (\NN_r(\Pm(\A)))  = 2, \qquad {\rm dim} (\NN_e(\Pm(\A)))  = n=  \alpha-1 .
$$
To find the asymptotic expansion of $M^{-1}(\lambda)$ we proceed in the same way as in the case of non-integer flux. First we construct the basis of $\NN_e(\Pm(\A))$ by normalizing  the  eigenfunctions of $\Pm(\A)$ in the same way as in Section  \ref{ssec-zero-ef}. We thus obtain 
 $n$ zero eigenfunctions of $\Pm(\A)$ which satisfy \eqref{ef-hat}.  This time, however, in addition to the zero eigenfunctions there are two linearly independent resonant states. To construct a basis of $\NN_r(\Pm(\A))$ satisfying \eqref{norm2} we let
\begin{equation} \label{phi-hat-2}
\phi_2 = c_2\, \Psim_2,  \qquad \Psim_2 = \Phim_2 + \sum_{j=1}^n\,  \LL \, \Phim_2,  \vm\, \psim_j\RR\, \psim_j ,
\end{equation} 
 where the constant $c_2$ is chosen so that $ \LL \, \phi_2, - \vm\, \phi_2\RR =1$. Next we define
\begin{equation} \label{kappa}
\kappa_{\scriptscriptstyle -} = \LL \Psim_2 ,\,  \vm\, \Phim_1 \RR, 
\end{equation} 
and put
\begin{equation}  \label{phi-hat-1}
\phi_1 = c_1\, \Psim_1  \qquad \text{with} \qquad \Psim_1= \Phim_1 +\kappa_{\scriptscriptstyle -}\, \Psim_2 + \sum_{j=1}^n\,  \LL \, \Phim_1,  \vm\, \psim_j\RR\, \psim_j ,
\end{equation} 
where the constant $ c_1$ is chosen such that $ \LL \, \phi_1,  \vm\, \phi_1\RR =-1$. In this way we obtain two linearly independent resonant states, $\phi_1$ and $\phi_2$, which satisfy conditions \eqref{norm2} with $W= \vm$.

\begin{rem} \label{rem-sigma} 
Contrary to the case of non integer flux, the resonant state $\Phim_1$ does not vanish at infinity, see \eqref{h-asymp},  \eqref{phi-eq}. As a consequence, the sub-leading term of $e^{h+i\chi}$ influences the expansion of $\phi_1$ up to the relevant order $|x|^{-1}$. This is where the coefficients $\vartheta_j$ enter into the expansion, see equation \eqref{varkappa} below.
\end{rem}

\begin{lem} \label{lem-res-int} 
Let  $\phi_1$ and $\phi_2$ be as above. Then 
\begin{align*} 
 \LL  \vm\, \phi_2, \, g_+\,   e^{i (\alpha+1) \theta} \RR & =  0,  \qquad \qquad \LL  \vm\, \phi_2, \, \g_\alpha\,   e^{i \alpha \theta} \RR = 0\, ,  \qquad\qquad\, \LL  \vm\, \phi_2, \, g_-\,   e^{i (\alpha-1) \theta} \RR  = -4\pi \, c_2\  \\ 
\LL  \vm\, \phi_1, \, \g_\alpha\,   e^{i \alpha \theta} \RR & =- 2\pi\, c_1 \, , \quad 
 \LL  \vm\, \phi_1, \, g_-\,   e^{i (\alpha-1) \theta} \RR  = -4\pi \, \varkappa \, c_1\,  ,  \quad  \LL  \vm\, \phi_1, \, g_+\,  e^{i (\alpha+1) \theta} \RR = 0 , 
\end{align*} 
where 
\begin{equation}\label{varkappa}
\varkappa= \kappa_{\scriptscriptstyle -} + \vartheta,  \qquad \text{and}\qquad \vartheta=\vartheta_1+i\vartheta_2. \\[3pt]
\end{equation}
Recall that the functions $\g_\alpha$ and $g_\pm$ are given by equations \eqref{g_m}, \eqref{g-alpha} and \eqref{short-pm}. 
\end{lem}

\begin{proof}
A combination of \eqref{phi-hat-2} resp.~\eqref{phi-hat-1} with equations \eqref{phi-infty} and \eqref{h-asymp}  shows that, as $|x|\to \infty$
\begin{align*}
\phi_1(x) & = c_1 \, e^{i \alpha \theta}  +c_1 \kappa_{\scriptscriptstyle -} \, |x|^{-1} \, e^{i (\alpha-1) \theta} + c_1 \vartheta \, |x|^{-1}\,  e^{i (\alpha-1) \theta} + \mathcal{O}(|x|^{-2})  \\
\phi_2(x) & = c_2 \, |x|^{-1}\,  e^{i (\alpha-1) \theta} +  \mathcal{O}(|x|^{-2}) .
\end{align*}
Application of Corollary \ref{cor-null} with $W=  \vm$  completes the proof.
\end{proof}

\subsection{The operator $M(\lambda)^{-1}$} As in Section \ref{sec-exp-non-int} we apply the Grushin-Schur method. However, since there are two resonant states, the procedure has to be modified accordingly. 

We define operators $\J_r: \C^{2} \to \HH^{1,-s}, \, \J_e: \C^{\alpha-1} \to \HH^{1,-s}\, $ and $\J^*_r: \HH^{1,-s}\, \to \C^{2}, \, \J^*_e: \HH^{1,-s}\, \to \C^{\alpha-1}$ by 
$$
\J_r \, z = z_1 \, \phi_1 + z_2 \,\phi_2 , \qquad \J_e\,  z = \sum_{k=1}^{\alpha-1} z_ {k}\,\psim_k, 
$$
and
$$
\J_r^* u =  \Big( \LL u, -  \vm\, \phi_1  \RR ,\, \LL u, -  \vm\, \phi_2  \RR \Big)^{\rm T} , \qquad  \J_e^* u =\Big(  \LL u, -  \vm \, \psim_1\RR , \dots ,  \LL u, -  \vm\, \psim_n\RR \Big)^{\rm T} \, .
$$
Now let 
\begin{equation} \label{j-split}
\J = \J_r+\J_e: \C^{\alpha+1} \to \HH^{1,-s}, \qquad \J^* = \J^*_r+\J^*_e :  \HH^{1,-s}\, \to \C^{\alpha+1} . \\[3pt]
\end{equation}
It follows that  $M(\lambda)$ is invertible in $\HH^{1,-s}\, $ if and only if  the $(\alpha+1)\times(\alpha+1)$ matrix 
\begin{align} \label{E-int}
E(\lambda) &= -\J^*M(\lambda) \J +\J^*M(\lambda)Q_0\Omega(\lambda)Q_0 M(\lambda)\J\, 
\end{align} 
is invertible in $\C^{\alpha+1}$. Furthermore, $M(\lambda)^{-1}$ satisfies equation \eqref{grushin} with $\Omega(\lambda)$ given by \eqref{D-lambda}, and with
\begin{align} \label{a12-int}
a_{12}(\lambda)& = \J- \Omega(\lambda)Q_0 M(\lambda) \J, \qquad a_{21}(\lambda)  = \J^*-\J^*M(\lambda)Q_0 \Omega(\lambda)\, .
\end{align}
Next we expand the matrix elements of $E(\lambda)$ to the needed level of precision. We start by estimating the second term on the right hand side of \eqref{E-int}. 
 With the help of \eqref{B0-eq-2-int} we expand $\Omega(\lambda)$, cf.~\eqref{D-lambda},  into a Neumann series; 
$$
\Omega(\lambda) = \Omega_0  -\frac{\Omega_0 Q_0 \G_1 Q_0 \Omega_0 }{\log\lambda-z_\alpha} + \mathcal{O}(\lambda\log\lambda)\, .
$$
The latter implies that 
\begin{equation} \label{E11-bis}
\LL  \vm\, \phi_1, \, M(\lambda)Q_0\Omega(\lambda)Q_0 M(\lambda)\, \phi_1 \RR = \frac{4\pi |c_1|^2}{(\log\lambda-z_\alpha)^2}  \Big(q_1 +\frac{q_2}{\log\lambda-z_\alpha} \Big)+  \mathcal{O}(\lambda) , \\[3pt]
\end{equation} 
for some $q_1, q_2\in\C$, and that all the other entries of $\J_r^*M(\lambda)Q_0\Omega(\lambda)Q_0 M(\lambda)\J_r$ are of order $\mathcal{O}(\lambda)$.

Now, let us consider the $2\times 2$ matrix 
\begin{equation*} 
\E(\lambda) = -\J_r^* M(\lambda) \J_r  + \J_r^*M(\lambda)Q_0\Omega(\lambda)Q_0 M(\lambda)\J_r\, .
\end{equation*}
By Proposition \ref{prop-exp-int}, equation \eqref{E11-bis},   and Lemma \ref{lem-res-int} it follows that
\begin{align} \label{e11-int}
\E_{11}(\lambda) 
 &=  \frac{4\pi\, |c_1|^2 }{\log\lambda -z_\alpha}\Big(1-\frac{q_1}{\log\lambda-z_\alpha}  -\frac{q_2}{(\log\lambda-z_\alpha)^2}\  \Big) - \pi |\varkappa \, c_1|^2\, \lambda\log\lambda + \mathcal{O}(\lambda) .
\end{align}
In the same way we find 
\begin{equation} \label{e12-int}
\begin{aligned}
\E_{22}(\lambda) &=   -\pi |c_2|^2\, \lambda\log\lambda + \pi |c_2|^2\, \w_\m\, \lambda +\mathcal{O}(\lambda/\log\lambda), \\[2pt]
 \E_{12}(\lambda) & =  \overline{\E_{21}} =  \pi \varkappa\,  c_1 \bar c_2\, \lambda \log\lambda  + \, \mathcal{O}(\lambda)  , 
\end{aligned}
\end{equation}
where 
\begin{align*}
\w_\m &= \frac{i |c_2|^{-2}}{16\pi }\,  \big| \LL  \vm\, \phi_2, \, g_+\,   e^{i (\alpha+1) \theta} \RR \big |^2 +  \frac{i |c_2|^{-2}}{16\pi }\, \big | \LL  \vm\, \phi_2, \, g_-\,   e^{i (\alpha-1) \theta} \RR \big |^2 +\frac{1}{\pi}\, \LL   \Phi^\m_2, \,\vm\, \G_3\,  \vm\, \Phi^\m_2 \RR 
\end{align*}
Note that in view of Lemma \ref{lem-res-int}  and Remark \ref{rem-G3-int} we have
\begin{equation} \label{w-alpha-2}
\w_\m = i\pi +m_\m  \qquad \text{with} \qquad m_\m:= \pi^{-1} \LL\,   \vm \Psim_2,\, \G_3\,  \vm \Psim_2 \RR \, \in \R\, .
\end{equation} 
To proceed we have to replace the vector \eqref{L-m}  by a $2\times(\alpha-1)$ matrix $\Ln$ with entries
\begin{equation} \label{mathcal-L} 
(\Ln )_{jk} = \LL \phi_j,   \vm\, \G_3\,   \vm\, \psim_k\RR\, .
\end{equation}
We then have 
\begin{equation*} 
 -\J_r^* M(\lambda) \J_e  + \J_r^*M(\lambda)Q_0\Omega(\lambda)Q_0 M(\lambda)\J_e\,  = \lambda \Ln + \mathcal{O}(\lambda/\log\lambda)\, . \\[2pt]
\end{equation*}
In order to expand the matrix elements of $E(\lambda)$ arising from eigenfunctions $\psim_j$ we need the following modification of Lemmas \ref{lem-jost} and \ref{lem-vg3v-0}.

\begin{lem}  \label{lem-jost-int}
Let $B$ satisfy assumption \ref{ass-B} with $\rho>4$. Then for $ j,k=1,\dots ,\alpha-1$, 
\begin{align} \label{i4-int}
\LL  \vm \,\psim_j,\,  \G_i\,   \vm\, \psim_k \RR  & = 0  \qquad i\in \{ 1,2, 4\}\\
\LL  \vm \,\psim_j, \, \G_5\,   \vm\, \psim_k \RR &  = 8 \pi^2  \,  |d^\m_n|^2  \, \delta_{j,n}\,  \delta_{k,n} \ , \label{i5-int}
\end{align}
where $d^\m_n:=d^\m_{n,n}$ is given by \eqref{ef-hat}.
\end{lem}

\begin{proof}
We  apply Corollary \ref{cor-null} with $W= \vm$. A comparison of the latter with the asymptotic behavior of $\psim_j$, see  equation \eqref{ef-hat}, gives 
\begin{equation} \label{w0gm-int}
\begin{aligned} 
\LL  \vm \,\psim_j,  \, \g_\alpha \, e^{i\alpha  (\cdot) } \RR  = \LL  \vm \,\psim_j,  \, g_\pm \, e^{i(\alpha\pm 1) (\cdot) } \RR=\LL  \vm \,\psim_j,  \, g_{\alpha+2} \, e^{i(\alpha+2) (\cdot) } \RR & = 0,  \quad \quad 1\leq j \leq n \\
 \LL  \vm \,\psim_j,  \, g_{\alpha-2} \, e^{i(\alpha-2) (\cdot)} \RR & = 0,   \qquad     1\leq j \leq n-1 \\
 \LL  \vm \,\psim_n,  \, g_{\alpha-2} \,e^{i(\alpha-2) (\cdot)}  \RR & = -8\pi\, d^\m_n\, .
 \end{aligned}
 \end{equation}
By  equations \eqref{G1-int} and \eqref{G2-int} this proves \eqref{i4-int} for $i=1,2$. Next we note that
\begin{equation} \label{gw45-int}
\begin{aligned}
\G_4\,   \vm\, \psim_j  & = \mathcal{\widetilde C}_4(j) \, e^{i\alpha  (\cdot) }\, \g_\alpha  \qquad\qquad\qquad\qquad\qquad\qquad\quad 1\leq j\leq n  \\
\G_5\,   \vm\, \psim_j  & = \mathcal{C}_5^+(j) \, e^{i(\alpha+1) (\cdot)}\, g_+  +  \mathcal{C}_5^-(j) \, e^{i(\alpha-1) (\cdot)}\, g_- \quad \  1\leq j\leq n-1 ,
\end{aligned} 
\end{equation} 
for some constants  $\mathcal{\widetilde C}_4(j),   \mathcal{C}_5^+(j) $ and  $ \mathcal{C}_5^-(j)$,
see  equations \eqref{g4-int} and \eqref{g5-int}. 
In view of \eqref{w0gm-int} this completes the proof of \eqref{i4-int}. For $j=n$ we get, using again \eqref{w0gm-int} and  \eqref{g5-int}
\begin{align*}
\G_5\,   \vm\, \psim_n  &  =  \mathcal{C}_5^+\, g_+ \, e^{i(\alpha+1) (\cdot)}  +  \mathcal{C}_5^-\, g_- \, e^{i(\alpha-1) (\cdot)}
  -\pi  d^\m_n \, g_{\alpha-2}\, e^{i(\alpha-2) (\cdot)}  . 
\end{align*}
Applying \eqref{w0gm-int} once again gives \eqref{i5-int}.
\end{proof}

\begin{lem} \label{lem-vg3v-int}
Let $B$ satisfy Assumption \ref{ass-B} with $\rho >4$, and let $u\in \Lp^{2,2+0}(\R^2)$ be such that $\G_1 u=\G_2^\lambda u=0$. Then
\begin{equation*}
\LL  \vm\, \psim_j, \, \G_3\, u\RR \ =\  - \LL \psim_j, \G_0\, u\RR\, , \qquad\forall\, j=1,\dots, n\, .
\end{equation*}
\end{lem} 

\begin{proof} 
This is a straightforward modification of the proof of Lemma \ref{lem-vg3v-0}.  From Proposition \ref{prop-exp-int} we infer that
\begin{align*}
\LL \psim_j, u\RR &= \LL \psim_j, (H_0-\lambda) \, R_0(\lambda) u \RR = \LL H_0\, \psim_j, (\G_0 +\lambda \G_3) u\RR -\lambda \LL \psim_j, \G_0 u\RR + o(\lambda) \\
&= \LL \psim_j, u\RR - \lambda\, \LL  \vm \, \psim_j, \G_3 u\RR -\lambda \LL \psim_j, 	\G_0 u\RR + o(\lambda).
\end{align*}
Diving by $\lambda$ and letting $\lambda\to 0$ completes the proof.
\end{proof}

\subsection{Expansion of $\Rm(\lambda,\A)$ }  Similarly as in Section \ref{sec-exp-non-int} we define the operator $ \X_\m= \pom \,  \vm\, \G_3\,  \vm$ and the resonant states 
\begin{equation} 
\varphi_2^\m = (1-\X_\m) \Phi_2^\m, \qquad \varphi_1^\m = (1-\X_\m) \Phi_1^\m + \kappa_\m\, \varphi_2^\m. 
\end{equation}
Moreover, we put
\begin{equation} \label{pi-12}
\Pi^\m_{jk} = \LL\, \cdot\, , \varphi_j^\m\RR\,   \varphi_k^\m,  \qquad \qquad j,k=1,2.
\end{equation}

We have

\begin{thm}  \label{thm-pauli-res-int}
Let $0<\alpha \in\Z$. Suppose that $B$ satisfies  \ref{ass-B} with $\rho > 7$. 
Then as $\lambda\to 0$, 
\begin{equation}  \label{res-pauli-int}
\Rm(\lambda,\A)   = -\lambda^{-1} \, \pom   +  \frac{ \Pi_{22}^\m }{\pi\,  \lambda\, (\log\lambda-\w_\m)} \,  - \K_\m \log\lambda\,     + \mathcal{O}(1)
\end{equation} 
holds in $\B(-1,s;1,-s)$ for all $s>3$, where
\begin{align*}
 \K_\m & = \frac{1}{4\pi}  \big[ \,\Pi_{11}^\m +\overline{ \varkappa}\, \Pi_{12}^\m + \varkappa\, \Pi_{21}^\m + |\varkappa|^2 \, \Pi_{22}^\m\big] + 
\frac{\pi\, |d^\m_n|^2 }{4} \ \psim  \LL \, \cdot \, ,  \psim \RR \,  .
\end{align*}
Recall that $\varkappa$ is given by Lemma \ref{lem-res-int} and equation \eqref{kappa}, and that $d^\m_n$ is defined in Lemma \ref{lem-jost-int}. Moreover, $\w_\m$ satisfies \eqref{w-alpha-2}. If $\alpha=1$, then the above formulas hold with $\pom =\Pi_{22}^\m=0$.
\end{thm}

\begin{proof} 
 It follows from \eqref{i4-int}, \eqref{i5-int}, Proposition \ref{prop-exp-int} and Lemma \ref{lem-vg3v-int}  that 
\begin{align*} 
\LL  \vm \,\psim_k, \, M(\lambda)\,\psim_j\RR &\  =\  \lambda\, \LL\psim_k ,  \psim_{j} \RR  -\frac{\pi}{4} \, \lambda^2\log\lambda\, |d^\m_n|^2\,  \delta_{j,n}\,  \delta_{k,n} + \mathcal{O}(\lambda^2)
\end{align*}
holds for all $j,k=1,\dots ,n$. Recall also that for integer flux we have $\alpha= n+1$, see equation \eqref{alpha}. Since \newline $\J_e^*M(\lambda)Q_0\Omega(\lambda)Q_0 M(\lambda)\J_e = \mathcal{O}(\lambda^2)$, we conclude with
\begin{equation} \label{E-matrix-int}
E(\lambda) =  \left( \begin{array}{cc}
\E(\lambda)  &  \lambda\, \mathcal L +  \mathcal{O}(\lambda/\log\lambda) \\
\lambda\, \mathcal  L^*+  \mathcal{O}(\lambda/\log\lambda) &  \sn(\lambda) + \mathcal{O}(\lambda^2)  
\end{array} \right) \, ,
\end{equation}
where $\sn(\lambda)$ is the $n\times n$ matrix with entries 
\begin{equation} \label{s-lambda-int}
\sn_{jk}(\lambda) = \lambda\, T_{jk}   -\frac{\pi}{4} \, \lambda^2\log\lambda\, |d^\m_n|^2\,  \delta_{j,n}\,  \delta_{k,n}\, .
\end{equation}
In order to expand $E^{-1} (\lambda)$, we first calculate the singular terms of the matrix $\E^{-1} (\lambda)$. By \eqref{e11-int} and \eqref{e12-int},
$$
{\rm det\, } \E(\lambda) =  -4\pi^2 |c_1 c_2|^2 \, \lambda\, \Big(1-\frac{q_1}{\log\lambda-z_\alpha}  -\frac{q_2}{(\log\lambda-z_\alpha)^2}\Big) \Big(1- \frac{\w_\m-z_\alpha}{\log\lambda -z_\alpha} + \mathcal{O}(\lambda \log\lambda) \Big)\, .
$$
The standard formula for the inverse of a $2\times 2$ matrix, which is a special case of \eqref{eq-feshbach}, then gives
\begin{align} \label{mathcal-E-inverse}
\E^{-1}(\lambda) &=  \frac{1}{{\rm det\, } \E(\lambda)} \left( \begin{array}{rr}
\E_{22}(\lambda) & -\E_{12}(\lambda) \\
-\E_{21}(\lambda) & \E_{11}(\lambda)
\end{array} \right) \nonumber \\[3pt]
&=  \frac{1}{4\pi\, |c_1\, c_2|^2}
\left( \begin{array}{rl}
|c_2|^2 \log\lambda & \bar c_1\, c_2\, \overline{ \varkappa}\, \log\lambda  \\
c_1\, \bar c_2\, \varkappa\, \log\lambda  & |c_1 \varkappa|^2\, \log\lambda - 4 |c_1|^2 \lambda^{-1} ( \log\lambda-\w_\m)^{-1}
\end{array} \right) +  \mathcal{O}(1)\, .
\end{align} 
To continue we note that $\sn(\lambda)$ is invertible for $\lambda$ small enough and that $\big(\sn(\lambda) + \mathcal{O}(\lambda^2)\big)^{-1} = \sn^{-1}(\lambda)+  \mathcal{O}(1)$.  Since $ \sn^{-1}(\lambda)= \mathcal{O}(\lambda^{-1})$, it follows that
$$
\Big[ \E(\lambda) - \big( \lambda\,  \mathcal L+ \mathcal{O}(\lambda/\log\lambda) \big) \big(\sn(\lambda) + \mathcal{O}(\lambda^2)\big)^{-1}\big(\lambda\, \mathcal L+\mathcal{O}(\lambda/\log\lambda) \big) \big]^{-1} =\E^{-1} (\lambda)  + \mathcal{O}(1).
$$
The Feshbach formula \eqref{eq-feshbach} now implies 
\begin{equation} \label{E-inverse-int}
E^{-1}(\lambda) =  \left( \begin{array}{cc}
\E^{-1} (\lambda) & - \E^{-1}(\lambda)\,  \mathcal L \, T^{-1} \\
- T^{-1}   \mathcal L^*\, \E^{-1}(\lambda) & T^{-1}  \mathcal L^*\, \E^{-1}(\lambda)\, \mathcal L \, T^{-1} 
\end{array} \right) + 
\left( \begin{array}{cc}
0  & 0 \\
0 & \sn^{-1}(\lambda) 
\end{array} \right)
+ \mathcal{O}(1) .
\end{equation} 
The inverse of $\sn(\lambda)$ is calculated in the same way as in case of non-integer flux, see equation \eqref{S-inverse}. Using \eqref{s-lambda-int}  and the Neumann series we find  
\begin{align} \label{S-inverse-int}
\sn^{-1}(\lambda) 
& =  \lambda^{-1} T^{-1} +\frac{\pi\, |d^\m_n|^2 }{4}\, \log \lambda \,  \big(0, 0,\dots ,T^{-1}_n\big)\,  T^{-1} +  \mathcal{O}(1). 
\end{align}
To continue let us abbreviate 
$$
{\rm P}_{jk} = \LL\, \cdot\, , \Psi_j^\m\RR\,   \Psi_k^\m . 
$$
Given $u\in \HH^{-1,s}$,  equations \eqref{mathcal-L} yields
$$
 \E^{-1}(\lambda)\, \mathcal L \, T^{-1} \, \J_e^*\, \G_0\, u = 
 \E^{-1}(\lambda) \Big( \begin{array}{c}
u_1\\
u_2
\end{array} \Big) =: \Big( \begin{array}{c}
v_1\\
v_2
\end{array} \Big) 
 , \quad  \text{where} \quad u_j = \LL \phi_j, \, \X_\m ^* u \RR  \, .
$$
Hence using \eqref{mathcal-E-inverse} we infer that 
\begin{align}
\J_e\, T^{-1}\,  \mathcal L^*\, \E^{-1}(\lambda)\, \mathcal L \, T^{-1} \, \J_e^*\, \G_0\, u & =\X_\m  \, \phi_1\, v_1 +\X_\m  \, \phi_2\, v_2 =  - \frac{\X_\m \,{\rm P}_{22}^\m\, \X_\m ^*\, u}{\pi\,  \lambda\, (\log\lambda-\w_\m)} \label{je-je}\\
& \quad + \frac{\log\lambda}{4\pi} \, \X_\m   \big[ \,{\rm P}_{11}^\m +\overline{\varkappa}\,{\rm P}_{12}^\m + \varkappa\,{\rm P}_{21}^\m +  |\varkappa|^2 \,{\rm P}_{22}^\m\big] \, \X_\m ^*\, u +  \mathcal{O}(1). \nonumber
\end{align}
In the same way,
\begin{align*}
\J_r\, \E^{-1}(\lambda)\, \mathcal L \, T^{-1} \, \J_e^*\, \G_0\, u &= - \frac{{\rm P}_{22}^\m\, \X_\m ^*\, u}{\pi\,  \lambda\, (\log\lambda-\w_\m)} + \frac{\log\lambda}{4\pi} \, \big[ \, {\rm P}_{11}^\m +\overline{ \varkappa}\, {\rm P}{12}^\m + \varkappa\, {\rm P}_{21}^\m +  |\varkappa|^2  \, {\rm P}_{22}^\m\big] \, \X_\m ^*\, u +  \mathcal{O}(1) \\[4pt]
\J_e \, T^{-1}   \mathcal L^*\, \E^{-1}(\lambda) \, \J_r^*\, \G_0\, u &= - \frac{ \X_\m  \, {\rm P}_{22}^\m\, u}{\pi\,  \lambda\, (\log\lambda-\w_\m)} + \frac{\log\lambda}{4\pi}\,  \X_\m  \big[ \, {\rm P}_{11}^\m +\overline{\varkappa}\, {\rm P}_{12}^\m + \varkappa\, {\rm P}_{21}^\m +  |\varkappa|^2  \, {\rm P}_{22}^\m\big] \,   u +  \mathcal{O}(1) ,
\end{align*}
and 
\begin{equation} \label{jr-jr}
\J_r \, \E^{-1}(\lambda)\, \J_r^*\, \G_0\, u = \frac{ {\rm P}_{22}^\m\, u}{\pi\,  \lambda\, (\log\lambda-\w_\m)} + \frac{\log\lambda}{4\pi} \big[ \, {\rm P}_{11}^\m +\overline{\varkappa}\, {\rm P}_{12}^\m + \varkappa\, {\rm P}_{21}^\m +    |\varkappa|^2 \, {\rm P}_{22}^\m\big] \,   u +  \mathcal{O}(1) . \\[3pt]
\end{equation}
Collecting the above identities gives 
\begin{align*}
& \J_e\, T^{-1}\,  \mathcal L^*\, \E^{-1}(\lambda)\, \mathcal L \, T^{-1} \, \J_e^*\, \G_0  +\J_r\, \E^{-1}(\lambda)\, \mathcal L \, T^{-1} \, \J_e^*\, \G_0  +\J_e \, T^{-1}   \mathcal L^*\, \E^{-1}(\lambda) \, \J_r^*\, \G_0 +\J_r \, \E^{-1}(\lambda)\, \J_r^*\, \G_0 \\[3pt]
& \quad =  \frac{ (1-\X_\m)  {\rm P}_{22}^\m (1-X_\m^*)}{\pi\,  \lambda\, (\log\lambda-\w_\m)} + \frac{\log\lambda}{4\pi} \,  (1-\X_\m)   \big[ \, {\rm P}_{11}^\m +\overline{\varkappa}\, {\rm P}_{12}^\m + \varkappa\, {\rm P}_{21}^\m +  |\varkappa|^2  \, {\rm P}_{22}^\m\big]  (1-\X_\m^*)  +  \mathcal{O}(1)  \\[3pt]
& \quad =   \frac{ \Pi_{22}^\m }{\pi\,  \lambda\, (\log\lambda-\w_\m)}  + \frac{\log\lambda}{4\pi} \,  \big[ \, \Pi_{11}^\m +\overline{\varkappa}\, \Pi_{12}^\m + \varkappa\, \Pi_{21}^\m +  |\varkappa|^2  \, \Pi_{22}^\m\big]  +  \mathcal{O}(1) , 
\end{align*}
where we have used the fact that 
$$
(1- \X_\m) \Psi_j^\m = \varphi_j^\m, \qquad j=1,2, \\[2pt]
$$
which follows from Lemma \ref{lem-vg3v-0} and equation \eqref{norm1}. 
Finally, from \eqref{S-inverse-int}, \eqref{poh} and \eqref{poh-w0} we get 
\begin{equation} \label{je-je-2}
\J_e\, \sn^{-1}(\lambda) \J_e^*\, \G_0  = \lambda^{-1}\, \pom    + \log\lambda\, \frac{\pi\, |d^\m_n|^2 }{4} \,  \psim\, \LL \, \cdot \, ,  \psim \RR \, +\mathcal{O}(1) .
\end{equation}
Using Lemma \ref{lem-res-int} and equations \eqref{w0gm-int} we now mimic the above calculations to obtain 
$\J\, E^{-1}(\lambda)\, \J^*\, M(\lambda)= \mathcal{O}(1)$, and $M(\lambda) \,J  E^{-1}(\lambda)\, \J^*\, = \mathcal{O}(1)$. Consequently, 
$$
M(\lambda)^{-1} =  \J\, E^{-1}(\lambda)\, \J^*  +  \mathcal{O}(1)\, , 
$$ 
and similarly as in \eqref{R-grushin} we infer that 
\begin{equation} \label{R-grushin-int}
R_\m(\lambda,\A)   =  -\J\, E^{-1}(\lambda)\, \J^*\, \G_0 + \mathcal{O}(1).
\end{equation} 
Inserting identities \eqref{je-je}-\eqref{je-je-2} into \eqref{E-inverse-int} and \eqref{R-grushin-int}  implies the claim.
\end{proof}

\begin{rem}
The coefficient on the right hand side of \eqref{res-pauli-int} proportional to $(\lambda \log\lambda)^{-1}$ arrises only from the $p-$wave resonant state $\varphi_2^\m$. Note that $\varphi_2^\m\in \Lp^p(\R^2)$ for every $p>2$. On the other hand, the operator $\K_\m $ contains contributions from the eigenfunction $\psim$ as well as from both resonant states $\varphi_1^\m$ and $\varphi_2^\m$. 
\end{rem}


As in the case of non-integer flux the resolvent expansion simplifies for radial magnetic fields.

\begin{cor} \label{cor-radial-int}
Suppose that $B$ satisfies Assumption \ref{ass-B} with $\rho > 7$. If $B$ is radial, then 
\begin{equation}  \label{res-int-radial}
\Rm(\lambda, \A)  =  -\lambda^{-1} \, \pom  +  \frac{ \LL\, \cdot\, , \Phim_2 \RR \, \Phim_2}{\pi\,  \lambda\, (\log\lambda-\w_\m)} \,  -  \frac{\log\lambda}{4\pi} \, \LL\, \cdot\, , \Phim_1 \RR \, \Phim_1  -\frac{\pi\, \log\lambda\,  \LL\, \cdot\, , \Phim_3 \RR \, \Phim_3}{4\, \|\Phim_3\|_{2}^4} \,   +\mathcal{O} (1)\,  \\[3pt]
\end{equation} 
holds in $\B(-1,s;1,-s)$ for all $s>3$, where $\Phim_3(x)$ is given by \eqref{phi-eq}.  
 If $\alpha=1$, then \eqref{res-int-radial} holds with $\pom =\Phim_3=0$. 
\end{cor} 

\begin{proof}
From  equations \eqref{psij-radial} and \eqref{W-rad} we deduce that $\X_\m \, \Phi_j^\m=0, j=1,2,$ and that 
$$
\LL \, \Phim_j,  \vm\, \psim_k \RR\ = 0 \qquad \forall\, j\in \{1,2\}, \quad \forall\, k\in \{1,\dots,\alpha-1\}. 
$$
Hence by \eqref{phi-hat-2},\eqref{phi-hat-1} and \eqref{kappa} we have $\kappa_{\scriptscriptstyle -}=0$, and
$$
\varphi_1^\m = \Phim_1 = (x_1+i x_2)^{\alpha}\, e^{h+i\chi}\, , \quad \varphi_2^\m = \Phim_2 = (x_1+i x_2)^{\alpha-1}\, e^{h+i\chi}\,  .
$$
By symmetry of $B$, we also have $\vartheta=0$, and hence $\varkappa=0$. Furthermore, using \eqref{norm1}, \eqref{psi-minus},  and the definition of $d^\m_n$, see Lemma \ref{lem-jost-int}, we find 
$$
|d^\m_n|^2 \, \psim\, \LL \, \cdot \, ,  \psim\RR  = |d^\m_n|^2\, \frac{ \LL\, \cdot\, , \psim_n \RR \, \psim_n}{\|\psim_n\|_{2}^4} = \frac{ \LL\, \cdot\, , \Phim_3 \RR \, \Phim_3}{\|\Phim_3\|_{2}^4}\, .
$$
Equation \eqref{res-int-radial} thus follows from Theorem \ref{thm-pauli-res-int}.
\end{proof}

\subsection{Expansion of $\Rp(\lambda, \A)$}  Here we state the expansion of $\Rp(\lambda, \A)$. The proof follows line by line the proof of Proposition \ref{prop-pauli-regular} and will be omitted. 

\begin{prop} \label{prop-pauli-regular-int}
Suppose that $B$ satisfies Assumption \ref{ass-B} for some $\rho > 7$, and that $0 <\alpha\in\Z$.  Then as $\lambda\to 0$, 
\begin{align} \label{res-pauli-reg-2}
\Rp(\lambda,\A)  & = (1+ \G_0 \vp )^{-1}\, \G_0  + (1+ \G_0 \vp)^{-1} \G_1\,  (1+ \vp \, \G_0)^{-1}\, \frac{1}{\log\lambda}  
+o\big((\log\lambda)^{-1} \big)
\end{align} 
holds in $\B(-1,s;1,-s)$ for all $s>3$.  
\end{prop}

\section{\bf  Resolvent expansions for $\alpha\leq 0$.}  
\label{sec-alpha-neg}
In this section we state the resolvent expansions for negative values of $\alpha$.

\subsection*{\bf Zero eigenfunctions of $P_\pp(\A)$} If $\alpha<0$, then by 
Lemma \ref{lem-ah-cash} we have   $\NN_e(P_\m(\A))  = \{0\}$, and $\NN_e(P_\pp(\A)) \neq \{0\}$ if and only if $\alpha <- 1$. Following the construction  of Section \ref{ssec-zero-ef} we then built the basis of $ \NN_e(P_\pp(\A))$ satisfying conditions \eqref{norm1} as follows;
\vskip0.2cm

\noindent We put $\ \psip_1 = d^\pp_{1,1} \, e^{i \chi-h} \ $ with the constant $d^\pp_{1,1}$ chosen such that  $\LL \psip_1, \, \vp\, \psip_1\RR = -1$.  
Next we put 
$$
\psip_2  = d^\pp_{2,1} \, e^{ i \chi-h}\, + d^\pp_{2,2}\,  \, e^{ i \chi-h}\, (x_1-i x_2) ,
$$ 
with $d^\pp_{2,1}$ and $d^\pp_{2,2}$ chosen so that $\LL \psip_2,  \vp\, \psip_2\RR =-1$ and $\LL \psip_1,  \vp\, \psip_2\RR =0$. 
Continuing in this way we obtain a family of zero eigenfunctions of $P_\pp(\A)$ such that
\begin{equation} \label{ef-hat-plus} 
\psip_j=  \sum_{k=1}^j \, d^\pp_{j,k} \,  (x_1- i x_2)^{k-1}\, e^{i\chi-h}, \qquad j=1,\dots, n, \qquad \LL \psip_j,  \vp\, \psip_i\RR = -\delta_{i,j}\, .
\end{equation}

\subsection{\bf Non-integer flux} 
 \label{sec-alpha-p}
First let us construct the zero resonant state. Similarly as in \eqref{phi-eq},\eqref{phi-2} ,  we put 
\begin{equation*}
\Phip_j = (x_1-i x_2)^{1-[\alpha]-j}\, e^{i\chi-h }\ , \qquad j=1,\dots, 1-[\alpha], 
\end{equation*}
and
\begin{equation*}
\phi = c_0 \,   \Psip , \qquad \text{with} \quad \Psip = \Phip_1  + \sum_{j=1}^n\, \LL  \Phip_1\, ,\,  \vp\,\psip_j\RR\, \psip_j\, ,
\end{equation*}
with $c_0$ such that $ \LL  \phi\, ,\,  \vp\,  \phi\, \RR\ = -1$.  We adopt analogous notation as for $\alpha >0$ and denote 
\begin{equation} \label{varphi-plus}
{\rm X}_\pp= \pom\,   \vp G_3  \vp \qquad \text{and} \qquad \varphi^\pp  = (1- {\rm X}_\pp) \, \Phi_1^\pp\, .
\end{equation}
Furthermore, we define 
\begin{equation} \label{psi-plus}
\psip =  \pop\,   \vp\, \psi_n^\pp \, .
\end{equation}
Finally, we let 
$$
 \eta_\pp  := \LL \Phip_1, \,  \vp\, G_3  \vp\, \Phip_1\, \RR\, ,  \qquad  \upsilon_\pp: =  \LL\, \Psip  ,\,  \vp G_3  \vp\, \pop  \, \vp G_3  \vp\,  \Psip \, \RR ,
 $$
 and
 \begin{equation} \label{constants-p}
\mathfrak{c}_\pp:= \eta_\pp-\upsilon_\pp, \qquad   \nu^\pp_n :=   \frac{ |d^\pp_{n}|^2}{\zeta(1+\alp)}  \, , \qquad  \mu^\pp_n =  \nu^\pp_n\, (T^{-1})_{nn}\, . \\[2pt]
\end{equation}
 
 \begin{thm} \label{thm-pauli-res-m}
Let $\alpha < -1,\,  \alpha\not\in\Z$. 
Suppose that $B$ satisfies  \ref{ass-B} with $\rho > 7$.  Then as $\lambda\to 0$, 
\begin{align} \label{res-pauli-alpham}
R_\pp(\lambda,\A)  & = -\lambda^{-1} \, \pop  - \frac{\nup_n\, \lambda^{-\alp-1}}{1+\mup_n\, \lambda^{-\alp}}\  \psip   \LL   \ \cdot\ , \psip \RR \, 
- \frac{ \zeta(\alp)\, \lambda^{\alp}}{1 +\cc_\pp\, \zeta(\alp)\, \lambda^{1+\alp}}\  \varphi^\pp  \LL   \ \cdot\ , \varphi^\pp \RR \,   +\mathcal{O}\big(1 \big)
\end{align} 
holds in $\B(-1,s;1,-s)$ for all $s>3$, where $\pop $ is the orthogonal projection on the zero eigenspace of $\Pp(\A)$ .
\end{thm}

As before, we provide a more precise expansion in the absence of zero eigenfunctions. 

\begin{cor} \label{cor-pauli-res2}
Suppose that $B$ satisfies Assumption \ref{ass-B} with $\rho > 7$.  Let $-1<\alpha<0$.  Then as $\lambda\to 0$, 
\begin{align*} 
R_\pp(\lambda,\A)  & = 
- \frac{ \zeta(\alpha)\, \lambda^{\alpha}}{1 +\eta_\pp\, \zeta(\alpha)\, \lambda^{1+\alpha}}\  \Phip_1 \,  \LL \,\cdot\, , \Phip_1 \RR\ + K_0^\pp + \mathcal{O}\big(|\lambda|^\mu \big) 
\end{align*} 
holds in $\B(-1,s;1,-s)$ for all $s>3$, where 
\begin{equation} \label{k-plus} 
K_0^\pp  = \Omega_0 G_0-\frac{1}{4\pi \alpha} \big (\, \Phip_1\, \lef\, \cdot \, , \, \Omega_0 \, g_0\rig +\Omega_0\,  g_0\, \lef\, \cdot\, , \,  \Phip_1\rig\, \big ) - \frac{\lef g_0,  \vp\, \Omega_0\, g_0 \rig +4\pi\alpha \delta_0+16\pi^2 \alpha^2}{16\pi^2 \alpha^2} \  \LL \,\cdot\, , \Phip_1 \RR\, \Phip_1 \, .
\end{equation}
\end{cor}

\smallskip

\begin{prop} 
Under the assumptions of Theorem \ref{thm-pauli-res-m}, as $\lambda\to 0$, 
\begin{align*}
R_\m(\lambda,\A)  & = (1+ G_0 \vm )^{-1}\, G_0  + (1+ G_0 \vm)^{-1} F_1(\alpha)\,  (1+ \vm G_0)^{-1}\, \lambda^{\mu} 
+o( \lambda^{\mu} )
\end{align*} 
holds in $\B(-1,s;1,-s)$ for all $s>3$, where $ F_1(\alpha)$ is defined in Proposition \ref{prop-pauli-regular}. 
\end{prop}

\subsection{\bf Integer flux} For $\alpha$ integer and negative we put 
\begin{align*} 
\Psip_2  & = \Phip_2 + \sum_{j=1}^n\,  \LL \, \Phip_2,  \vp\, \psip_j\RR\, \psip_j , \qquad \kappa_\pp  = \LL \Psip_2 ,\,  \vp\, \Phip_1 \RR, \\[-5pt]
\Psip_1 &= \Phip_1 +\kappa_\pp\, \Psip_2 + \sum_{j=1}^n\,  \LL \, \Phip_1,  \vp\, \psip_j\RR\, \psip_j  .
\end{align*} 
Moreover, similarly as in Section \ref{sec-exp-int} we define $\X_\pp= \pop \,  \vp\, \G_3\,  \vp$, and
\begin{equation*} 
\varphi_2^\pp = (1-\X_\pp) \Phi_2^\pp, \qquad \varphi_1^\pp = (1-\X_\m) \Phi_1^\pp + \kappa_\pp\, \varphi_2^\pp, \qquad 
\Pi^\pp_{jk} = \LL\, \cdot\, , \varphi_j^\pp\RR\,   \varphi_k^\pp,  \quad  j,k=1,2.
\end{equation*}

\begin{thm}  \label{thm-pauli-int-m}
Let $\alpha < 0,\, \alpha\in\Z$. Suppose that $B$ satisfies  \ref{ass-B} with $\rho > 7$. 
Then as $\lambda\to 0$, 
\begin{equation}  \label{res-pauli-int-m}
R_\pp(\lambda,\A)   = -\lambda^{-1} \, \pop   +  \frac{\Pi_{22}^\pp }{\pi\,  \lambda\, (\log\lambda-\w_\pp)} \,  - \log\lambda\,   \K_\pp + \mathcal{O}(1)
\end{equation} 
holds in $\B(-1,s;1,-s)$ for all $s>3$, where 
$$
\w_\pp =  i\pi +m_\pp,  \qquad \text{with} \qquad m_\pp:= \pi^{-1} \LL\,   \vp \Psip_2,\, \widetilde \G_3\,  \vp \Psip_2 \RR \, \in \R\, ,
$$
and where
\begin{align*}
 \K_\pp & = \frac{1}{4\pi} \big[ \,\Pi_{11}^\pp +\overline{ \varkappa}\, \Pi_{12}^\pp + \varkappa\, \Pi_{21}^\pp + |\varkappa|^2 \, \Pi_{22}^\pp\big] + 
\frac{\pi\, |d^\pp_n|^2 }{4} \  \psip  \LL \, \cdot \, ,  \psip \RR   \,  
\end{align*}
with $\varkappa= \kappa_\pp+\vartheta$. If $ \alpha = -1$, then equation \eqref{res-pauli-int-m} holds with $\pop=0$. 
\end{thm}

\smallskip

\begin{prop} 
Under the assumptions of Theorem \ref{thm-pauli-int-m}, as $\lambda\to 0$, 
\begin{align} \label{res-pauli-reg-int}
R_\m(\lambda,\A)  & = (1+ \G_0 \vm )^{-1}\, \G_0  + (1+ \G_0 \, \vm)^{-1} \G_1\,  (1+ \vm\, \G_0)^{-1}\, \frac{1}{\log\lambda}  
+o\big((\log\lambda)^{-1} \big)
\end{align} 
holds in $\B(-1,s;1,-s)$ for all $s>3$.  
\end{prop}

\subsection{The case $\alpha=0$}  
\label{ssec-zero-flux} 
In this case we cannot  apply the perturbation argument as it was done for $\alpha\neq 0$. Indeed, for $\alpha=0$ 
we have $H_0=-\Delta$, and the expansion of $R_0(\lambda)$ is thus singular as $\lambda\to 0$.  On the other hand the fact that $\alpha=0$ in combination with 
Proposition \ref{prop-gauge} shows that the coefficients of $H_0- P_\ppm(\A) = -\Delta -P_\ppm(\A)$ decay arbitrarily fast for $\rho$ large enough.
With the help of  \cite{mu} we then obtain

\begin{cor} \label{cor-pauli-res-0}
Let $\alpha=0$. 
Suppose that $B$ satisfies Assumption \ref{ass-B} with $\rho >7$.   Then the expansions
\begin{equation} \label{r-zero-flux}
R_\pm(\lambda, \A)  =  - \frac{\log\lambda}{4\pi}   \, \LL \, \cdot\, ,\,  e^{\mp h+i\chi} \RR\, e^{\mp h+i\chi} \, +\, \mathcal{O} (1)
\end{equation}
hold in $\B(-1,s;1,-s), \, s>3,$ as $\lambda\to 0$. 
\end{cor}

The proof of Corollary \ref{cor-pauli-res-0} will be given at the end of Section \ref{sec-weak}.


\section{\bf Dirac operator}  
\label{sec-dirac}
It was already mentioned in Section \ref{sec-intro} that the results obtained for the resolvent of the Pauli operator can be applied to analyze the resolvent of the Dirac operator as well. In view of 
\eqref{pauli-dirac} we find that 
\begin{equation} \label{pauli-dirac-res}
(D_m(\A) -\lambda)^{-1}  =  (D_m(\A) +\lambda)  \begin{pmatrix}
R_\m(\lambda^2 -m^2, \A) & 0\\
0 &    R_\pp(\lambda^2 -m^2, \A) 
\end{pmatrix}  . 
\end{equation}
Note that the right hand side of \eqref{pauli-dirac-res} is well defined and belongs to $\B(0,s;0,-s)$ for all $s>3$.

We can thus deduce the expansion of  $(D_m(\A) -\lambda)^{-1}$ for $\lambda\to\pm m$ directly from the expansions $R_\pp(\lambda, , \A)$ and $R_\m(\lambda , \A)$ for $\lambda\to 0$. In doing so it will be important to notice that, by  \eqref{gauge-pauli} and \eqref{gauge-transf}, 
\begin{equation}  \label{D-modes}
\D(\A)  \big(e^{i\chi-h}\, v\big) = -e^{i\chi-h}\,  (i\pd_1 +\pd_2) v,  \qquad \text{and} \qquad  \D(\A)^*\big(  e^{i\chi+h}\, v \big)= -e^{i\chi+h}\,  (i\pd_1 -\pd_2) v\, .
\end{equation} 

Throughout this section we assume that $B$ satisfies  \ref{ass-B} with some $\rho > 7$.

\subsection{Massive Dirac operator} The Dirac operator has, for $m>0$, two thresholds of the essential spectrum; $\pm m$. It turns out that the resolvent expansions at these thresholds are different. 

\begin{thm} \label{thm-dirac-1} 
Let $m>0$.  \par

\noindent If $\underline{\alpha >0}$, then
\begin{equation} \label{dirac-1} 
\begin{aligned} 
(D_m(\A) -\lambda)^{-1}  & =  \mathcal{O}(1) \qquad\qquad\qquad\qquad\qquad\qquad\qquad\quad\, \text{as} \quad\lambda\to -m, \\[4pt]
(D_m(\A) -\lambda)^{-1}  & =  
\begin{pmatrix}
 2m R_\m(\lambda^2-m^2, \A)   & 0\\
0 &  0
\end{pmatrix} 
+ \mathcal{O}(1) , \qquad \text{as} \quad \lambda\to m ,
\end{aligned}
\end{equation} 
with $R_\m(\cdot\, , \A)$ satisfying \eqref{res-pauli} for $\alpha\not\in\Z$, and \eqref{res-pauli-int} for $\alpha\in\Z$. 

\vskip0.2cm

\noindent If $\underline{\alpha <0}$, then
\begin{equation} \label{dirac-2} 
\begin{aligned} 
(D_m(\A) -\lambda)^{-1}  & =  
\begin{pmatrix}
0 & 0\\
0 &      -2m R_\pp(\lambda^2-m^2, \A)
\end{pmatrix} 
+ \mathcal{O}(1) ,\, \qquad \text{as} \quad \lambda\to -m, \\[4pt] 
(D_m(\A) -\lambda)^{-1}  & =  \mathcal{O}(1) \qquad\qquad\qquad\qquad\qquad\qquad\qquad\qquad  \, \text{as} \quad\lambda\to m , 
\end{aligned}
\end{equation} 
with $R_\pp(\cdot\, , \A)$ satisfying \eqref{res-pauli-alpham} for $\alpha\not\in\Z$, and \eqref{res-pauli-int-m} for $\alpha\in\Z$. 

\vskip0.3cm

\noindent If $\underline{\alpha =0}$, then as  $\lambda\to \pm m$,
\begin{align}  \label{dirac-3} 
(D_m(\A) -\lambda)^{-1}  & =  
\frac{m \log(\lambda \mp m)}{4\pi} 
\begin{pmatrix}
 e^{i\chi+h}\, \LL \, \cdot\, , e^{i\chi+h} \RR & 0\\
0 &     -e^{i\chi-h}\,  \LL \, \cdot\, , e^{i\chi-h} \RR   
\end{pmatrix} + \mathcal{O}(1)  \ .
\end{align}
 All the expansions above are to be understood in $\B(0,s;0,-s)$ for $s>3$.
\end{thm}

\begin{proof}
From Lemma \ref{lem-ah-cash}, equation \eqref{D-modes} and the threshold expansions of  $R_\pp(\lambda\, , \A)$ and $R_\m(\lambda\, ,\A)$ obtained in Sections \ref{sec-exp-non-int}, \ref{sec-exp-int} and \ref{sec-alpha-neg}
we deduce that 
\begin{equation*} 
\D(\A) R_\pp(\lambda^2-m^2, \A) = \mathcal{O}(1) \qquad \text{and} \qquad \D(\A)^* R_\m(\lambda^2-m^2, \A) = \mathcal{O}(1)  \qquad  \text{as} \ \ \lambda\to\pm m 
\end{equation*}
hold in $\B(0,s;0,-s)$ for any $\alpha\in\R$. Hence by \eqref{pauli-dirac-res}
\begin{equation}  \label{r-dirac-gen}
(D_m(\A) -\lambda)^{-1}  =    
\begin{pmatrix}
(\lambda+m)\, R_\m(\lambda^2 -m^2, \A) & \mathcal{O}(1)  \\
 \mathcal{O}(1)  &    (\lambda-m)\, R_\pp(\lambda^2 -m^2, \A)
\end{pmatrix}  .
\end{equation} 
Since $R_\pp(\lambda^2 -m^2, \A) = \mathcal{O}(1)$ as $\lambda\to \pm m$ when $\alpha>0$, and  $R_\m(\lambda^2 -m^2, \A) = \mathcal{O}(1)$ as $\lambda\to \pm m$ when $\alpha<0$, this proves equations \eqref{dirac-1} and \eqref{dirac-2}. Expansion \eqref{dirac-3} follows from \eqref{r-dirac-gen} and \eqref{r-zero-flux}.
\end{proof}
 
 \begin{rem} 
Equations \eqref{dirac-1} and \eqref{dirac-2} reflect, together with other things, the well known fact that when $\alpha>0$, then only $m$ (and not $-m$) is a zero-mode of $D_m(\A)$, and vice-versa for $\alpha<0$, see e.g.~\cite[Sec.~7]{th2}. 

\noindent Furthermore, if $0$ is an eigenvalue of $P_\m(\A)$ respectively $P_\pp(\A)$, then $m$ respectively $-m$ is an eigenvalue of $D_m(\A)$  with the same eigenspace. For example,  if $1 < \alpha$, then the leading term on the right hand side of \eqref{dirac-1} is given by 
$$
2mR_\m(\lambda^2-m^2,\A)  = -(\lambda-m)^{-1}\, \pom  + o\big( (\lambda-m)^{-1}\big)  \qquad \lambda \to m,
$$
see \eqref{res-pauli}, \eqref{res-pauli-int}. In general we notice that the resolvent expansion of the Dirac operator for $\lambda \to \pm m$ is qualitatively equivalent to the resolvent expansion of the Pauli operator for $\lambda\to 0$. The situation changes when $m=0$.
 \end{rem}

\subsection{Massless Dirac operator} For $m=0$ we have 

\begin{cor} \label{cor-dirac}
Let $m=0$ and assume that $\alpha \neq 0$. Then\\[1pt]
\begin{equation*} 
\begin{aligned} 
(D_0(\A) -\lambda)^{-1}  & =  
 \begin{pmatrix}
\Theta(\alpha)\, \lambda\,  R_\m(\lambda^2, \A)   & 0\\
0 &  \Theta(-\alpha)\, \lambda \,  R_\pp(\lambda^2, \A)
\end{pmatrix} 
+ \mathcal{O}(1) , \qquad \text{as} \quad \lambda\to 0 ,\\[5pt]
\end{aligned}
\end{equation*} 
holds in $\B(0,s;0,-s)$ for $s>3$, with $R_\m(\cdot\, , \A)$ satisfying \eqref{res-pauli} for $\alpha\not\in\Z$, and \eqref{res-pauli-int} for $\alpha\in\Z$. Here $\Theta(\cdot)$ denotes the Heaviside function. 
\end{cor}

\begin{proof} 
The claim follows  from \eqref{r-dirac-gen} upon setting $m=0$. 
\end{proof}

As in the case of Pauli operator we give a separate expansion for values of $\alpha$ in the interval $(-1, 1)$.

\begin{thm} \label{thm-dirac-2} 
If $|\alpha| <1$, then as $\lambda\to 0$,  
\begin{align}  \label{dirac-5} 
(D_0(\A) -\lambda)^{-1}  & =  
\begin{pmatrix}
\frac{- \zeta(\alpha)\, \lambda^{1-2\alpha}}{1 +\eta_\m \zeta(\alpha)\, \lambda^{2-2\alpha}}\  \Phim_1 \,  \LL \,\cdot\, , \Phim_1 \RR  & \D(\A) (1+G_0 \vp)^{-1} G_0 \\[6pt]
\D(\A)^* K_0^\m& 0      
\end{pmatrix} + \mathcal{O}(|\lambda|^\mu) ,  \quad \text{if} \ \ \alpha >0, 
\end{align}
and 
\begin{align}  \label{dirac-6} 
(D_0(\A) -\lambda)^{-1}  & =  
\begin{pmatrix}
0  & \D(\A) K_0^\pp \\[6pt]
\D(\A)^*(1+G_0 \vm)^{-1} G_0   & \frac{ -\zeta(\alpha)\, \lambda^{1+2\alpha}}{1 +\eta_\pp \zeta(\alpha)\, \lambda^{2+2\alpha}}\  \Phip_1 \,  \LL \,\cdot\, , \Phip_1 \RR
\end{pmatrix} + \mathcal{O}(|\lambda|^\mu) ,  \quad \text{if} \ \ \alpha <0, 
\end{align}

hold in $\B(0,s;0,-s)$ for $s>3$. The operators $K_0^\pp$ and $K_0^\m$ are given by \eqref{k-hat} and \eqref{k-plus}. 

If $\alpha =0$, then $(D_0(\A) -\lambda)^{-1} = \mathcal{O}(1)$.  
\end{thm}

\vskip0.2cm

\begin{proof} 
Assume that $0 <\alpha <1$. From \eqref{res-pauli-2} and \eqref{D-modes} we obtain 
$$
\D(\A)^* R_\m(\lambda^2,\A) = \D(\A)^* K_0^\m + \mathcal{O}(|\lambda|^\mu),  \quad \text{and} \quad   \lambda R_\m(\lambda^2,\A)  = \frac{- \zeta(\alpha)\, \lambda^{1-2\alpha}}{1 +\eta_\m\, \zeta(\alpha)\, \lambda^{2-2\alpha}}\  \Phim_1 \,  \LL \,\cdot\, , \Phim_1 \RR  + \mathcal{O}(\lambda)
$$
in $ \B(0,s;0,-s)$, as $\lambda\to 0$. On the other hand, by \eqref{res-pauli-reg}
$$
\D(\A) R_\pp(\lambda^2,\A) =  \D(\A) (1+G_0 \vp)^{-1} G_0 + \mathcal{O}(|\lambda|^\mu),   \quad \text{and} \quad  \lambda R_\pp(\lambda^2,\A) =  \mathcal{O}(\lambda).
$$
Equation \eqref{dirac-5} thus follows from \eqref{pauli-dirac-res}. The proof of \eqref{dirac-6} and of the case $\alpha=0$ is completely analogous.
\end{proof}

Theorem \ref{thm-dirac-2} shows that  the resolvent expansion of the massless Dirac operator is regular if $|\alpha| \leq 1/2,$ although the resolvent expansion of the Pauli operator is singular.  Similar effect occurs in the absence of a magnetic field in the case of the free massless Dirac operator  in dimension two, see e.g.~\cite{egg}.


\section{\bf Time decay of the wave-functions}
\label{sec-tdecay} 
In this section apply the resolvent expansions for Pauli and Dirac operators obtained above to derive asymptotic equations for the evolution operators 
$e^{-it P_\pm(\A)}$ and $e^{-it D_m(\A)}$ for $t\to\infty$. These asymptotic expansions  show how fast the solutions to the time dependent Pauli and Dirac equation decay  {\it locally} in time.

\subsection{The Pauli operator} 
\label{ssec-pauli-time} Before stating the main results of this section we recall some preliminaries which will be often used below. First, 
if $F:\R \to \B(0,s;0-s )$ is such that $F(\lambda)=0$ in a vicinity of zero and $\partial^j F \in L^1(\R; \B(0,s;0-s ))$ for a non-negative integer $j$, then 
\begin{equation} \label{eq-jk}
\int_\R e^{-i t \lambda}\,F(\lambda)\, d\lambda = o(\, t^{-j}) \qquad t\to \infty
\end{equation}
in $ \B(0,s;0,-s )$. Equation \eqref{eq-jk} is a consequence of the Riemann-Lebesgue Lemma, we refer to \cite[Lem.~10.1]{JK} for a proof. 
We will also need a couple of identities from the theory of Fourier transforms, namely
\begin{equation} \label{fourier-1}
 \int_\R e^{-i t \lambda}\, (\lambda+i0)^\nu\, (\log( \lambda+i0))^k d\lambda =  -2 \sum_{j=0}^k {k \choose j} \Big[ \partial_s^{k-j} \big (\sin(\pi s) \, e^{\frac{i\pi s}{2}} \Gamma(s+1) \big)\Big |_{s=\nu} \Big] \, t^{-\nu-1}\, (\log t)^j , 
\end{equation}
which holds for $t>0, \, \nu\in\R$ and any non-negative integer $k$, and
\begin{equation} \label{fourier-2}
 \int_\R e^{-i t \lambda}\, (\lambda+i0)^{-1} \, (\log(\lambda+i0))^{-k}\, d\lambda = 2\pi i\, 
 \sum_{m=k}^N  A_{km}\,  (\log t)^{-m} + \mathcal{O}\big(( \log t)^{-N-1}\big)
\end{equation}
which hols for any $t>0, \ N>0$, and $ k\in\Z$, see \cite[Lems.~6.6, 6.7]{mu}. Here $A_{km}$ are numerical coefficients satisfying $A_{kk}=(-1)^{k+1}$, see  \cite[Thm.~4.4]{mu}.

\medskip

\begin{thm} \label{thm-pauli-time-1}
Let $\alpha\not\in\Z$. Assume that $B\in C^\infty(\R^2)$ satisfies Assumption \ref{ass-B} with some $\rho > 7$. Suppose moreover that  for any multi-index $\beta\in\N^2$ with $|\beta|\geq 1$, 
\begin{equation*}
| \pd^\beta B(x)| \ \lesssim\ \x ^{-1-|\beta|}\, .
\end{equation*} 
 The following asymptotic equations hold in $\B(0,s;0,-s ), \, s>3,$ as $t\to\infty$: 

 If $\underline{\alpha >0}$, then $ e^{-i t P_\pp(\A)}= \mathcal{O}(t^{-1-\mu})$, and 
\begin{align}
e^{-i t P_\m(\A)}   & = \pom  -\frac{\nu_\m}{\pi\omega_\m} \   \psim  \LL  \ \cdot\ , \psim \RR \,  \sum_{j=1}^N\,  (-\omega_\m)^j \sin(\pi\alp j)\, e^{\frac{i\pi j\alp}{2}} \, \Gamma(\alp j)\  t^{-j\alp }   \nonumber \\
& \quad  +  \frac{1}{\pi\cc_\m} \,    \varphi^\m  \LL  \ \cdot\ , \varphi^\m \RR  \sum_{j=1}^N\, ( i\cc_\m \zeta(\alp) )^j \sin(\pi\alp j)\, e^{-\frac{i\pi j\alp}{2}} \, \Gamma( j(1-\alp)) \  t^{j(\alp-1)} \label{pauli-time-1} \\
& \quad +\mathcal{O} \big(t^{-\mu (N+1)}\big)  + o(t^{-1}) \nonumber 
\end{align} 
for any $N\geq 1$.

If $\underline{\alpha <0}$, then $ e^{-i t P_\m(\A)}= \mathcal{O}(t^{-1-\mu})$, and 
\begin{align}
 e^{-i t P_\pp(\A)}  & = \pop + \frac{ \nu_n^\pp}{\pi\mu_n^\pp} \   \psip  \LL   \ \cdot\ , \psip \RR \,  \sum_{j=1}^N\,  (-\mu_n^\pp)^j \sin(\pi\alp j)\, e^{-\frac{i\pi j\alp}{2}} \Gamma(-\alp j)\  t^{j\alp }  \nonumber  \\
& \quad  - \frac{1}{\pi\cc_\pp} \,    \varphi^\pp  \LL  \ \cdot\ , \varphi^\pp \RR   \sum_{j=1}^N\, ( i\cc_\pp \zeta(\alp) )^j \sin(\pi\alp j)\, e^{\frac{i\pi j \alp}{2}} \, \Gamma( j(1+\alp)) \  t^{-j(\alp+1)}\label{pauli-time-2} \\
& \quad  +\mathcal{O} \big(t^{-\mu (N+1)}\big)  
+ o(t^{-1})  \nonumber  
\end{align} 
for any $N\geq 1$.  Recall that $\alp$ and $n$ denote the fractional and integer parts of $\alpha$, and that $\mu = \min\{|\alp|, 1-|\alp|\}$.
\end{thm}

\begin{rem} 
The leading terms of the second and third contribution on the right hand side of \eqref{pauli-time-1} are proportional to $t^{-\alp}$ and $t^{-1+\alp}$ (and accordingly for $\alpha<0$). In particular, if $|\alpha|<1$, then $e^{-i t P_\ppm(\A)} =  \mathcal{O}(t^{-1+|\alpha|})$. The same decay rate was observed  for the heat kernel generated by $P_\ppm(\A)$ in the case of radial and compactly supported magnetic field, see \cite[Sec.~3.2]{kov1}.

 Note also that the remainder term in \eqref{pauli-time-1} and \eqref{pauli-time-2} becomes  $o(t^{-1})$ as soon as $N$ gets large enough, depending on $\alpha$. For example, if 
$\alp = 1/2$, then combining \eqref{zeta}, \eqref{constants-m} and equation \eqref{pauli-time-1} with $N=1$ gives 
$$
e^{-i t P_\m(\A)}   = \pom   -2 (i \pi\, t)^{-\frac 12}\ |d_n^\m|^2 \,  \psim  \LL  \ \cdot\ , \psim \RR  -\frac 12 (i \pi )^{-\frac 32}\,  t^{-\frac 12} \,    \varphi^\m  \LL  \ \cdot\ , \varphi^\m \RR  + o(t^{-1}) .
$$
\end{rem}

\begin{proof}[\bf Proof of Theorem 	\ref{thm-pauli-time-1}]  Let $s>3$. We will use the identity
\begin{equation}  \label{propagator}
e^{-i t P_\ppm(\A)}  = \frac{1}{2\pi i}\!\! \int_\R e^{-it\lambda}\, R_\ppm(\lambda,\A)\, d\lambda
\end{equation}
 in $\B(0,s;0,-s )$. 
Hence we have to study the behaviour of $R_\ppm(\lambda,\A)$ not only for $\lambda\to 0$, but also for $|\lambda|\to \infty$. To do so we introduce a function $\xi\in C_0^\infty(\R)$ such that $0\leq \xi \leq 1$ and $\xi =1$ in a vicinity of $0$. 
Since $B\in C^\infty(\R^2)$, we can apply \cite[Thm.~5.1]{ro}  to deduce that 
\begin{equation} \label{lap-2}
\|\partial_\lambda^{j} R_\ppm(\lambda,\A) \|_{\B(0,s;0,-s)} \  =  \mathcal{O}\big(\lambda^{-\frac{j+1}{2}}\big)\, ,  \qquad  \lambda\to+\infty 
\end{equation}
for any $j\geq 0$ and any $s>j +\frac12$. On the other hand, the operators $P_\ppm(\A)$ have no negative spectrum. Hence 
$$
\|\partial_\lambda^{j} R_\ppm(\lambda,\A) \|_{\B(0,s;0,-s)} \ = \mathcal{O}\big(\lambda^{-(j+1)}\big)\, ,  \qquad  \lambda\to-\infty \, .
$$ 
 Applying \eqref{eq-jk} with $F(\lambda) = (1-\xi(\lambda))\,  R_\ppm(\lambda,\A)$ and $j=2$ then implies 
\begin{equation} \label{hep}
 \int_\R e^{-it\lambda}\, (1-\xi(\lambda))\, R_\ppm(\lambda,\A)\, d\lambda =  o(\, t^{-2}) \qquad t\to \infty
\end{equation}
in $\B(0,s;0,-s )$. Assume first that $\alpha>0$. Expanding the second and third term on the right hand side of \eqref{res-pauli} into the geometric series yields
\begin{align*}
\frac{ \lambda^{\alp-1}}{1+\omega_\m\, \lambda^\alp}  & =  -\frac{1}{\omega_\m} \sum_{j=1}^N \, (-\omega_\m)^j\, \lambda^{j \alp-1} + \mathcal{O} \big(\lambda^{\alp(N+1)-1}\big) \, , \\
\frac{ \lambda^{-\alp}}{1+\cc_\m \zeta(\alp) \, \lambda^{1-\alp}} & = -\frac{1}{\cc_\m \zeta(\alp)}  \sum_{j=1}^N\,  (-\cc_\m \zeta(\alp))^j\, \lambda^{j(1- \alp)-1} + \mathcal{O} \big(\lambda^{(1 -\alp)(N+1)}\big)\, .
\end{align*}
From \eqref{res-pauli}, \eqref{eq-jk} and \eqref{fourier-1} we thus deduce that  
$$
\frac{1}{2\pi i} \int_\R e^{-it\lambda}\, \xi(\lambda)\, R_\m(\lambda,\A)\, d\lambda =\  \text{right hand side of \ \eqref{pauli-time-1}}\, .
$$
Since $P_\pp(\A)$ has no positive eigenvalues and no zero modes, \cite[Thm.~2.5]{ko} implies $ e^{-i t P_\pp(\A)}= \mathcal{O}(t^{-1-\mu})$. The proof for $\alpha<0$ follows the same lines.
\end{proof}  


\vskip0.2cm

\begin{thm} \label{thm-pauli-time-2}
Let $\alpha\in\Z$. Let $B$ satisfy assumptions of Theorem \ref{thm-pauli-time-1}. The following asymptotic equations hold in $\B(0,s;0,-s ), \, s>3,$ as $t\to\infty$: 

If $\underline{\alpha >0}$, then $ e^{-i t P_\pp(\A)}= \mathcal{O}(t^{-1} (\log t)^{-2})$, and 
\begin{align}
  e^{-i t P_\m(\A)}  & = \pom + \frac{\Pi_{22}^\m \,}{\pi \log t}\, \big(1+\mathcal{O}((\log t)^{-1})\big) - i \, \K_\m\ t^{-1}  + o(t^{-1}) \label{pauli-time-3} .
\end{align}

 If $\underline{\alpha <0}$, then $ e^{-i t P_\m(\A)}=\mathcal{O}(t^{-1} (\log t)^{-2})$, and 
\begin{align}
  e^{-i t P_\pp(\A)}  & = \pop +  \frac{\Pi_{22}^\pp \,}{\pi \log t}\,   \big(1+\mathcal{O}((\log t)^{-1})\big) -i  \, \K_\pp\ t^{-1}
+ o(t^{-1}) .\label{pauli-time-4}
\end{align} 

 If $\underline{\alpha =0}$, then
\begin{align}
 e^{-i t P_\ppm(\A)}  & = \frac{1}{4\pi i\, t}  \, \LL \, \cdot\, , e^{\mp h+i\chi} \RR\, e^{\mp h+i\chi} + o(t^{-1}) .
\label{pauli-time-5}
\end{align} 
\end{thm}

\begin{proof} 
Suppose that $\alpha>0$. We proceed as in the proof of Theorem \ref{thm-pauli-time-1} and use again the smooth cut-off function $\xi$. Then \eqref{hep} still holds true, and from \eqref{fourier-1}, \eqref{fourier-2} and \eqref{res-pauli-int} we deduce that 
$$
\frac{1}{2\pi i} \int_\R e^{-it\lambda}\, \xi(\lambda)\, R_\m(\lambda,\A)\, d\lambda =  \pom + \frac{\Pi_{22}^\pp}{\pi \log t}\,  \big(1+\mathcal{O}((\log t)^{-1})\big) - i \, \K_\m\ t^{-1}  + o(t^{-1}) \, .
$$
This proves \eqref{pauli-time-3}. Equations \eqref{pauli-time-4} and \eqref{pauli-time-5} follow from \eqref{res-pauli-int-m} and \eqref{r-zero-flux} in the same way.
\end{proof}

\begin{rem} 
The regularity assumption on $B$ in Theorems \ref{thm-pauli-time-1} and \ref{thm-pauli-time-2} could be considerably relaxed. The only place in the proof which requires $B\in C^\infty(\R^2)$ is equation \eqref{lap-2}, cf.~ \cite{ro}. The latter, however, could be derived also from the commutator method developed in \cite{jmp} under much weaker regularity condition on $B$. Since this would require additional analysis involving lengthly calculations of multiple commutators, we don't  dwell on it.
\end{rem}

\begin{rem}
Equation \eqref{pauli-time-5} shows that when $\alpha=0$, then the long time behavior of $e^{-i t P_\pp(\A)}$ and $e^{-i t P_\m(\A)}$ is qualitatively the same as that of the free evolution operator $ e^{i t \Delta}$. This is in contrast with the case of a spinless particle, where the wave-functions decay faster than $\mathcal{O}(t^{-1})$ even if the flux is zero, see \cite[Thm.~2.7]{ko}.
\end{rem}

\subsection{The Dirac operator}  
\label{ssec-dirac-time}
When dealing with the time dependent Dirac equation one inevitably faces the following problem;  neither the weighted resolvent $(D_m(\A) -\lambda)^{-1}$ nor its derivatives w.r.t.~$\lambda$ vanish as $|\lambda|\to\infty$, see equations \eqref{pauli-dirac-res} and \eqref{lap-2}. This naturally suggests to treat the contributions to the time evolution separately from small and high energies. Since the contribution from high energies is typically, to the leading order, independent of the magnetic field, we will concentrate on the energies close to zero. To this end we consider  the decay in time of the vector valued function 
\begin{equation} \label{dirac-evolution}
u(t) : = 
\begin{pmatrix}
u_1(t) \\
u_2(t)    
\end{pmatrix} 
= e^{-it D_m(\A)} \, \xi(D_m(\A)) 
\, v ,
\qquad 
v= \begin{pmatrix}
v_1 \\
v_2  
\end{pmatrix} ,
\end{equation}
where $v_1, v_2\in \Lp^{2,s}(\R^2)$ for a suitable $s$, and where $\xi$ is the cut-off function introduced above. 

If $m>0$ and $\alpha\not\in\Z$, then one easily verifies that  $u(t)$ decays qualitatively in the same way as the wave-functions of the Pauli operator. Indeed, denoting
$$
\po = \begin{pmatrix}
\Theta(\alpha) \, \pom\\
\Theta(-\alpha)\, \pop
\end{pmatrix} ,
$$
and mimicking the proof of Theorems \ref{thm-pauli-time-1}  we find, in view of \eqref{dirac-1} and \eqref{dirac-2}, that 
\begin{equation*} 
u(t) - \po\, v  = \mathcal{O} \big(t^{-|\alp|}\big)  +\mathcal{O} \big(t^{|\alp|-1}\big)  + o(t^{-1}) \qquad  (m>0, \, \alpha\not\in\Z)\, 
\end{equation*}
in $\Lp^{2,-s}(\R^2;\C^2)$. Analogous result holds for the integer flux.

We therefore turn our attention to the massles Dirac operator.

\begin{thm} \label{thm-dirac-time-1}
Suppose that $\alpha\not\in\Z$ and that $m=0$. Let $B$ satisfy assumptions of Theorem \ref{thm-pauli-time-1}, and let  $u(t)$ be given by \eqref{dirac-evolution}. The following asymptotic equations hold in $\Lp^{2, -s}(\R^2), \, s>3,$ as $t\to\infty$: 

 If $\underline{\alpha >0}$, then $ u_2(t) = o(t^{-1})$, and 
\begin{align}
 u_1(t) -\pom u_1  & = - \frac{\nu_\m}{\omega_\m \pi} \, \psim\,   \LL    v_1 , \psim \RR \,  \sum_{j=1}^{N}\,  (-\omega_\m)^j \sin(2\pi\alp j)\, e^{i\pi j\alp} \, \Gamma(2 j\alp )\ t^{-2j\alp }   \nonumber \\
& \quad - \frac{1}{\pi\cc_\m} \,   \varphi^\m  \LL   v_1 , \varphi^\m \RR   \sum_{j=1}^N\, ( \cc_\m \zeta(\alp) )^j \sin(2\pi\alp j)\, e^{i\pi j\alp} \, \Gamma(2j(1-\alp)) \  t^{2j(\alp-1)}  \nonumber \\
& \quad +\mathcal{O} \big(t^{-2\mu (N+1)}\big)  + o(t^{-1}) \label{dirac-time-1}
\end{align} 
for any $N\geq 1$.

If $\underline{\alpha <0}$,  then $ u_1(t) = o(t^{-1})$, and 
\begin{align*}
 u_2(t) -\pop u_2  & = - \frac{\nu_n^\pp}{\mu_n^\pp \pi} \  \psip\,   \LL    v_2 , \psip \RR \,  \sum_{j=1}^{N}\,  (-\mu_n^\pp)^j \sin(2\pi\alp j)\, e^{i\pi j\alp} \, \Gamma(2\alp j)\ t^{-2j\alp }   \nonumber \\
& \quad - \frac{1}{\pi\cc_\pp} \,   \varphi^\pp  \LL   v_2 , \varphi^\pp \RR   \sum_{j=1}^N\, ( \cc_\pp \zeta(\alp) )^j \sin(2\pi\alp j)\, e^{i\pi j\alp} \, \Gamma(2j(1-\alp)) \  t^{2j(\alp-1)}  \nonumber \\
& \quad +\mathcal{O} \big(t^{-2\mu (N+1)}\big)  + o(t^{-1}) 
\end{align*} 
for any $N\geq 1$. 
\end{thm}

\begin{proof} 
By \eqref{dirac-evolution} we have
\begin{equation} \label{u(t)}
u(t) = \frac{1}{2\pi i} \int_\R e^{-it \lambda}\, \xi(\lambda)\, (D_0(\A)-\lambda)^{-1}\, v \
d\lambda\, .
\end{equation} 
To find the expansion of the integral  for $t\to\infty$ we note that if $\alpha>0$, then in view of \eqref{res-pauli} and \eqref{res-pauli-reg}\\[1pt]
$$
\lambda R_\m(\lambda^2,\A) =  -\lambda^{-1} \, \pom  - \frac{\nu_\m\, \lambda^{2\alp-1}}{1+\omega_\m\, \lambda^{2\alp}}\  \psim  \LL  \, \cdot\ , \psim \RR \, 
-\zeta(\alp)\, \lambda^{1-2\alp} \frac{  \varphi^\m  \LL  \ \cdot\ , \varphi^\m \RR  }{1 +\cc_\m\, \zeta(\alp)\, \lambda^{2-2\alp}}  +\mathcal{O}\big(1 \big)\\[4pt]
$$
and $\lambda R_\pp(\lambda^2,\A)=\mathcal{O}(\lambda)$  as $\lambda\to 0$. Now we proceed as in the proof of Theorem \ref{thm-pauli-time-1} and expand the second and third term on the right hand side into a geometric series. 
Equation \eqref{dirac-time-1} then follows from Corollary  \ref{cor-dirac} and equation \eqref{fourier-1}. We again omit the proof for $\alpha<0$.
\end{proof}


\begin{cor}\label{cor-dirac-time}
Let $u(t)$ be given by \eqref{dirac-evolution}, and assume that $|\alpha|\leq 1/2$.  Under the assumptions of Theorem \ref{thm-dirac-time-1} the following expansions hold n $\Lp^{2, -s}(\R^2)$ as $t\to\infty$; 
\begin{align*}
u_1(t) &=  - \frac{\zeta(\alpha)}{\pi} \ t^{2\alpha-2}\  \Phim_1 \,  \LL \, v_1\, , \Phim_1 \RR  \, \sin(2\alpha \pi) \, e^{i\pi\alpha} \, \Gamma(2(1-\alpha))  + \mathcal{O} \big(t^{-1-\alpha}\big), \quad \quad u_2(t)=\mathcal{O} \big(t^{-1-\alpha}\big)  \quad \,\text{if} \ \ \alpha>0 , \\[3pt]
u_2(t) &=  - \frac{\zeta(\alpha)}{\pi} \ t^{-2\alpha-2}\  \Phip_1 \,  \LL \, v_2\, , \Phip_1 \RR  \, \sin(2\alpha \pi) \, e^{-i\pi\alpha} \, \Gamma(2(1+\alpha)) + \mathcal{O} \big(t^{-1+\alpha}\big) , \quad u_1(t)= \mathcal{O} \big(t^{-1+\alpha}\big) \quad \text{if} \ \ \alpha<0.
\end{align*}
Moreover, if $\alpha=0$, then $u(t)= o(t^{-1})$ in $\Lp^{2, -s}(\R^2;\C^2)$.
\end{cor} 

\begin{proof}
Under the hypotheses of the Corollary we have $\mu =|\alpha|$.  
Assume first that $0<\alpha \leq 1/2$. Then by \eqref{dirac-5} 
$$ 
(D_0(\A) -\lambda)^{-1}   =  
\begin{pmatrix}
- \zeta(\alpha)\, \lambda^{1-2\alpha}\,  \Phim_1 \,  \LL \,\cdot\, , \Phim_1 \RR  & \D(\A) (1+G_0 \vp)^{-1} G_0 \\[6pt]
\D(\A)^* K_0^\m& 0      
\end{pmatrix} + \mathcal{O}(|\lambda|^\alpha)  \, . \\[5pt]
$$
The claim thus follows from equations \eqref{eq-jk}, \eqref{fourier-1} and \eqref{u(t)}. The case $-1/2\leq \alpha\leq 0$ is analogous.
\end{proof}

\begin{rem} 
It follows from Corollary \ref{cor-dirac-time} that for small but non-zero values of flux the wave-functions decay even faster than in the absence of a magnetic field. Similar diamagnetic effect 
 was observed in \cite{cf} for the massless Dirac operator with Aharonov-Bohm magnetic field, see also \cite{avs}. The authors of \cite{cf} showed that the range of local smoothing estimates widens, w.r.t.~ the non-magnetic operator, as soon as $\alpha\not\in\Z$. Recall that the condition $|\alpha|\leq 1/2$ represents no restriction in case of the Aharonov-Bohm field.
\end{rem}

\begin{rem} 
Dispersive estimates in $\Lp^1-\Lp^\infty$ setting for non-magnetic Dirac operators in $\R^2$ perturbed by a matrix-valued electric potential were established in \cite{eg2,egg}.
\end{rem}

\section{\bf Weakly coupled eigenvalues of the Pauli operator} 
\label{sec-weak}
\noindent  In this last section we apply the resolvent expansions obtained in Sections \ref{sec-exp-non-int} and \ref{sec-exp-int}  the problem of asymptotic behavior of eigenvalues of the operators $P_\pm(A)  - \eps V$ where $V:\R^2\to\R$ is an electric potential with a suitable decay at infinity.  For radial magnetic  and electric field this has been already studied, with variational methods, in \cite{bcez, fmv}. Here we extend some of these results to non-radial $B$ and/or $V$. 
Since this application is not the central topic of the paper, we limit ourselves to the case $0< \alpha  <  1$.

\begin{prop} \label{prop-weak-1}
Let $0<\alpha<1,$ and let $V \lesssim \langle \, \cdot\, \rangle ^{-\rho}$ for some $\rho >6$.  Suppose that $B$ satisfies Assumption \ref{ass-B} with $\rho>7$.  Then the following assertions hold for any $A\in \Lp_{\rm \, loc}^2(\R^2)$ such that $\rt A =B$.

\begin{enumerate}
\item 
The negative spectrum of $P_\pp(A)  - \eps V$ is empty for $|\eps|$ small enough. 



\medskip

\item  If $\int_{\R^2} V |\Phim_1|^2 >0$, then the operator  $\Pm(A)  - \eps V$ has for $\eps>0$ small enough exactly one negative eigenvalue which satisfies 
\begin{equation} \label{E-eps}
\Lambda(\eps) =  -\left(\frac{4^{\alpha-1}\, \Gamma(\alpha)}{ \pi\, \Gamma(1-\alpha)}\ \int_{\R^2} V |\Phim_1|^2 dx \right)^{\frac 1\alpha}\, \eps^{\frac 1\alpha}\, \big(1+ \mathcal{O}\big(\eps^{\min\{1, \frac 1\alpha -1\}}\big)\big) ,\qquad \eps\to 0+ .
\end{equation}
\end{enumerate}
\end{prop}


\begin{proof} 
Since the spectra of  $P_\pm (A)-\eps V$ are gauge invariant, it suffices to prove the statement for $A=\A$. From the assumptions on $V$ it follows that 
$$
\sigma_{\rm es} \big(P_\pm(\A)-\eps V\big)  = [0,\infty), 
$$
and that the operator $|V|^{1/2}\, (\Pm(\A) +z)^{-1}\, V^{1/2}$ is compact on $\Lp^2(\R^2)$ for every $z>0$.  Here we use the notation $V^{1/2} = |V|^{1/2} \, \mathrm{sign} (V)$. 
Hence the negative parts of the spectra of $P_\pm(\A )-\eps V$ are purely discrete and the Birman-Schwinger principle  implies that 
\begin{equation} \label{bs-princ}
 -z \in \sigma_{\rm d}\big(P_\pm(\A )-\eps V\big)  \quad \Longleftrightarrow \quad  1 \in \sigma \big ( \eps\, |V|^{1/2}\, (P_\pm(\A ) +z))^{-1}\, V^{1/2}\big).
\end{equation}
To prove $(1)$ we note that by the assumptions on $V$ and by Proposition \ref{prop-pauli-regular}
\begin{equation} \label{ko-eq}
|V|^{1/2} \, (P_\pp(\A ) +z)^{-1}\, V^{1/2} = \mathcal{O}(1),  \qquad z\to 0
\end{equation}
holds in $\Lp^2(\R^2)$. In view of  \eqref{bs-princ} this implies that $P_\pp(\A )  - \eps V$ has no negative eigenvalues for $|\eps|$ sufficiently small.

Let us prove $(2)$. Since $\alpha \in (0,1)$, we have $\alp= \alpha, \, n=0,$ and $\Psim = \Phim_1$, see equation \eqref{phi-2}. Put $\lambda = -z, \, z>0$. From \eqref{zeta} we get
$$
 \zeta(\alpha)\, \lambda^{-\alpha} = - \frac{4^{\alpha-1} \, \Gamma(\alpha)}{\pi \, \Gamma(1-\alpha)}\  z^{-\alpha},
$$
where we have used identity \eqref{g-z}. For the purpose of this proof let us introduce the following terminology: given a self-adjoint operator $H$ we denote
$$
N_{[a,b]} (H) = \# \Big\{{\rm eigenvalues\  of  }\  H\ in \ [a,b]\Big\} , 
$$
where the  eigenvalues are counted with their multiplicities. The Birman-Schwinger principle then ensures that 
\begin{equation}  \label{bs-princ-2}
N_{(-\infty, z]}(P_\m(\A)  - \eps V)  =  N_{[1,\infty)} (\eps\, |V|^{1/2}\, (P_\m(\A) +z)^{-1}\, V^{1/2}) .
\end{equation}
Since $|V|^{1/2}\, (P_\m(\A) +z)^{-1}\, V^{1/2}$ is a Hilbert-Schmidt operator on $\Lp^2(\R^2)$, see Remark \ref{rem-hs}, and since the first term on the right hand side of \eqref{res-pauli-2} is a rank one operator, it follows from \eqref{bs-princ-2}
\begin{align*}
N_{(-\infty, z]}(P_\m(\A)  - \eps V)  &  \ \leq \ 1+ N_{[1,\infty)} (\eps\, |V|^{1/2}\, \mathcal{S}(z)\, V^{1/2})   \leq 1+ \eps^2 \, \| |V|^{1/2}\,  \mathcal{S}(z)\, V^{1/2}\|^2_{\rm hs},
\end{align*} 
where 
$$
S(z) = (P_\m(\A) +z)^{-1} + \frac{ \zeta(\alpha)\, (-z)^{-\alpha}}{1 +\eta_\m\, \zeta(\alpha)\, (-z)^{1-\alpha}}\  \Phim_1 \,  \LL \,\cdot\, , \Phim_1 \RR\, .
$$
Since $\| |V|^{1/2}\,  \mathcal{S}(z)\, V^{1/2}\|_{\rm hs} = \mathcal{O}(1)$ as $z\to 0$, see \eqref{res-pauli-2}, we infer that for $\eps$ small enough $P_\m(\A)  - \eps V$ has at most one 
negative eigenvalue. On the other hand, the hypothesis $\int_{\R^2} V |\Phim_1|^2 >0$ and a simple test function argument show that $P_\m(\A)  - \eps V$ has for any $\eps>0$ at least one negative eigenvalue. Hence for $\eps\to 0+$ there is exactly one negative eigenvalue $\Lambda(\eps)=-z(\eps)$. By \eqref{bs-princ}  $z(\eps)$ is the unique value of $z$ for which the operator
\begin{equation} \label{K-eps}
K_\eps(z) = 1 - \eps\, |V|^{1/2}\, (\Pm(\A ) +z)^{-1}\, V^{1/2}
\end{equation} 
is {\em not invertible} in $\Lp^2(\R^2)$. In order to locate this value of $z$ we let  
$$
\p u =\V^{-1} \,  \lef u, V^{1/2}\, \Phim_1\rig \, |V|^{1/2}\, \Phim_1, \qquad \p_0 =1 -\p, \qquad \text{with} \quad \V := \int_{\R^2} V |\Phim_1|^2  \neq 0\, ,
$$
so that $\p \p_0=0$. Moreover, we define operators $I : \C\to \Lp^2(\R^2) \, $ and $I^*:\Lp^2(\R^2)\, \to \C$ as follows;
$$
Iw  = \V^{-1/2}\, |V|^{1/2}\, \Phim_1\, w,  \qquad  I^* u  =  |\V|^{-1/2}\,  \LL u,  V^{1/2}\, \Phim_1 \RR ,
$$
where $\V^{-1/2} = |\V|^{-1/2}\, \mathrm{sign\, } \V$. 
Then $II^* = \p$ and $I^* I= \id$, and by Theroem \ref{thm-pauli-res2},  
\begin{equation}  \label{weak-eq1}
|V|^{1/2} (\Pm(\A ) +z)^{-1} V^{1/2}=  \mathscr{F}(z)\, \V\, \p  + \mathcal{O}(1), \qquad z\to 0,
\end{equation}
where
\begin{equation} \label{omega-z}
 \mathscr{F}(z)  = \frac{4^{\alpha-1} \, \Gamma(\alpha)}{\pi \, \Gamma(1-\alpha)}\  z^{-\alpha}\, \big(1+\mathcal{O}\big( z^{1-\alpha} \big)\big) .
\end{equation}
Hence $\p_0  K_\eps(z)  \p_0$ is boundedly  invertible on $\p_0 \Lp^2(\R^2)$ uniformly in $z\in (-\delta, \delta)$ for some $\delta>0$.
Now we apply the Grushin-Schur method to the operator $K_\eps(z)$ in $\Lp^2(\R^2)$ in the same way as it was done in Section \ref{ssec-grushin} with the operator $M(\lambda)$ in $\HH^{-1,s}$, replacing the operators $Q, Q_0, J$ and $J^*$ by $\p,\p_0, I$ and $I^*$ respectively. In particular, by mimicking equations \eqref{a-lam}, \eqref{a12} and \eqref{E} we find that $K_\eps(z)$ is {\em not invertible} in 
$\Lp^2(\R^2)$ if and only if $z$ satisfies
\begin{equation} \label{eq-z0}
I^* K_\eps(z) I= I^* K_\eps(z) \p_0 \big( \p_0 \, K_\eps(z) \, \p_0\big)^{-1} \p_0 K_\eps(z) I .
\end{equation}
Standard perturbation arguments  show that   such $z$ tends to zero as $\eps \to 0$. 
Since $\p_0 I= I^* \p_0 = 0$, it follows from \eqref{weak-eq1} that $I^* K_\eps(z) \p_0 \big( \p_0 \, K_\eps(z) \, \p_0\big)^{-1} \p_0 K_\eps(z) I= \mathcal{O}(\eps^2)$. 
In view of $I^* \p I= \id$  and equation \eqref{weak-eq1} we thus conclude that  equation \eqref{eq-z0} is equivalent to 
\begin{equation} \label{eq-z}
\eps \, \V\, \mathscr{F}(z) = 1+ \mathcal{O}(\eps)\, .
\end{equation}
By assumption, $\V >0$. We thus conclude from \eqref{omega-z} 
that for $\eps>0$ small enough the solution to \eqref{eq-z}  satisfies
$$
z= \left(\frac{4^{\alpha-1}\, \Gamma(\alpha)}{ \pi\, \Gamma(1-\alpha)}\ \V \right)^{\frac 1\alpha} \eps^{\frac 1\alpha}\, \big(1+ \mathcal{O}\big(\eps^{\min\{1, \frac 1\alpha -1\}}\big)\big) .
$$
By \eqref{bs-princ} we have $\Lambda(\eps) = -z$, and thus the claim. 
\end{proof}

\begin{rems} \label{rems-weak}
Some comments concerning Proposition \ref{prop-weak-1}  are in place. 

\begin{enumerate}
\item 
Equation \eqref{E-eps} coincides, including the order of the remainder term, with the asymptotic expansion obtained by variational arguments in  \cite[Thm.~1.2]{fmv} for radial $B$ and radial non-negative $V$. Notice that with the notation of \cite{fmv} one has $|\Phim_1|^2 = 2\pi |\Omega_N^{-}|^2$.


\item  When both $B$ and $V$ are radial, $V\geq 0$, and $1 < \alpha \not\in\Z$, then Corollary \ref{cor-radial} in combination with a straightforward modification of the proof of Proposition \ref{prop-weak-1} shows that in addition  to $\Lambda(\eps)$ the operator $\Pm(A)  - \eps V$ has, for $\eps>0$ small enough,  other negative eigenvalues $\Lambda_j(\eps)$ satisfying 
$$
\Lambda_j(\eps) = -\frac{\int_{\R^2} V |\varphi_j|^2}{\|\varphi_j\|_2^2} \, \eps\, \big(1+ \mathcal{O}(\eps)\big), \quad 1\leq j\leq n-1, \qquad 
\Lambda_n(\eps) = -\frac{\int_{\R^2} V |\varphi_n|^2}{\|\varphi_n\|_2^2} \, \eps \, \big(1+ \mathcal{O}(\eps^{\alp})\big), 
$$
where $\varphi_j = (x_1+i x_2)^{j-1}\, e^{h}$. The lower order of the remainder in the expansion of $\Lambda_n(\eps)$ comes from the second term on the right hand side of \eqref{R-eq-radial}. This is in agreement with \cite[Thm.~1.2]{fmv}, and moreover it shows that the order of the remainder term in \cite[Eq.~(1.11)]{fmv} cannot be improved.

\item The same approach as in Proposition \ref{prop-weak-1} can be used to calculate the leading terms of the asymptotics of weakly coupled eigenvalues of $P_\pm(A)  - \eps V$ for 
$\alpha\geq 1$. This will be done elsewhere. 

\end{enumerate}
\end{rems}

\begin{proof}[\bf Proof of Corollary \ref{cor-pauli-res-0}]
For the sake of definiteness we consider only the  resolvent of $P_\m(\A)$.
Proposition \ref{prop-gauge} still holds true which shows that the perturbation operator $ \vm$ satisfies, for $\rho$ large enough,  assumptions of \cite[Thm.~8.1]{mu}.   
We may thus apply \cite[Thm.~4.3.(vi)]{mu}. Note that when $\alpha=0$, then by Lemma \ref{lem-ah-cash} the operator $P_\m(\A)$ has one zero resonant state, namely $e^{h+i\chi}$, and no zero eigenfunctions. Hence
in the notation of \cite{mu} we have $\mu=0, P=0, Q=0, Q_k=0, k\neq -1$, and $Q_{-1} = C_0\,  \LL \, \cdot\, ,\,  e^{h+i\chi} \RR\, e^{ h+i\chi}$, where $C_0$ is a positive constant. From \cite[Thm.~4.3.(vi)]{mu} we then deduce that 
\begin{equation} \label{murata}
R_\m(\lambda, \A)  =  - C_0\, \log\lambda \, \LL \, \cdot\, ,\,  e^{h+i\chi} \RR\, e^{h+i\chi} \, +\, \mathcal{O} (1), \qquad \lambda\to 0.
\end{equation} 
It remains to determine the value of $C_0$. To do so we consider the discrete spectrum of the perturbed operator $P_\m(\A) -\eps V$ with $\eps >0$ and $V\geq 0$. Suppose moreover, that $V$ is radial and compactly supported. Then  in the same way as in the proof of Proposition \ref{prop-weak-1} it follows from \eqref{murata} that for $\eps$ small enough 
the operator $P_\m(A)  - \eps V$ has exactly one negative eigenvalue $\Lambda_\m(\eps)$, and that 
\begin{align} \label{Lambda-zero-flux}
\Lambda_\m(\eps) & = - \exp\left[-\frac{1}{\eps \, C_0  \V} \Big(1+\mathcal{O}(\eps) \Big)\right]  \qquad\eps\to 0+ ,
\end{align}
where 
$$
\V= \int_{\R^2} V(x)\,  e^{ 2h(x)}\,  dx. 
$$
On the other hand, for radial and compactly supported $B$ we know from \cite[Thm.~1.3]{fmv}, see also  \cite{bcez},  that 
$$
\Lambda_\m(\eps) = - \exp\left[ -\frac{4\pi}{\eps   \V} \Big(1+\mathcal{O}(\eps) \Big)\right] \qquad\eps\to 0+ ,
$$
A comparison with $\eqref{Lambda-zero-flux}$ gives $C_0=\frac{1}{4\pi}$ and hence the claim.
\end{proof}

\appendix

\section{\bf Auxiliary results  } 
\label{sec-app-aux} 
In this section we collect some technical results which will be needed in the proofs of Propositions \ref{prop-exp} and \ref{prop-exp-int}. We denote by $\|\cdot \|_{\hs(M)}$ the Hilbert-Schmidt norm of an operator in $\Lp^2(M)$.

\begin{lem} \label{lem-pm12}
Equation \eqref{pm12-estim} holds for all $m\in\Z$. Moreover,  for $|\lambda|$ small enough
\begin{align} 
|\sigma_m(\lambda)|  +|\sigma'_m(\lambda)| & \, \lesssim\, \frac{1}{\Gamma^2(|m-\alpha|)} \label{sigma-m-2} \\
|q_m|\,   & \lesssim \ \frac{4^{-|m-\alpha|}}{(1+m^2)\, \Gamma^2(|m-\alpha|)} \label{cm-upperb}
\end{align}
\end{lem}

\begin{proof}
Note that 
\begin{equation} \label{pm1-zero}
\mathscr{P}_{m,1}(0)= -\frac{1}{\gamma_m\, \Gamma(1-|m-\alpha|)} \, ,
\end{equation}
with $\gamma_m$ given by \eqref{delta-m}. Hence
\begin{equation*}
q_m = -\frac{4^{-|m-\alpha|} a_m\, \gamma_m\, \Gamma(1-|m-\alpha|)}{2 |m-\alpha|(1+|m-\alpha|)\, \Gamma(|m-\alpha|)}\, .
\end{equation*}
Since $|a_m \gamma_m| \lesssim 1$, inequality  \eqref{cm-upperb} follows from \eqref{g-z}.
\smallskip

\noindent Now let $\nu >0, \, \nu\not\in\Z$. Then by \eqref{f-nu} 
\begin{equation} \label{f-m-nu}
|f_\nu(z)| + |f'_\nu(z)|\,  \lesssim \,  \frac{e^{|z|/4}}{\Gamma(1+\nu)} \, , \qquad \forall\, z\in\R.
\end{equation}
On the other hand, using \eqref{g-z} we find that 
$$
(-1)^k\, \Gamma(k+1-\nu) = \frac{\pi\, \sec(\nu\pi)\,}{ \Gamma(\nu-k)} \qquad k=0,\dots ,[\nu]\, .
$$
Hence
\begin{align} \label{f-nu-exp}
f_{-\nu}(z) & = \frac{\sin(\nu\pi)\, \Gamma(\nu)}{\pi} \left[1+\frac{z}{4(\nu-1)} +\dots +\frac{z^k}{4^k\, k! (\nu-1)(\nu-2)\cdots (\nu-[\nu])}\right]+ \sum_{k=[\nu]+1}^\infty \, \frac{ (-z)^{k}}{4^k\, k!\ \Gamma(k+1-\nu)} \nonumber\\
& =  \frac{\sin(\nu\pi)\, \Gamma(\nu)}{\pi} \left[1+\frac{z}{4(\nu-1)} +\dots +\frac{z^k}{4^k\, k! (\nu-1)(\nu-2)\cdots (\nu-[\nu])}\right] + \mathcal{O}\big(4^{-\nu}\, |z|^\nu\, e^{|z|/4}\big)
\end{align} 
This in combination with \eqref{f-m-nu} and the definition of $P_{m,j}(z), j=1,2,$ see Section \ref{ssec-G-nint}, gives
\begin{equation*}
\frac{\mathscr{P}_{m,2}(\lambda)}{\mathscr{P}_{m,1}(\lambda)}  = -2\sigma_m(\lambda) -2^{1+2|m-\alpha|}\, q_m\, \lambda +\mathcal{O}\Big(  \frac{\lambda^2}{\Gamma(|m-\alpha|) \, \Gamma(|m-\alpha|+1) }\Big). 
\end{equation*}
Equation \eqref{pm12-estim} now follows in view of \eqref{g-z}. Next, equation \eqref{f-nu-exp} which in combination with \eqref{g-z} implies that
\begin{equation} \label{sigma-exp} 
\sigma_m(\lambda) =\frac{v'_m(\lambda,1)-v_m(\lambda,1)\, |m-\alpha| }{v'_m(\lambda,1)+v_m(\lambda,1)\, |m-\alpha| }\  \frac{ \sec(\pi |m-\alpha|)}{\Gamma(|m-\alpha|)}\ f_{|m-\alpha|}(\lambda)\,(1+\mathcal{O}(\lambda)), \qquad \lambda\to 0,
\end{equation} 
To continue we recall  that the hypergeometric function $M$ is given by 
\begin{equation*} 
M(a,b,z) = \sum_{n=0}^\infty \, \frac{(a)_n}{(b)_n}\ \frac{z^n}{n!},
\end{equation*}
where 
$(a)_n = a(a+1)\cdots(a+n-1),  \  (b)_n = b(b+1)\cdots(b+n-1)$ and  $a_0=b_0=1$.
Hence 
\begin{equation*} 
1 \leq M(a,b,z) \leq \, e^{|z|} \qquad 0< a\leq b, \quad  \ z\in\C
\end{equation*}
Since $ \pd_zM(a,b,z) = \frac{a}{b} M(a+1,b+1,z)$, see \cite[Sec.~13]{as}, it follows that for $|\lambda|$ small enough
\begin{align}
| \partial_\lambda^{(j)} v_m(\lambda,r)| \  & \lesssim \ (1+|m|^j)\, (2\, \alpha\, r)^{|m|} , \qquad j=0,1,2, \label{vm-estim-1} \\
 |v'_m(\lambda,r)| +| \partial_\lambda v'_m(\lambda,r)| \  & \lesssim \ (1+|m|)\, (2\, \alpha\, r)^{|m|-1}  \label{vm-estim-2}
\end{align} 
Moreover, by \cite[Eq.~(5.31)]{ko}, 
\begin{equation} \label{gamma-upperb}
\gamma_m \lesssim \, (1+|m|)^{-1}\, (2\alpha)^{-|m|},  \\[3pt]
\end{equation} 
Inserting the above bounds together with  \eqref{f-m-nu} into equation \eqref{sigma-exp} yields \eqref{sigma-m-2}.
\end{proof}
Note that by \eqref{Q-nu}
\begin{equation} \label{QQ-eq}
\begin{aligned}
 Q_{|m-\alpha|}(\lambda, \cdot)\,  & = 2^{|m-\alpha|}\,  \lambda^{-|m-\alpha|/2} \, J_{|m-\alpha|}(\sqrt{\lambda}\ \cdot)  \\
 Q_{-|m-\alpha|}(\lambda, \cdot) & =  2^{-|m-\alpha|}\, \lambda^{|m-\alpha|/2} \, J_{-|m-\alpha|}(\sqrt{\lambda}\ \cdot) \, .
 \end{aligned}
\end{equation}
On the other hand, by \cite[Lem.~7.1]{ko} and \cite[Eq.~(7.14)]{ko},
\begin{equation} \label{JJ}
\begin{aligned}
\big\|\, \x^{-s}\,  J_{|m-\alpha|}(\sqrt{\lambda}\  r) \, J_{|m-\alpha|}(\sqrt{\lambda}\ t) \, \x^{-s}\big\|_{\hs((1,\infty)^2, r dr)} &   \lesssim\,   \frac{\lambda^{|m-\alpha|}}{1+m^2} \\
\big\|\, \x^{-s}\,  J_{\pm |m-\alpha|}(\sqrt{\lambda}\  r) \, J_{\mp |m-\alpha|}(\sqrt{\lambda}\  t) \, \x^{-s}\big\|_{\hs((1,\infty)^2, r dr)}  & \ \lesssim\  \frac{1}{1+|m|} \\
\big\|\, \x^{-s}\,  J_{-|m-\alpha|}(\sqrt{\lambda}\  r) \, J_{-|m-\alpha|}(\sqrt{\lambda}\  t) \, \x^{-s}\big\|_{\hs((1,\infty)^2, r dr)}  & \ \lesssim  \ \lambda^{-|m-\alpha|} \ \frac{\Gamma^2(|m-\alpha|)}{1+|m|} 
\end{aligned}
\end{equation}
hold for any $s >5/2$. Hence 
\begin{equation} \label{QQ-bound}
\begin{aligned}
\|  \x^{-s}\, Q_{|m-\alpha|}(\lambda,r)\,  Q_{|m-\alpha|}(\lambda,t) \, \x^{-s} \|_{\hs((1,\infty)^2, r dr)} & = \mathcal{O}\big( (1+m^2)^{-1}\big) \\
\| \x^{-s}\, Q_{-|m-\alpha|}(\lambda,r) \, Q_{-|m-\alpha|}(\lambda,t) \, \x^{-s} \|_{\hs((1,\infty)^2, r dr)}  & = \mathcal{O}\big(4^{-|m-\alpha|} \, (1+|m|)^{-1} \Gamma^2(|m-\alpha|)\big) \ \\
\| \x^{-s}\,  Q_{\pm|m-\alpha|}(\lambda,r) \, Q_{\mp|m-\alpha|}(\lambda,t)\, \x^{-s} \|_{\hs((1,\infty)^2, r dr)} & =  \mathcal{O}\big( (1+|m|)^{-1}\big)
\end{aligned}
\end{equation}
for all $s>5/2$ and $\lambda$ small enough.

\medskip

\noindent In order to describe the asymptotic behaviour of $R_{m,0}(\lambda; r,t )$ in the region $r < t\leq 1,$ it will be useful to introduce the operators $k_m(\lambda)$ and $\ell_m(\lambda)$ with kernels
\begin{equation}\label{k-kappa}
\begin{aligned}
k_m(\lambda;r,t) & = \frac{ \Gamma(\frac 12+|m|-\frac{m\alpha}{\sqrt{\alpha^2-\lambda}})}{2\pi\, \Gamma(1+2 |m|)} \,\ v_m(\lambda,r) \, v_m(\lambda,t) \, , \\
\ell_m(\lambda;r,t) & = \frac{ \Gamma(\frac 12+|m|-\frac{m\alpha}{\sqrt{\alpha^2-\lambda}})}{2\pi\, \Gamma(1+2 |m|)} \,\  v_m(\lambda,r) \, u_m(\lambda,t) \, .
\end{aligned}
\end{equation}
where $v_m$ is given by \eqref{vm-tot}, and 
\begin{align*}
u_m(\lambda, t) &=  e^{-t\, \sqrt{\alpha^2-\lambda}}\, \big(2\, t \, \sqrt{\alpha^2-\lambda}\, \big)^{|m|}\, U\Big(\frac 12+|m|-\frac{m\alpha}{\sqrt{\alpha^2-\lambda}}, 1+2 |m|, 2 t\sqrt{\alpha^2-\lambda} \Big) \, .
\end{align*}
Moreover, let $k'_m(\lambda)$ and $\ell'_m(\lambda)$ be the operators with kernels $\partial_\lambda k_m(\lambda;r,t)$ and $\partial_\lambda \ell_m(\lambda;r,t)$ respectivelly. 

\begin{lem}  \label{lem-kl}
There exists $\delta>0$ such that 
\begin{align} 
\sup_{|\lambda| \leq \delta}\, \sup_{m\in\Z} \big(\,  \beta_m \|k_m(\lambda)\|_{\hs((0,1)^2, r dr)}+\|\ell_m(\lambda)\|_{\hs((0,1)^2, r dr)} \big)   & < \infty,  \label{sup-kl} \\
 \sup_{m\in\Z} \big(\, \beta_m \|k'_m(0)\|_{\hs((0,1)^2, r dr)}+\|\ell'_m(0)\|_{\hs((0,1)^2, r dr)} \big)   & < \infty  \label{sup-kl-der}, 
\end{align} 
with $\beta_m$ defined in \eqref{delta-m}. 
Moreover, as $\lambda\to 0$
\begin{equation} \label{kl-expand}
\begin{aligned} 
\beta_m \, k_m(\lambda) & = \beta_m \,k_m(0) +\lambda \,  \beta_m \, k'_m(0) + \mathcal{O}_{\Lp^2(0,1)}(\lambda^2) \\
\ell_m(\lambda) & = \ell_m(0) +\lambda \,  \ell'_m(0) + \mathcal{O}_{\Lp^2(0,1)}(\lambda^2)
\end{aligned}
\end{equation} 
in $\Lp^2(0,1)$, with the remainder terms uniform in $m$. 
\end{lem} 

\begin{proof} 
We will need the bounds, cf.~\cite[Eq.(5.31)]{ko} and \cite[Lemma~A.1]{ko}, 
\begin{equation} \label{beta-upperb}
\beta_m \   \lesssim\  \frac{(2\alpha)^{-2|m|}\, \Gamma(2|m|)}{(1+|m|)\, \Gamma(\frac 12 +2|m|)} , \qquad 
\Gamma(a)\, \left |\, U(a,b,z) \right| \  \leq  e^z\, z^{1-b}\,   \Gamma(b-1)   \quad  a\geq 1
\end{equation}
Since $\pd_zU(a,b,z) =-a U(a+1,b+1,z)$, see  \cite[Eq.13.4.21]{as}, it follows that 
\begin{equation}\label{um-estim}
  |\partial_\lambda^{(j)} u_m(\lambda,t)|  \ \lesssim \ (2 \alpha)^{-|m|} \, t^{1-j-|m|} \, \frac{ \Gamma(2|m|)}{\Gamma(j+\frac 12 +|m|-m)} \ (1+|m|^j), \qquad j=0,1,2.
\end{equation}
for $|\lambda|$ small enough. 
This  together with \eqref{vm-estim-1} proves \eqref{sup-kl}.  
\begin{align*}
\partial_\lambda k_m(0;r,t) &  =    \frac{ \Gamma(\frac 12+|m|-m)}{2\pi\, \Gamma(1+2 |m|)} \Big [\,  \partial_\lambda (v_m(\lambda,r)\, v_m(\lambda, t))\big |_{\lambda=0} 
 +\frac{m} {2 \alpha^2} \ \psi_0\big(|m|-m+1/2\big)\, v_m(0,r)\, v_m(0,t) \Big ] \\ 
 \partial_\lambda \ell_m(0;r,t) &  =    \frac{ \Gamma(\frac 12+|m|-m)}{2\pi\, \Gamma(1+2 |m|)} \Big [\,  \partial_\lambda (v_m(\lambda,r)\, u_m(\lambda, t))\big |_{\lambda=0} 
 +\frac{m} {2 \alpha^2} \ \psi_0\big(|m|-m+1/2\big)\, v_m(0,r)\, u_m(0,t) \Big ],
\end{align*}  
where
$$
\psi_0(z) = \int_0^\infty \big(e^{-t} -(1+t)^{-z} \big) \frac{dt}{t}   
$$
is the digamma function. From \eqref{vm-estim-1} and  \eqref{um-estim} we deduce that for any $j=0,1,2$, 
\begin{align}
\int_0^1  |\partial_\lambda^{(j)} v_m(\lambda,r) |^2 \, r\, dr  \   & \lesssim\ (1+m^2)^{j-1}\, |2\alpha|^{2|m|} \ \label{v-integ-estim}  \\
\int_0^1\int_0^1  |\partial_\lambda^{(j)} (v_m(\lambda,r)\, v_m(\lambda, t))  |^2 \, r t\, dr  dt \  & \lesssim\  (1+m^2)^{j-1}\, |2\alpha|^{4|m|}   \label{vv-integ-estim} \\
\int_0^1\int_r^1  |\partial_\lambda^{(j)} (v_m(\lambda,r)\, u_m(\lambda, t))  |^2 \, t\, dt  \, r\, dr  \  & \lesssim\    \frac{ (1+m^2)^{j-1}\,\, \Gamma^2(2 |m|)}{ \Gamma^2(\frac 12+|m|-m) }  \label{vu-integ-estim}
\end{align} 
uniformly in $\lambda$ for $|\lambda|$ small enough. Since $\psi_0(z) = \mathcal{O}(\log z)$, and $\psi_0'(z) = \mathcal{O}(1)$ as  $z\to\infty$,  this in combination \eqref{beta-upperb} implies \eqref{sup-kl-der} and \eqref{kl-expand}.
 \end{proof}

\section{\bf Operator $R_0(\lambda)$; proof of Proposition \ref{prop-exp}} 
\label{sec-app-b} 
The goal of this section is to give a proof of Proposition \ref{prop-exp}. Hence we suppose that $\alpha\not\in\Z$ and denote
\begin{align*}
X(\lambda) &=  R_0(\lambda) -G_0 - \lambda^{\alp} G_1^\lambda - \lambda^{1-\alp} G_2^\lambda 
 - \lambda G_3 -  \lambda^{1+\alp} G_4^\lambda -  \lambda^{1+2\alp}G_5^\lambda   -\lambda^{2-\alp} G_6^\lambda -\lambda^{3-2\alp} G_7^\lambda. 
\end{align*} 
In view of equation \eqref{B0-eq-1},  we have to show that 
\begin{equation} \label{enough-X}
 \|  \x^{-s}\, X(\lambda)\, \x^{-s} \|_{\Lp^2(\R^2)} = \mathcal{O}(\lambda^2)\, ,
\end{equation} 
for $s>3$. The integral kernel of $X(\lambda)$ can be written by the partial wave decomposition as follows
$$
X_m(x,y) = \sum_{m\in\Z} X_m(\lambda;r,t )\, e^{ im (\theta-\theta')}\, ,
$$
To prove \eqref{enough-X} it thus suffices to show that
\begin{equation*} 
\sup_{m\in\Z} \|  \x^{-s}\, X_m(\lambda)\, \x^{-s} \|_{\Lp^2(\R_+, r dr)} = \mathcal{O}(\lambda^2)\, ,
\end{equation*} 
where $X_m(\lambda)$ denotes the operator in $\Lp^2(\R_+, r dr)$ with integral kernel $X_m(\lambda;r,t)$.
This will be done in a series of  Lemmas below, in which we formulate upper bounds on the Hilbert-Schmidt of $X_m$ restricted to various parts of $\R^2$ according to the values of $r$ and $t$ . It is useful to define the functions 
\begin{equation} \label{j-m}
j_m(r) = -2 \sin(\pi |m-\alpha|)\,  \pd_\lambda\!\!\left( \frac{v_m(\lambda, r)}{ \mathscr{P}_{m,1}(\lambda)}\right)\Big|_{\lambda=0}\, \quad \text{if} \quad r<1, \qquad  j_m(r) = \partial_\lambda g_n(0,r) \quad \text{if} \ \ r\geq 1, 
\end{equation}
and $\varrho_m :(0,\infty)\to \R$ by 
\begin{equation} \label{rho-q}
\begin{aligned}
\varrho_m(r)  & = \frac{q_m h_m(r) +4^{-|m-\alpha|}\, j_m(r)}{4 \, \Gamma(1+|m-\alpha|)\, \sin(\pi |m-\alpha|)} \, .
\end{aligned}
\end{equation}

\subsubsection*{\bf Case $1<r<t $} In this region we deduce form  \eqref{gm-new}, \eqref{g-z} and \eqref{sigma-m},  that 
\begin{align} \label{gm3-1} 
G_{m,3}(r,t)  &= \frac{t^{-|m-\alpha|}}{16\pi} \left[ \frac{g_m(r)\, t^2}{ |m-\alpha|((m-\alpha)^2-1)} - \frac{2\delta_m\, r^{2+|m-\alpha|} }{(m-\alpha)^2-1}  +\frac{4\Gamma(|m-\alpha|)\, r^{-|m-\alpha|}\, }{\Gamma(1-|m-\alpha|)} \big( \sigma'_m(0) +4^{|m-\alpha|}\, q_m\big) \right] 
\end{align}
with $\delta_m$ given by \eqref{delta-m}.

\begin{lem} \label{lem-Tm}
Let $T_m(\lambda)$ be the integral operator on $\Lp^2((1,\infty)^2, r dr)$ with integral kernel given by \eqref{tm}. Then 
$$
\big\|\, \x^{-s}\,  T_m(\lambda) \, \x^{-s}\big\|_{\hs((1,\infty)^2, r dr)}\,  = \mathcal{O}(\lambda^2\, (1+|m|)^{-1}),
$$
for all $s> 3$. 
\end{lem}

\begin{proof}
From the definition of the operators $\Sigma_{m,j}(\lambda)$, see Section \eqref{ssec-G-nint}, and from equations \eqref{gm-new} and \eqref{QQ-bound} it follows that 
$$
\sum_{j=0}^3 \big\|\, \x^{-s}\,  \Sigma_{m,j}(\lambda) \, \x^{-s}\big\|_{\hs((1,\infty)^2, r dr)} =  \mathcal{O}\big( 4^{-|m-\alpha|} \big) +  \mathcal{O}\big( (1+|m|)^{-1}\big)\, .
$$
Here we have used also Lemma \ref{lem-pm12} to estimate $\sigma_m(\lambda)$ and $q_m$. On the other hand, by \eqref{rm-hat}, \eqref{eq-jy}  and \eqref{JJ} 
\begin{equation} \label{Rhat-estim}
 \big\|\, \x^{-s}\,  \widehat R_m(\lambda)  \, \x^{-s}\big\|_{\hs((1,\infty)^2, r dr)} =  \mathcal{O}\big(\lambda^{-|m-\alpha|} \, (1+|m|)^{-1} \Gamma^2(|m-\alpha|)\big) 
\end{equation}
Since $\epsilon_m=  \mathcal{O}\big(\Gamma^{-2}(|m-\alpha|)\big)$, see \eqref{pm12-estim} and \eqref{g-z}, the above estimates in combination with \eqref{tm} imply the claim. 
\end{proof}

Note also that in vieew of \eqref{gm3-1} and Lemma \ref{lem-pm12} 
\begin{equation} \label{hs-gm3-1}
 \big\|\, \x^{-s}\, G_{m,3} \, \x^{-s}\big\|_{\hs((1,\infty)^2, r dr)}  =  \mathcal{O}((1+|m|)^{-1}) \qquad \forall\, s>3. 
\end{equation}

\begin{lem} \label{lem-Zm-1}
Let $s>3$. Then
$$
\sup_{m\in\Z} \big\|\, \x^{-s}\, X_{m}(\lambda)  \, \x^{-s}\big\|_{\hs((1,\infty)^2, r dr)}  =  \mathcal{O}(\lambda^2)\, .
 $$
\end{lem}
\begin{proof} 
The functions $Q_{\pm\nu}(\lambda, \cdot )$  are power series in $\lambda$. With the help of Lemma \ref{lem-pm12} and equations \eqref{gm-new} and \eqref{QQ-bound} it is thus straightforward to verify that 
$$
\big\|\, \x^{-s}\,\big( \Sigma_{m,0}(\lambda)- \Sigma_{m,0}(0) -\lambda \pd_\lambda \Sigma_{m,0}(0)\big) \, \x^{-s}\big\|_{\hs((1,\infty)^2, r dr)}  =  \mathcal{O}(\lambda^2), 
$$
and
$$
\big\|\, \x^{-s}\,\big( \Sigma_{m,2}(\lambda)- \Sigma_{m,2}(0) \big) \, \x^{-s}\big\|_{\hs((1,\infty)^2, r dr)}  =  \mathcal{O}(\lambda). 
$$
Similarly it follows that 
\begin{align}
 \big\|\, \x^{-s}\, g_m(\lambda,r) g_m(\lambda,t)  \, \x^{-s}\big\|_{\hs((1,\infty)^2, r dr)}  & =   \mathcal{O}\big( (1+|m|)^{-1}\big),\label{sup-m} \\
 \big\|\, \x^{-s}\,  (g_m(\lambda,r) g_m(\lambda,t) -  g_m(0,r) g_m(0,t)) \, \x^{-s}\big\|_{\hs((1,\infty)^2, r dr)}  & = \mathcal{O}(\lambda)\, , \label{g-g0}
\end{align} 
Now the claim follows from \eqref{rm0-full} and Lemma \ref{lem-Tm}.
\end{proof}

\subsubsection*{\bf Case $r <1<t $}  In this region we have 
\begin{equation}\label{r-kernel-2}
R_{m,0}(\lambda; r,t ) = \frac{ v_m(\lambda,r) \big ( J_{|m-\alpha|}(\sqrt{\lambda}\, t ) +i \,Y_{|m-\alpha|}(\sqrt{\lambda}\, t )\big)}{2\pi (B_m(\lambda)\ -i A_m(\lambda))}\, ,
\end{equation} 
see \cite[Sec.5]{ko}. By \eqref{am-1} and \eqref{abm-w} 
\begin{equation*}  
\frac{1}{B_m(\lambda)\ -i A_m(\lambda)} = \frac{ 2^{-|m-\alpha|}  \lambda^{\frac{|m-\alpha|}{2}}\, \sin(\pi  |m-\alpha|)}{i\, \mathscr{P}_{m,1}(\lambda) \mathscr{V}_m(\lambda)} \, .
\end{equation*} 
Hence \eqref{r-kernel-2} can be rewritten as
\begin{align} \label{case2-1}
R_{m,0}(\lambda; r,t ) & = \frac{v_m(\lambda, r)}{2\pi \mathscr{P}_{m,1}(\lambda) \mathscr{V}_m(\lambda)} \Big[ 4^{-|m-\alpha|}  \lambda^{|m-\alpha|} \, Q_{|m-\alpha|}(\lambda,t) -Q_{-|m-\alpha|}(\lambda,t)\Big]\, .
\end{align}
By  \eqref{QQ-bound} and \eqref{pm1-zero}
\begin{equation} \label{kov15-aux}
\begin{aligned}
 \big\|\, \x^{-s}\, v_m(\lambda, r) \, 4^{-|m-\alpha|}  \lambda^{|m-\alpha|} \, Q_{|m-\alpha|}(\lambda,t)   \, \x^{-s}\big\|_{\hs((0,1)\times(1,\infty), r dr)} & =  \mathcal{O}(1) \\
\frac{1}{ \mathscr{P}_{m,1}(\lambda)}\  \big\|\, \x^{-s}\, v_m(\lambda, r) \, Q_{-|m-\alpha|}(\lambda,t)   \, \x^{-s}\big\|_{\hs((0,1)\times(1,\infty), r dr)} &  =  \mathcal{O}(1) 
\end{aligned} 
\end{equation} 
hold for all $s>5/2$ uniformly in $m$. 
Recall also that $\varrho_m$ is defined as a linear combination of $j_m$ and $h_m$ in \eqref{rho-q}. Using \eqref{kov15-aux} together with \eqref{sup-m} and \eqref{g-g0} we then infer from \eqref{wm-exp} that 
\begin{equation} \label{case2-2}
\begin{aligned}
R_{m,0}(\lambda; r,t ) & = G_{m,0}(r,t) + \lambda G_{m,3}(r,t) +  \frac{ \zeta^{-1}(|m-\alpha|)\, \lambda^{|m-\alpha|} \, g_m(r)\, g_m(t)}{16\pi^2 (m-\alpha)^2\,(1-p_m(\lambda)  \lambda^{|m-\alpha|})}  
\\ & \quad 
- \frac{ e^{-i |m-\alpha| \pi}\, \lambda^{1+|m-\alpha|}}{1-p_m(\lambda)  \lambda^{|m-\alpha|}} \, \big[ g_m(r) \,	\varrho_m(t) + g_m(t)\,  \varrho_m(r) \big]  
\\[4pt] & \quad 
- \frac{q_m \, e^{-i |m-\alpha| \pi}\,  \zeta^{-1}(|m-\alpha|)\, \lambda^{1+2|m-\alpha|}  }{16\pi^2 (m-\alpha)^2 (1-p_m(\lambda)  \lambda^{|m-\alpha|})^2}\, g_m(r) g_m(t) + \widetilde X_{m}(\lambda; r,t )\, 
\end{aligned}
\end{equation}
where the kernel $G_{m,3}$ reads 
\begin{equation} \label{gm3-2} 
G_{m,3}(r,t )  = \gamma_m \, 2^{-2|m-\alpha|-2} \left(\frac{t^{2-|m-\alpha|}}{4(1+|m-\alpha|)}\, v_m(0,r) +t^{-|m-\alpha|}\, \partial_\lambda v_m(0,r) \right),
\end{equation} 
and where $ \widetilde X_m(\lambda; r,t )$  is an inegral kernel of an operator $\widetilde X_m(\lambda)$ which satisfies 
\begin{equation}  \label{z-tilde}
\sup_{m\in\Z} \big\|\, \x^{-s}\,  \widetilde X_m(\lambda)  \, \x^{-s}\big\|_{\hs((0,1)\times(1,\infty), r dr)}  = \mathcal{O}(\lambda^2)\, .
\end{equation}
Moreover, equation \eqref{gm3-2} in combination with estimates \eqref{vm-estim-1} and \eqref{gamma-upperb} implies that 
 \begin{equation} \label{hs-gm3-2}
 \big\|\, \x^{-s}\, G_{m,3} \, \x^{-s}\big\|_{\hs((0,1)\times(1,\infty), r dr)}  =  \mathcal{O}(4^{-|m-\alpha|} (1+|m|)^{-1}) \qquad \forall\, s>3. 
\end{equation} 

\smallskip

\begin{lem} \label{lem-Zm-2} 
Let $s>3$. Then
$$
\sup_{m\in\Z} \big\|\, \x^{-s}\, X_{m}(\lambda)  \, \x^{-s}\big\|_{\hs((0,1)\times(1,\infty), r dr)}  =  \mathcal{O}(\lambda^2)\, 
 $$
\end{lem}

\begin{proof}
Equations \eqref{um-estim}, \eqref{beta-upperb} together with \eqref{vm-estim-1}, \eqref{vm-estim-2} and the definitions of $h_m$ and $j_m$, see equations \eqref{h_m} and \eqref{j-m}, show that 
\begin{align*}
\big\|\, \x^{-s}\, \big(  g_m(r) \,j_m(t) + g_m(t)\,  j_m(r) \big) \, \x^{-s}\big\|_{\hs((0,1)\times(1,\infty), r dr)} & =  \mathcal{O}((1+m^2)^{-1})\\
\big\|\, \x^{-s}\, \big(  g_m(r) \,h_m(t) + g_m(t)\,  h_m(r) \big) \, \x^{-s}\big\|_{\hs((0,1)\times(1,\infty), r dr)} & =  \mathcal{O}((1+m^2)^{-1})
\end{align*}
The claim now follows from \eqref{case2-2} and \eqref{z-tilde}. 
\end{proof}

\subsubsection*{\bf Case $r <t\leq 1 $} 
We will use the notation 
\begin{equation*}
e_m = \partial_\lambda u_m(\lambda,1) \big |_{\lambda=0} \, , \qquad e'_m = \partial_\lambda u'_m(\lambda,1) \big |_{\lambda=0} \, .
\end{equation*}
Note also that by \eqref{sm-eq2} and \eqref{wronski} 
\begin{equation} \label{ab-wronski}
 a_m \, b'_m\, - a'_m\, b_m\, = \frac{\Gamma(1+2|m|)}{\Gamma(|m|-m+1/2)} \qquad \forall\, m\in\Z.
\end{equation}

\medskip

\noindent By \cite[Sec.5]{ko} we  have 
\begin{equation}\label{r-kernel-3}
R_{m,0}(\lambda; r,t ) = \ell_m(\lambda;r,t) - \frac{C_m(\lambda) \, k_m(\lambda;r,t)}{B_m(\lambda)-i A_m(\lambda)}\, ,
\end{equation} 
where $ \ell_m(\lambda;r,t), k_m(\lambda;r,t)$ are defined in \eqref{k-kappa}, and
\begin{align} 
C_m(\lambda) &=  (u'_m(\lambda,1)-u_m(\lambda,1)\, |m-\alpha| )\, J_{|m-\alpha|}(\sqrt{\lambda}) +\sqrt{\lambda}\ u_m(\lambda,1)\,  J_{|m-\alpha|+1}(\sqrt{\lambda})   \nonumber \\
& \quad  +i \, \big[ (u'_m(\lambda,1)-u_m(\lambda,1)\, |m-\alpha| )\, Y_{|m-\alpha|}(\sqrt{\lambda}) +\sqrt{\lambda}\ u_m(\lambda,1)\,  Y_{|m-\alpha|+1}(\sqrt{\lambda})\, \big ] 
. \label{cm-eq} 
\end{align}
From \eqref{eq-jy} we find the following expansion for $C_m(\lambda)$ ; 
\begin{align*}
C_m(\lambda) & = \frac{ (b'_m +|m-\alpha|\, b_m) \, 2^{|m-\alpha|}\, \lambda^{-\frac{|m-\alpha|}{2}}}{i \sin(|m-\alpha|\, \pi) \Gamma(1-|m-\alpha|)} \Big[ \, 1+ C_m^{(1)}\, \lambda +C_m^{(2)}\, \lambda^{|m-\alpha|} +\mathcal{O}(\epsilon_m \lambda^2) \Big], 
\end{align*} 
where $\epsilon_m$ is given by \eqref{pm12-estim}, and where 
$$
C_m^{(1)} = \frac{e'_m + |m-\alpha|\, e_m}{b'_m +|m-\alpha|\, b_m} +\frac{1}{1+|m-\alpha|}\, , \qquad C_m^{(2)} =  2^{|m-\alpha|}\,  e^{-i |m-\alpha|\pi}\, \epsilon_m\, 
\frac{b'_m - |m-\alpha|\, b_m}{b'_m +|m-\alpha|\, b_m} \, .
$$
This in combination with \eqref{b-ia}, \eqref{r-kernel-3}  and Lemma \ref{lem-kl}  then gives equation \eqref{case2-2} with 
\begin{align} \label{gm3-3} 
G_{m,3}(r,t ) &  = \partial_\lambda \ell_m(0;r,t) - \beta_m \,  \partial_\lambda k_{m}(0;r,t) -C_m^{(1)}\, \beta_m \, k_m(0;r,t)\, ,
\end{align} 
and with 
\begin{equation} \label{zm-3} 
\sup_{m\in\Z} \big\|\, \widetilde X_{m}(\lambda)  \big\|_{\hs((0,1)^2, r dr)}  =  \mathcal{O}(\lambda^2)\,  .
\end{equation}
Moreover, since  $C_m^{(1)} = \mathcal{O}(1)$, from Lemma \ref{lem-kl} we conclude that
 \begin{equation} \label{hs-gm3-3}
\sup_{m\in\Z} \big\|\, G_{m,3} \, \big\|_{\hs((0,1)^2, r dr)}  =  \mathcal{O}(1) \, .
\end{equation} 
Now let $G_3$ be the integral operator with the kernel 
\begin{equation} \label{g3-kernel}
G_3(x,y) = \sum_{m\in\Z} G_{m,3}(r,t )\, e^{ im (\theta-\theta')}\, ,
\end{equation} 
where $ G_{m,3}(r,t )$ is given by \eqref{gm3-1}, \eqref{gm3-2} and \eqref{gm3-3}. Equations \eqref{hs-gm3-1}, \eqref{hs-gm3-2} and  \eqref{hs-gm3-3} thus show that $\G_3 \in 
\B_0(0,s;0,-s)$ for any $s>3$.

\begin{proof}[\bf Proof of Proposition \ref{prop-exp}]  
With the above preparatory results at hand, it remains to collect the estimates which control the remainder terms. 
First we introduce some additional notation. Let 
\begin{equation*} 
\| u\|_{\HH_0^{k,s}} = \|\, \x^s (1+H_0)^{k/2}\, u\|_{\Lp^2(\R^2)}, 
\end{equation*}
and recall that $H_0$ is given by \eqref{H0-def}. Moreover, we denote by
$$
\HH_0^{k,s} = \big\{\, u: \| u\|_{\HH_0^{k,s}} < \infty\big\}
$$
the associated weighted Sobolev space, and by 
$$
\B_0(k,s,k',s') = \B(\HH^{k,s}, \HH^{k',s'}) 
$$
the space of bounded linear operators from $\HH_0^{k,s}$ into $\HH_0^{k',s'}$. 
Lemmata \ref{lem-Zm-1}, \ref{lem-Zm-2}, and equations \eqref{hs-gm3-1}, \eqref{hs-gm3-2}, \eqref{zm-3} imply that expansion 
\eqref{B0-eq-1} holds in $\B(0,s;0,-s)=\B_0(0,s;0,-s)$. From $(1+H_0)R_0(\lambda) = 1+(1+\lambda) R_0(\lambda)$ we infer that \eqref{B0-eq-1} holds in 
$\B_0(-2,s;0,-s)$, and therefore also in $\B_0(-1,s;1,-s)$. Since $A_0$ is bounded,  the Sobolev norms $\|u\|_{\HH_0^{1,s}}$ and $\|u\|_{\HH^{1,s}}$ are equivalent. 
We thus conclude that \eqref{B0-eq-1} holds in $\B(-1,s;1,-s)$. 
\end{proof}

\section{\bf Operator $R_0(\lambda)$; proof of Proposition \ref{prop-exp-int}} 
\label{sec-app-c}
\noindent  When $\alpha\in\Z$, then the integral kernel of $R_0(\lambda)$ is still given by equations \eqref{r0-eq}. We will treat separately the channels $m=\alpha, \ m=\alpha\pm 1$ and $m=\alpha\pm 2$, since these are the only ones which include logarithmic factors of order smaller than $\mathcal{O}(\lambda^2)$.  
In what follows we will calculate the exact form of these terms for each of these values of $m$ and each of the three regions of $\R^2$ considered above. We denote
$$
\tilde\gamma =\gamma -\log 2
$$
where $\gamma$ is the Euler constant. Note that the expression for $J_\nu$ in \eqref{eq-jy} continues to hold even when $\nu=n\in\N$, while for $Y_\nu$ we have to use \cite[Eq.~9.1.11]{as} instead. The latter states that as $z\to 0$,
\begin{equation}\label{yn-eq}
\begin{aligned}  
Y_0(z)  & = \frac 2\pi \big( \log z +\tilde\gamma\big) J_0(z) +\frac{z^2}{2\pi} +  \mathcal{O}(z^4) =  \frac 2\pi \big( \log z +\tilde\gamma\big) \Big (1-\frac{z^2}{4}+\frac{z^4}{64} \Big) +  \mathcal{O}(z^4)   \\[4pt]
Y_n(z)  & = -\frac{2^n}{\pi z^n} \sum_{k=0}^{n-1} \frac{(n-k-1)!\, z^{2k}}{k!\, 4^k}\,  +\frac 2\pi \log(z/2)\, J_n(z) + \frac{z^n}{2^n \pi} \Big[2\gamma -\sum_{k=1}^{n-1} k^{-1}\Big] +  \mathcal{O}(z^{n+2}) \, , \quad n\geq 1 .
\end{aligned} 
\end{equation}
In particular, for $n=1,2$ we have 
\begin{equation} \label{jy1} 
\begin{aligned} 
J_1(z) & = \frac z2 \Big(1-\frac{z^2}{8}\Big) +  \mathcal{O}(z^5),  \qquad Y_1(z)  = -\frac{2}{\pi z} +\frac{z}{2\pi} \big( \log (z^2) -1+2\tilde\gamma \big) -\frac{z^3}{16\pi}\, \log(z^2) +  \mathcal{O}(z^3)\, ,\\[4pt]
J_2(z) & = \frac{z^2}{8} \Big(1-\frac{z^2}{12}\Big) +  \mathcal{O}(z^6), \ \, \quad
Y_2(z)  = -\frac{4}{\pi z^2} -\frac 1\pi +\frac{z^2}{8\pi} \big( \log (4z^2) -2+4\tilde\gamma \big) +  \mathcal{O}(z^4\log z) .
\end{aligned}
\end{equation}

As mentioned above, we will next expand the integral kernel of $R_0(\lambda)$ separately for $m=\alpha, \, m=\alpha\pm 1,\, m=\alpha\pm 2$ and $|m-\alpha| \geq 3$. 
This is done by a series of elemntary, though quite lengthly calculations. Therefore we give full details only in the case  $1<r<t$. 

\subsubsection*{\bf Case $1 <r<t $}  Recall that integral kernel of $R_0(\lambda)$  in this region is given by equation \eqref{r-kernel}. The mot important feature of these calculations is the fact that the term 
 $A_m(\lambda)\,  J_{|m-\alpha|}(\sqrt{\lambda}\, r) +B_m(\lambda)\, Y_{|m-\alpha|}(\sqrt{\lambda}\, r)$ contains no logarithmic factors of order less than $\mathcal{O}(\lambda^2)$.

\noindent Let $\underline{m=\alpha}$. This is the  only channel contributing to the operator $\G_1$.  Let
\begin{equation} \label{z-alpha}
z_\alpha= i\pi -2\tilde\gamma + \frac{2a_\alpha}{a'_\alpha} \,  , 
\end{equation}
and let us adopt the notation 
\begin{equation} \label{b-frak-m}
s_m = \partial_\lambda v_m(\lambda,1) \big |_{\lambda=0} \, , \qquad s'_m = \partial_\lambda v'_m(\lambda,1) \big |_{\lambda=0} \, .
\end{equation}
We need to expand $A_\alpha(\lambda)$ and $B_\alpha(\lambda)$ up to order $\mathcal{O}(\lambda^2)$. From \eqref{eq-ABm}, \eqref{eq-jy} and \eqref{yn-eq},\eqref{jy1} we find 
\begin{align} \label{A-alpha}
A_\alpha(\lambda) &= \frac 2\pi (a_\alpha -a'_\alpha \tilde\gamma) -\frac{a'_\alpha}{\pi}\, \log\lambda  +\frac{E_1}{\pi} \, \lambda\log\lambda +\frac{E_2}{4\pi} \,\lambda^2\log\lambda - \frac{E_3}{2\pi}\, \lambda   + \mathcal{O}(\lambda^2), 
\end{align}
where 
\begin{align*}
E_1 &  =  \frac{a'_\alpha}{4} -s'_\alpha -\frac{a_\alpha}{2}, \qquad E_ 2 = s'_\alpha -\frac{a'_\alpha}{16} +\frac{a_\alpha}{2} -2s_\alpha, \qquad
E_3= a'_\alpha(1-\tilde\gamma) +a_\alpha(2\tilde\gamma-1) +4s_\alpha' \tilde\gamma-4s_\alpha\, .
\end{align*}
Similarly, it holds 
\begin{align} \label{B-alpha}
B_\alpha(\lambda) &= a'_\alpha -E_1 \lambda + \frac{E_2}{4} \, \lambda^2 + o(\lambda^2)\, .
\end{align}
Using \eqref{eq-jy} and \eqref{yn-eq} once more we then infer that the terms proportional to $\log\lambda, \,  \lambda\log\lambda,$ and to $\lambda^2\log\lambda$ in the  $A_\alpha(\lambda)\,  J_0(\sqrt{\lambda}\, r) +B_\alpha(\lambda)\, Y_0(\sqrt{\lambda}\, r) $ cancel out, and 
\begin{align} \label{aj+by} 
A_\alpha(\lambda)\,  J_0(\sqrt{\lambda}\, r) +B_\alpha(\lambda)\, Y_0(\sqrt{\lambda}\, r) & = \frac{2a'_\alpha}{\pi} \, \g_\alpha(r) +\frac{\lambda}{2\pi} \, \hk_\alpha(r) \  + \mathcal{O}(\lambda^2(1+ r^4))\, ,
\end{align} 
with $\g_\alpha(\cdot)$  given by \eqref{g-alpha}, and with 
\begin{equation} \label{h-frak-alpha}
\hk_\alpha(r) = a_\alpha'\, r^2(1-  \g_\alpha(r) ) - 4 E_1 \Big(\tilde\gamma -\g_\alpha(r) + \frac{a_\alpha}{a_\alpha'} \Big)  -E_3 .
\end{equation}
Next we evaluate the denominator of  \eqref{r-kernel}. Using \eqref{A-alpha} and \eqref{B-alpha} we obtain
\begin{equation}  \label{wronski-alpha}
\begin{aligned} 
(B_\alpha(\lambda)-i A_\alpha(\lambda))^{-1} &=  \frac{-i\pi}{a'_\alpha( \log \lambda-z_\alpha)} \, \left[1 + \frac{E_1 \lambda\log\lambda}{a'_\alpha( \log \lambda-z_\alpha) }  -\frac{(E_3+2i \pi E_1)\, \lambda}{2 a'_\alpha( \log \lambda-z_\alpha) }+ \mathcal{O}(\lambda^2) \right] \\[3pt]
&= \frac{-i\pi}{a'_\alpha( \log \lambda-z_\alpha)} \, \left[1 + \frac{E_1  \lambda}{a'_\alpha}  +\frac{\lambda }{2 a'_\alpha( \log \lambda-z_\alpha) }\, \Big( 4E_1\Big (\frac{a_\alpha}{a_\alpha'}-\tilde\gamma\Big) -E_3\Big)+ \mathcal{O}(\lambda^2) \right] \, . 
\end{aligned}
\end{equation} 
Moreover, by \eqref{eq-jy} and \eqref{yn-eq}
\begin{align*}
J_0(\sqrt{\lambda}\, t) +iY_0(\sqrt{\lambda}\, t) &= \Big(1-\frac{\lambda t^2}{4}\Big) \Big[1+ \frac{2i}{\pi} \big(\log (\sqrt{\lambda}\, t) +\tilde\gamma\big)\Big] +\frac{i \lambda t^2}{2\pi} +  \mathcal{O}(\lambda^2 t^4)\\[4pt]
&= \frac{i(  \log \lambda-z_\alpha)}{\pi} \left [ 1 + \frac{2\g_\alpha(t)}{ \log \lambda-z_\alpha} -\frac{t^2\lambda}{4} +\frac{t^2\lambda (1-\g_\alpha(t))}{2( \log \lambda-z_\alpha)}\right] + \mathcal{O}(\lambda^2(1+ t^4)) \, .
\end{align*}
Hence
\begin{equation} \label{jy0-wron}
\frac{a'_\alpha \big(J_0(\sqrt{\lambda}\, t) +iY_0(\sqrt{\lambda}\, t)\big)}{B_\alpha(\lambda)-i A_\alpha(\lambda)} =1+  \frac{2\g_\alpha(t)}{ \log \lambda-z_\alpha}  +\frac{\lambda}{4a'_\alpha} \big(4E_1-a'_\alpha t^2\big) + \frac{\lambda \hk_\alpha(t)}{2a'_\alpha( \log \lambda-z_\alpha)} +\frac{2 a_\alpha' E_4\, \lambda \g_\alpha(t)}{( \log \lambda-z_\alpha)^2}+\mathcal{O}(\lambda^2(1+ t^4))\,  ,
\end{equation}
with 
\begin{equation} \label{e-4}
E_4 = \frac{2 E_1}{a_\alpha'} \Big (\frac{a_\alpha}{a_\alpha'}-\tilde\gamma\Big) -\frac{E_3}{2 a_\alpha'}\, .
\end{equation} 
Inserting the above expansions into equation \eqref{r-kernel} finally yields
\begin{align} \label{m-alpha}
R_{\alpha,0}(\lambda;r,t) & =\G_{\alpha,0}(r,t) +  \frac{\g_\alpha(r)\, \g_\alpha(t)}{\pi(\log\lambda -z_\alpha)} \Big[1+\frac{E_4\, \lambda}{(\log\lambda -z_\alpha)}\Big]
+\lambda\, \G_{\alpha,3}(r,t) +  \frac{\lambda \big[ \g_\alpha(r) \hk_\alpha (t) + \g_\alpha(t) \hk_\alpha (r) \big] }{4\pi a'_\alpha (\log\lambda -z_\alpha)}  
\nonumber \\
& \quad + \mathcal R^{(1)}_{\alpha,0}(\lambda;r,t), 
\end{align}
where we have abbreviated 
\begin{equation}  \label{G-alpha-3-1}
\G_{\alpha,3}(r,t)  = \frac{1}{8 \pi a'_\alpha} \big ( \hk_\alpha(r) +\g_\alpha(r) (4 E_1 -a'_\alpha t^2) \big),
\end{equation} 
and where the remainder term satisfies 
$$
\mathcal R^{(1)}_{\alpha,0}\  = \Big[ A_\alpha(\lambda)\,  J_0(\sqrt{\lambda}\, r) +B_\alpha(\lambda)\, Y_0(\sqrt{\lambda}\, r) \Big] \, \mathcal{O}(\lambda^2(1+ t^4))
+ \frac{ J_0(\sqrt{\lambda}\, t) +iY_0(\sqrt{\lambda}\, t)}{B_\alpha(\lambda)-i A_\alpha(\lambda)}\, \mathcal{O}(\lambda^2(1+ r^4))\, .
$$
From \eqref{JJ} and \eqref{A-alpha} we thus easily deduce that
\begin{equation}  \label{R1-estim-1}
\| \x^{-s}\,\mathcal R^{(1)}_{\alpha,0}\, \x^{-s}\|_{\hs((1,\infty)^2, r dr)}  = \mathcal{O}(\lambda^2)  \qquad \forall\, s> 3.
\end{equation}
$\underline{m= \alpha\pm 1}$ : These are the only two channels contributing to the operator $\G_2$.  
We adopt  the notation introduced in \eqref{b-frak-m} and for the sake of brevity abbreviate 
$$
a_\pm= a_{\alpha\pm 1}, \qquad a'_\pm= a'_{\alpha\pm 1}, \qquad s_\pm = s_{\alpha\pm 1}, \qquad s'_\pm = s'_{\alpha\pm 1}\, .
$$
Using  \eqref{eq-ABm} ,\eqref{yn-eq} and \eqref{jy1} we then find
\begin{equation} \label{ab-lambda-2}
\begin{aligned}
A_{\alpha\pm 1} (\lambda)  & = \frac{1}{\pi\sqrt{\lambda}} \Big[ 2(a_\pm +a'_\pm) + E_\pm \, \lambda + \frac{a_\pm-a'_\pm}{2} \ \lambda \log\lambda\, +F_\pm \,\lambda^2\log\lambda +\mathcal{O}(\lambda^2)\Big] , \\
B_{\alpha\pm 1}(\lambda) & =   \frac 12\, (a'_\pm-a_\pm)\ \sqrt{\lambda} \ - F_\pm\,  \lambda^{3/2} +\mathcal{O}(\lambda^2)\, ,
\end{aligned}
\end{equation}
where
$$
 E_\pm  = \frac 12 (a_\pm+a'_\pm) +\tilde\gamma(a_\pm-a'_\pm) +2(s_\pm+s'_\pm) , \quad F_\pm = \frac{1}{16} \big[ 8(s_\pm-s'_\pm) +a'_\pm -3a_\pm\big] \, .
$$
Hence
\begin{equation} \label{B-i-A} 
B_{\alpha\pm 1}(\lambda) -i A_{\alpha\pm 1}(\lambda) = \frac{-2i}{\gamma_\pm \pi\, \sqrt{\lambda}} \Big[ 1+ \frac{\lambda E_\pm \gamma_\pm}{2} -\frac{\delta_\pm}{4}\, \lambda(\log\lambda-i\pi) +\frac{F_\pm \gamma_\pm }{2} \, \lambda^2 (\log\lambda-i\pi) +\mathcal{O}(\lambda^2) \Big] .
\end{equation}
Recall also that the functions $g_\pm$ are defined by \eqref{short-pm}.
A careful calculation then shows that the terms proportional to $\lambda \log \lambda$ and $\lambda^2\log\lambda$ in $A_{\alpha\pm 1} (\lambda)   J_1(\sqrt{\lambda}\, r) +  B_{\alpha\pm 1}(\lambda)  Y_1(\sqrt{\lambda}\, r)$ cancel out, and 
\begin{equation} \label{aj1+by1}
A_{\alpha\pm 1} (\lambda)   J_1(\sqrt{\lambda}\, r) +  B_{\alpha\pm 1}(\lambda)  Y_1(\sqrt{\lambda}\, r) = \frac{ g_\pm(r) }{\gamma_\pm \pi}  + \lambda\, \hk_\pm(r) + \mathcal{O}(\lambda^2(1+r^3)) , 
\end{equation}
with $\gamma_\pm$ given by \eqref{short-pm}, and with 
\begin{equation} \label{h-frak-pm}
\hk_\pm(r) = \frac{1}{4\pi} \Big[\, 2 r\Big (E_\pm-\frac{r^2}{4 \gamma_\pm}\, \Big) +(a'_\pm-a_\pm)\, (2\log r -1+2\tilde\gamma)\, r +\frac{8 F_\pm}{r} \, \Big] .
\end{equation}
To continue we denote 
\begin{equation}  \label{j-frak-pm} 
\mathfrak{j}_\pm(r) =
  \frac{\gamma_\pm}{4}
\left\{
\begin{array}{l@{\qquad}l}
\ 4\, \partial_\lambda v_{\alpha\pm1}(0,r) & r \leq 1  , \\[3pt]
\ \pi \hk_\pm(r) & 1  < r .
 \end{array}
\right .
\end{equation}
From \eqref{ab-lambda-2} and \eqref{jy1} we thus deduce  that
\begin{align} \label{frac-alpha-pm}
\frac{J_1(\sqrt{\lambda}\, t) +iY_1(\sqrt{\lambda}\, t)}{B_{\alpha\pm 1}(\lambda)-i A_{\alpha\pm 1}(\lambda)} &=\gamma_\pm \Big[\, \frac 1t-\frac{ g_\pm(t)}{4} \,\lambda (\log\lambda -i\pi)- f_\pm(t)\, \lambda  +\frac{\delta_\pm\, g_\pm(t)}{16}\, \big(\lambda(\log\lambda-i\pi)\big)^2 -\frac{\mathfrak{j}_\pm(t)}{16}\,  \lambda^2(\log\lambda-i\pi) 
\Big] \, \nonumber\\ 
& \quad \ + \mathcal{O}(\lambda^2(1+t^3))
\end{align} 
where $\delta_\pm$ is given by \eqref{short-pm}, and where
\begin{equation} \label{f-pm-ipi}
f_\pm(t) = \frac{E_\pm\, \gamma_\pm}{2 t}+ \frac t4 \big(2 \log t -1 +2\tilde\gamma \big).
\end{equation}
This together with \eqref{r-kernel} and \eqref{aj1+by1},\eqref{j-frak-pm}  gives 
\begin{align} \label{m-alpha-pm}
 R_{\alpha\pm 1,0}(\lambda;r,t)  &=G_{\alpha\pm 1,0}(r,t)  -\frac{1}{64\pi}\, g_\pm(r)\, g_\pm(t) \, \lambda( \log\lambda-i\pi) \big[4+ \delta_\pm \, \lambda( \log\lambda-i\pi) \, \big ] +\lambda\, \G_{\alpha\pm 1,3}(r,t) \\
&
\quad -\frac{1}{64\pi}\, \lambda^2(\log\lambda-i\pi)  \, \big[g_\pm(r)\, \mathfrak j_\pm(t)  +g_\pm(t)\, \mathfrak j_\pm(r) \big ]
+\mathcal{R}^{(1)}_{\alpha\pm 1,0}(\lambda;r,t), \nonumber 
\end{align}
with the integral kernel of the linear term given by 
\begin{equation}  \label{G-pm1-3-1}
\G_{\alpha\pm 1,3}(r,t)  =  \frac{1}{4\pi}  \big[g_\pm(r) f_\pm(t) + 2\,   \mathfrak j_\pm(r)\, t^{-1}  \big] \, ,
\end{equation} 
and 
with the remainder term satisfying
\begin{equation}  \label{R1-estim-2}
\| \x^{-s}\,\mathcal{R}^{(1)}_{\alpha\pm 1,0}\, \x^{-s}\|_{\hs((1,\infty)^2, r dr)} = \mathcal{O}(\lambda^2) \qquad s> 3.
\end{equation}
$\underline{m= \alpha\pm 2}$: From \eqref{eq-jy} and \eqref{yn-eq} we derive 
\begin{equation}\label{ABm-alpha+2}
\begin{aligned}
A_{\alpha\pm 2} (\lambda) & = \frac{4}{\pi \lambda} \Big[ 2a_{\alpha\pm 2} +a'_{\alpha\pm 2} +\frac\lambda 4\big( 8\, s_{\alpha\pm 2} +a'_{\alpha\pm 2} \big) +\frac{(2a_{\alpha\pm 2}-a'_{\alpha\pm 2})}{32}\, \lambda^2\log\lambda + \mathcal{O}(\lambda^2)\Big] \\
B_{\alpha\pm 2} (\lambda) & = \frac\lambda 8\, (a'_{\alpha\pm 2}-2a_{\alpha\pm 2}) +\frac{\lambda^2}{8} \Big( 2s'_{\alpha\pm 2}-4s_{\alpha\pm 2} +\frac{4a_{\alpha\pm 2}-a'_{\alpha\pm 2}}{12}\Big) + \mathcal{O}(\lambda^3)\, .
\end{aligned}
\end{equation}
Hence the terms proportional to $\lambda^2\log\lambda$ in $A_{\alpha\pm 2} (\lambda)   J_2(\sqrt{\lambda}\, r) +  B_{\alpha\pm 2}(\lambda)  Y_2(\sqrt{\lambda}\, r) $ cancel out, and we obtain
\begin{equation*} 
A_{\alpha\pm 2} (\lambda)   J_2(\sqrt{\lambda}\, r) +  B_{\alpha\pm 2}(\lambda)  Y_2(\sqrt{\lambda}\, r) = \frac{g_{\alpha\pm 2}(r)}{2\pi \, \gamma_{\alpha\pm 2}}\,  + \hk_{\alpha\pm 2}(r) \, \lambda + \mathcal{O}(\lambda^2(1+r^3) ),
\end{equation*}
with 
$$
\hk_{\alpha\pm 2}(r) = r^2 \Big( \frac{s_{\alpha\pm 2}}{\pi} +\frac{a'_{\alpha\pm 2}}{8\pi} -\frac{r^2}{24\, \gamma_{\alpha\pm 2}} \,  \Big)\, .
$$
We have 
\begin{align} \label{jy2-wrons}
\frac{J_2(\sqrt{\lambda}\, t) +iY_2(\sqrt{\lambda}\, t)}{B_{\alpha\pm 2}(\lambda)-i A_{\alpha\pm 2}(\lambda)} &= -\frac{Y_2(\sqrt{\lambda}\, t)}{ A_{\alpha\pm 2}(\lambda)} +  
 i \pi \gamma_{\alpha\pm 2}\, \frac{g_{\alpha\pm 2}(t)}{32}\, \lambda^2  +
\mathcal{O}(\lambda^2(1+t^3) ) \nonumber \\
&=  \gamma_{\alpha\pm 2} \Big[ t^{-2} +\frac\lambda 4 (1-4t^{-2} ) -\frac{g_{\alpha\pm 2}(t)}{32}\, \lambda^2(\log\lambda -i\pi) \Big] + \mathcal{O}(\lambda^2(1+t^3) ),
\end{align} 
where the remainder term is real-valued. Hence
inserting the above expansions into \eqref{r-kernel} gives
\begin{align} \label{m-alpha-pm-2}
 R_{\alpha\pm 2,0}(\lambda;r,t)  & =G_{\alpha\pm 2,0}(r,t) +\lambda\, \G_{\alpha\pm 2,3}(r,t) - g_\pm(r) g_\pm(t) \frac{\lambda^2( \log\lambda-i\pi)}{256 \pi}
 +\mathcal{R}^{(1)}_{\alpha\pm 2,0}(\lambda;r,t),
\end{align}
where  the kernel $\G_{\alpha\pm 2,3}(r,t) $ is given by 
\begin{equation}  \label{G-pm1-3-2}
\G_{\alpha\pm 2,3}(r,t)  =  \frac{1}{4 t^2} \Big[ \gamma_{\alpha\pm 2}\, \hk_{\alpha\pm 2}(r) + \frac{g_{\alpha\pm 2}(r)}{8\pi} \, (t^2-4) \Big] \, ,
\end{equation} 
and where the remainder term is self-adjoint and satisfies
\begin{equation}  \label{R1-estim-3}
\| \x^{-s}\,\mathcal{R}^{(1)}_{\alpha\pm 2,0}\, \x^{-s}\|_{\hs((1,\infty)^2, r dr)}  = \mathcal{O}(\lambda^2) \qquad s> 3.
\end{equation}
 \underline{$|m-\alpha| \geq 3$} : These channels do not contribute to logarithmic corrections up to order $\mathcal{O}(\lambda^2)$. 
From  \eqref{eq-ABm}, \eqref{eq-jy}, \eqref{yn-eq} and \eqref{b-frak-m} we learn that 
\begin{equation} \label{abm-expand}
\begin{aligned} 
A_m(\lambda) & = \frac{2^{|m-\alpha|}\, (|m-\alpha|-2)!}{\pi\, \lambda^{|m-\alpha|/2}}\ \big (A_m^{(0)} + \lambda\, A_m^{(1)} +  \mathcal{O}(\lambda^2)\big) \\
B_m(\lambda) & = \frac{\lambda^{|m-\alpha|/2}}{|m-\alpha|!\, 2^{|m-\alpha|+1}}\ \big (B_m^{(0)} + \lambda\, B_m^{(1)} +  \mathcal{O}(\lambda^2)\big) ,
\end{aligned}
\end{equation}
where the remainder terms are uniform in $m$, and where 
\begin{align*}
A_m^{(0)} & = (|m-\alpha|\, a_m +a'_m) (|m-\alpha|-1), \quad A_m^{(1)} = \frac 14 \big[ a'_m -a_m +4 (|m-\alpha|-1) (|m-\alpha|\, s_m-s'_m )\big] \\
B_m^{(0)} &=2 (a'_m- |m-\alpha|\, a_m) , \qquad   \qquad  \qquad\  B_m^{(1)} = \frac{|m-\alpha| -2a_m -a'_m}{2(|m-\alpha|+1)} -2 (|m-\alpha|\, s_m-s'_m) 
\end{align*}
Hence 
\begin{align*}
& A_m(\lambda)\, J_{|m-\alpha|}(\sqrt{\lambda}\, r) + B_m(\lambda)\, Y_{|m-\alpha|}(\sqrt{\lambda}\, r)   =  \frac{(|m-\alpha|\, a_m +a'_m)  g_m(r)  }{\pi |m-\alpha|} 
- \frac{\lambda\, r^{-|m-\alpha|} \big(4 B_m^{(1)} - r^2 \big)}{8\pi ( (m-\alpha)^2-|m-\alpha|)}  \\
&  \qquad\qquad \qquad\qquad\qquad+ \frac{\lambda\, r^{|m-\alpha|} \big( 4(|m-\alpha|+1) A_m^{(1)} - r^2\big )  }{\pi |m-\alpha| ( (m-\alpha)^2-1)}  
+ \mathcal{O}(\lambda^2(1+r^3))  , 
\end{align*} 
and 
\begin{align*}
\frac{J_{|m-\alpha|}(\sqrt{\lambda}\, t) +iY_{|m-\alpha|}(\sqrt{\lambda}\, t)}{B_m(\lambda)-i A_m(\lambda)} &= -\frac{Y_{|m-\alpha|}(\sqrt{\lambda}\, t)}{ A_m(\lambda)} +  \mathcal{O}(\lambda^2(1+t^3) ) \nonumber \\
&= \frac{t^{-|m-\alpha|}}{A_m^{(0)} } \Big(|m-\alpha|-1 +\frac{t^2\, \lambda}{4}\Big) \Big(1-\frac{A_m^{(1)}\, \lambda}{A_m^{(0)}}\Big) +\mathcal{O}(\lambda^2(1+t^3) )\, .
\end{align*} 
Inserting these estimates in \eqref{r-kernel} gives
\begin{align} \label{r-m-1}
R_{m,0}(\lambda;r,t) & =G_{m,0}(r,t) + \G_{m,3}(r,t)\, \lambda\,  +\mathcal{R}^{(1)}_{m,0}(\lambda;r,t),
\end{align}
with the coefficient of the linear given by 
\begin{align}  \label{hm-1}
\G_{m,3}(r,t) & =  \frac{A_m^{(0)}\, t^2 - 4(|m-\alpha|-1) A_m^{(1)}}{4\pi \, t^{|m-\alpha|}\, ( (m-\alpha)^2-|m-\alpha|)} \, \frac{g_m(r)}{A_m^{(0)}} \\
& \qquad +  \frac{t^{-|m-\alpha|}}{8\pi |m-\alpha| A_m^{(0)} }  \left[2r^{|m-\alpha|} \Big(A_m^{(1)} -\frac{r^2}{4(|m-\alpha|+1)}\Big)- r^{-|m-\alpha|} \big(B_m^{(1)} +\frac{r^2}{4} \big)\right] \nonumber
\end{align}
and with the remainder term which  satisfies
\begin{equation}  \label{R1-estim-4}
\sup_{m\in\Z} \| \x^{-s}\,\mathcal{R}^{(1)}_{m,0}\, \x^{-s}\|_{\hs((1,\infty) r dr)}  =  \mathcal{O}(\lambda^2) \qquad s> 3.
\end{equation}

\subsubsection*{\bf Case $r <1<t $}  Let \underline{$m=\alpha$} :  We set
\begin{equation} \label{j-frak-alpha} 
\mathfrak{j}_\alpha(r) =
  \frac{1}{4 a'_\alpha}
\left\{
\begin{array}{l@{\qquad}l}
\ 4\, \partial_\lambda v_{\alpha}(0,r) & r \leq 1  , \\[3pt]
\ \hk_\alpha(r) & 1  < r . 
 \end{array}
\right .
\end{equation} 
From \eqref{jy0-wron} and \eqref{r-kernel-2} we then deduce  that
\begin{align} \label{m-alpha-2}
R_{\alpha,0}(\lambda;r,t) & =\G_{\alpha,0}(r,t)+   \frac{\g_\alpha(r)\, \g_\alpha(t)}{\pi(\log\lambda -z_\alpha)} \Big[1+\frac{E_4\, \lambda}{ \log\lambda -z_\alpha}\Big]
+\lambda\, \G_{\alpha,3}(r,t) +  \frac{\lambda \big[ \g_\alpha(r) \, \mathfrak{j}_\alpha (t) + \g_\alpha(t)\,  \mathfrak{j}_\alpha (r) \big] }{\pi  (\log\lambda -z_\alpha)} + \mathcal{O}(\lambda^2\, t^2) \nonumber \\
\end{align}
where we have abbreviated
\begin{align} \label{G-alpha-3-2}
\G_{\alpha,3}(r,t) & =  \frac{\mathfrak{j}_\alpha(r)}{2\pi} + \frac{\g_\alpha(r)}{8 \pi a'_\alpha} \, (4 E_1 -a'_\alpha t^2) \, . 
\end{align}

\medskip

\noindent  \underline{$m=\alpha\pm 1$} : We  proceed similarly as above.From \eqref{r-kernel-2} and \eqref{frac-alpha-pm} we get
\begin{align} \label{m-pm-2}
 R_{\alpha\pm 1,0}(\lambda;r,t)  &=G_{\alpha\pm 1,0}(r,t)  -\frac{g_\pm(r)\, g_\pm(t) }{64\pi}\ \lambda (\log\lambda-i\pi) \big[4+\delta_\pm \, \lambda( \log\lambda -i\pi) \big ] +\G_{\alpha\pm 1, 3}(r,t)\, \lambda\\
&\quad  -\frac{ \lambda^2(\log\lambda-i\pi)}{64 \pi} \, [g_\pm(r)\, \mathfrak{j}_\pm(t) +g_\pm(t)\, \mathfrak{j}_\pm(r) ]
+\mathcal{O}(\lambda^2(1+t^3))  \nonumber ,
\end{align}
with the kernel of $\G_{\alpha\pm 1, 3}$ defined by 
\begin{equation} \label{g-pm1-3-3}
\G_{\alpha\pm 1, 3}(r,t)=  \frac{1}{4\pi} \big[g_\pm(r) f_\pm(t) +2\, \mathfrak{j}_\pm(r) \, t^{-1}  \big] 
\end{equation}

\noindent $\underline{m= \alpha\pm 2}$: From  \eqref{r-kernel-2}, \eqref{jy2-wrons}, and \eqref{v-integ-estim} we obtain
\begin{align} \label{m-pm-2-b}
 R_{\alpha\pm 2,0}(\lambda;r,t)  & =G_{\alpha\pm 2,0}(r,t) +\G_{\alpha\pm 2, 3}(r,t)\, \lambda  - g_{\alpha\pm 2}(r) g_{\alpha\pm 2}(t) \frac{\lambda^2 (\log\lambda-i\pi)}{256 \pi}
 +\mathcal{O}(\lambda^2(1+t^3)) . 
\end{align}
with the kernel of $\G_{\alpha\pm 2, 3}$ defined by 
\begin{equation} \label{g-pm2-3-3}
\G_{\alpha\pm 2, 3}(r,t)= \frac{\lambda\, t^{-2}}{2\pi  } \Big[ \frac{g_{\alpha\pm 2}(r)}{16} \, (t^2-4) -\gamma_{\alpha\pm 2}\, \partial_\lambda v_{\alpha\pm 2}(0,r) \Big] .
\end{equation}

 \underline{$|m-\alpha| \geq 3$}: Let 
\begin{equation} \label{g-tilde}
\tilde g_m(\lambda,r) := -\frac{2 \sin(\pi |m-\alpha|) v_m(\lambda, r)}{ \mathscr{P}_{m,1}(\lambda)}\, \qquad r<1,
\end{equation}
so that 
\begin{equation*}
\tilde g_m(0,r) = \frac{ \pi\, g_m(r)}{ \Gamma(1+|m-\alpha|)} \qquad \text{and}\qquad  \pd_\lambda \tilde g_m(0,r)  = j_m(r), \qquad r<1. 
\end{equation*}
With the help of \eqref{case2-1} we then deduce from \eqref{r-kernel-2}  that 
\begin{align*} 
R_{m,0}(\lambda;r,t) & = \frac{\tilde g_m(\lambda, r)\, Q_{-|m-\alpha|}(\lambda,t)}{4\pi \sin(|m-\alpha|\pi)}  + \mathcal{R}^{(2)}_m(\lambda;r,t) ,
\end{align*}
where in view of \eqref{QQ-bound} and \eqref{v-integ-estim}  
$$
\sup_{m\in\Z} \|\, \x^{-s}\,  \mathcal{R}^{(2)}_m(\lambda) \, \x^{-s}\|_{\hs((0,1)\times (1,\infty), r dr)}  =  \mathcal{O}(\lambda^2)\,  , \qquad \forall\, s>3.
$$
Hence by \eqref{g-tilde} and \eqref{Q-nu}
\begin{align} \label{r-m-2}
R_{m,0}(\lambda;r,t) & =\G_{m,0}(r,t)
 + \lambda\, \G_{m,3}(r,t)+ \mathcal{O}_{\Lp^2((0,1)\times(1,\infty))}(\lambda^2)
\end{align}
where the remainder is uniform in $m$, and where 
\begin{align} \label{hm-2}
\G_{m,3}(r,t) & =  \frac{\gamma_m}{2\pi\, t^{|m-\alpha|}}\,\left[\partial_\lambda v_m(0,r) +  \frac{v_m(0,r)\, t^2}{4(|m-\alpha|-1)} \right]\, .
\end{align}

\subsubsection*{\bf Case $r <t\leq 1 $}  

 \underline{$m=\alpha$} : By  \eqref{ab-wronski}, \eqref{cm-eq} and  \eqref{wronski-alpha} ,
\begin{align*} 
\frac{C_\alpha(\lambda)}{B_\alpha(\lambda)-i A_\alpha(\lambda)} &= \frac{b'_\alpha}{a_\alpha'} -\frac{2\Gamma(1+2|\alpha|)}{\sqrt{\pi}\, (a'_\alpha)^2 \, (\log \lambda-z_\alpha)} 
+E_5\, \lambda\, +\frac{ \lambda}{\log \lambda-z_\alpha} -\frac{\Gamma(1+2|\alpha|)\, E_4\, \lambda}{\sqrt{\pi}\, (a'_\alpha)^3 \, (\log \lambda-z_\alpha)^2}  +  \mathcal{O}(\lambda^2) \, , 
\end{align*}
with 
$$
E_5 = \frac{1}{4 a_\alpha'} \Big(2b_\alpha-b'_\alpha+4 E_1 \frac{b'_\alpha}{a'_\alpha}\Big).
$$ 
By inserting this result into \eqref{r-kernel-3} and using Lemma \ref{lem-kl} we obtain, with an obvious abuse of notation, 
\begin{align*}
R_{\alpha,0}(\lambda;r,t) & =  \G_{\alpha,0}(r,t)+ \g_\alpha(r) \g_\alpha(t) \frac{ \log\lambda-z_\alpha+ E_4\, \lambda}{\pi (\log\lambda-z_\alpha)^2}
 + \lambda\, \G_{\alpha,3}(r,t) +  \frac{\lambda \big[ \g_\alpha(r)\, \mathfrak{j}_\alpha (t) + \g_\alpha(t)\, \mathfrak{j}_\alpha (r) \big] }{2\pi  (\log\lambda -z_\alpha)}+\mathcal{O}_{\Lp^2(0,1)}(\lambda^2)
\end{align*}
where the integral kernel of $\G_{\alpha,3}$ satisfies 
\begin{align} \label{g-alpha-3-3}
\G_{\alpha,3}(r,t) & =  \partial_\lambda \ell_\alpha(0;r,t) -\frac{b'_\alpha}{ a'_\alpha} \  \partial_\lambda k_\alpha(0;r,t) - E_5 (a'_\alpha)^2\, \g_\alpha(r)\, \g_\alpha(t)\, .
\end{align}
Recall that the functions $k_m$ and $\ell_m$ are defined in \eqref{k-kappa}.

  \underline{$m=\alpha\pm 1$} :  Below, as usual, we abbreviate $\beta_{\alpha\pm 1} =\beta_\pm$.
Recall that the coefficients $\beta_m$ are defined by \eqref{delta-m}.
With the help of  \eqref{cm-eq}, \eqref{jy1},  and \eqref{ab-lambda-2} we find 
\begin{align*}
\frac{C_{\alpha\pm 1}(\lambda)}{B_{\alpha\pm 1}(\lambda)-i A_{\alpha\pm 1}(\lambda)}  & = \beta_{\pm} -  \gamma_\pm^2\, (a_\pm b'_\pm-a'_\pm b_\pm)\, \frac{\lambda(\log\lambda-i\pi)}{2} 
-\frac{\gamma_\pm^2\, \lambda}{2}\,  \big[ \tilde\gamma  (a_\pm b'_\pm-a'_\pm b_\pm) +(s_\pm+s'_\pm)(b_\pm +b'_\pm)\big]
\\
&  \quad +E_6\, \lambda^2(\log\lambda-i\pi) +\gamma_\pm^2\, \delta_\pm (a_\pm b'_\pm-a'_\pm b_\pm)\, \frac{(\lambda(\log\lambda-i\pi))^2}{8}  + \mathcal{O}(\lambda^2)\, ,
\end{align*} 
where as usual we have abbreviated $\beta_{\alpha\pm 1} =\beta_\pm$, and where $E_6= E_5+(s_\pm+s'_\pm)(b_\pm +b'_\pm)$. 
Inserting this  into \eqref{r-kernel-3} and using equation \eqref{ab-wronski} together with Lemma \ref{lem-kl} yields 
\begin{equation} \label{m-pm-3}
\begin{aligned}
 R_{\alpha\pm 1,0}(\lambda;r,t)  &=G_{\alpha\pm 1,0}(r,t)  -\frac{g_\pm(r)\, g_\pm(t) }{64\pi}\ \lambda (\log\lambda-i\pi) \big[4 +\delta_\pm\, \lambda (\log\lambda-i\pi) \big ] + \G_{\alpha\pm 1,3}(r,t) \, \lambda \\
&\quad -\frac{ \lambda^2(\log\lambda-i\pi)}{64 \pi} \, \big[g_\pm(r)\, \mathfrak{j}_\pm(t) +g_\pm(t)\, \mathfrak{j}_\pm(r) \big]
+ \mathcal{O}_{\Lp^2(0,1)}(\lambda^2).
\end{aligned}
\end{equation}
with the integral kernel  of $\G_{\alpha\pm 1,3}$ given by
\begin{align} \label{g-pm1-3-4}
\G_{\alpha\pm 1,3}(r,t) & =  \partial_\lambda \ell_{\alpha\pm 1}(0;r,t)  -\beta_{\pm } \,  \partial_\lambda k_{\alpha\pm 1}(0;r,t)\, +\frac{g_\pm(r) g_\pm (t)}{16\pi} \big[ \tilde\gamma +(s_\pm+s'_\pm)(b_\pm+b'_\pm) \big].
\end{align}

 \underline{$|m-\alpha| = 2$} : From \eqref{cm-eq}, \eqref{jy1}, and \eqref{ABm-alpha+2} we deduce that 
 \begin{align*} 
\frac{C_{\alpha\pm 2}(\lambda)}{B_{\alpha\pm 2}(\lambda)-i A_{\alpha\pm 2}(\lambda)}  &  
=\beta_{\alpha\pm 2}  + E_7 \, \lambda\ + \frac{\gamma_{\alpha\pm 2}^2 \, \lambda^2( \log\lambda-i\pi)}{8} \, \big( b_{\alpha\pm 2}\, a'_{\alpha\pm 2} -b'_{\alpha\pm 2}\, a _{\alpha\pm 2} \big)  +  \mathcal{O}(\lambda^2) \, ,
\end{align*}
with
$$
E_7 = \frac{\gamma_{\alpha\pm 2}}{4} \big [ (8s_{\alpha\pm 2}+a_{\alpha\pm 2}') +b_{\alpha\pm 2}'+2b_{\alpha\pm 2} +4e'_{\alpha\pm 2} +8e_{\alpha\pm 2} -\gamma_{\alpha\pm 2}\big] \, .
$$
In view of\eqref{r-kernel-3}, \eqref{ab-wronski} and Lemma \ref{lem-kl} we thus conclude with 
\begin{equation}
\begin{aligned} \label{m-pm-4}
R_{\alpha\pm 2,0}(\lambda;r,t) & =  G_{\alpha\pm 2,0}(r,t) + \lambda\, \G_{\alpha\pm 2,3}(r,t)  - g_{\alpha\pm 2}(r) \, g_{\alpha\pm 2}(t)\,  \frac{\lambda^2 (\log\lambda-i\pi)}{256 \pi} \, +\,  \mathcal{O}_{\Lp^2(0,1)}(\lambda^2)
\end{aligned}
\end{equation}
where the integral kernel of $\G_{\alpha\pm 2,3}$  reads 
\begin{align} \label{g-pm2-3-4}
\G_{\alpha\pm 2,3}(r,t)  & =  \partial_\lambda \ell_{\alpha\pm 2} (0;r,t) - \beta_{\alpha\pm 2} \,  \partial_\lambda k_{\alpha\pm 2}(0;r,t) \, -\frac{ E_7\, \pi}{16 }\ g_{\alpha\pm 2}(r)\, g_{\alpha\pm 2}(t)\, .
\end{align}

\underline{$|m-\alpha| \geq 3$} : Equations \eqref{cm-eq} and \eqref{yn-eq}  imply 
$$
C_m(\lambda) = -\frac{i\, 2^{|m-\alpha|} (|m-\alpha|-1)!}{\pi \lambda^{|m-\alpha|/2}}\,  \Big [\, b_m |m-\alpha|  +b'_m + \lambda\,  \mathcal C_m^{(1)} + \mathcal{O}\big(\lambda^2) \Big], 
$$
where $ \mathcal C_m^{(1)}= |m-\alpha|\, e_m +e'_m +\frac 14(b_m +b'_m/|m-\alpha| )$, and where the remainder term $\mathcal{O}\big(\lambda^2)$ is uniform in $m$. Hence from \eqref{abm-expand} and  \eqref{r-kernel-3} we get 
\begin{align*}
\frac{C_m(\lambda)}{B_m(\lambda)-i A_m(\lambda)}  & = \beta_m + \gamma_m \Big( \mathcal C_m^{(1)} -\frac{\gamma_m \, A_m^{(1)}}{|m-\alpha|-1} \Big)\, \lambda +  \mathcal{O}\big(\lambda^2).
\end{align*} 
This in combination Lemma \ref{lem-kl} gives 
\begin{equation}
 \label{m-pm-5}
R_{m,0}(\lambda;r,t)  =  G_{m,0}(r,t) + \lambda\, \G_{m,3}(r,t)  +\mathcal{O}_{\Lp^2(0,1)}(\lambda^2)
\end{equation}
with the remainder term uniform in $m$ and with
\begin{equation} \label{hm-3} 
\G_{m,3}(r,t) = \Big( \mathcal C_m^{(1)} -\frac{\gamma_m \, A_m^{(1)}}{|m-\alpha|-1} \Big)\, \gamma_m\, k_m (0;r,t)+ \partial_\lambda \ell_{m} (0;r,t) - \beta_{m} \,  \partial_\lambda k_m(0;r,t)  .
\end{equation}

\subsection*{\bf Operator $\G_3$} \label{ssec-g3-int} 
The integral kernel of $\G_3$ splits as follows 
\begin{equation} \label{g3-kernel-int}
\G_3(x,y) = \sum_{m\in\Z} \G_{m,3}(r,t )\, e^{ im (\theta-\theta')}\, ,
\end{equation} 
where $ \G_{m,3}(r,t )$ is, for the corresponding values of $m$ and $r,t$,  by equations \eqref{G-alpha-3-1}, \eqref{G-pm1-3-1}, \eqref{G-pm1-3-2}, \eqref{hm-1}, \eqref{G-alpha-3-2}, \eqref{g-pm1-3-3}, \eqref{g-pm2-3-3}, \eqref{hm-2}, \eqref{g-alpha-3-3}, \eqref{g-pm1-3-4}, \eqref{g-pm2-3-4} and \eqref{hm-3}.
From the latter equations and from Lemma \ref{lem-kl} and equation 	\eqref{v-integ-estim} we then deduce that
\begin{equation} \label{g3-int-bound}
\sup_{m\in\Z} \|\, \x^{-s}\,  \G_{m,3} \, \x^{-s}\|_{\hs(\R^2)}  \ <  \ \infty\,  , \qquad \forall\, s>3. 
\end{equation}
.


\begin{proof}[\bf Proof of Proposition \ref{prop-exp-int}] 
Let $s>3$. As in the case of integer flux, it suffices to prove expansion \eqref{B0-eq-2-int} in $\B(0,s;0,-s)$.
The bound \eqref{g3-int-bound} and equations  \eqref{G1-int}-\eqref{g5-int} show that all the operators on the right hand side of \eqref{B0-eq-2-int} belong to $\B(0,s;0,-s)$. 
 From the results collected in this section, in particular from equations  \eqref{R1-estim-1}, \eqref{R1-estim-2},  \eqref{R1-estim-3},  \eqref{R1-estim-4}, \eqref{m-alpha-2}, \eqref{m-pm-2}, \eqref{m-pm-2-b}, \eqref{r-m-2},  \eqref{m-pm-3}, \eqref{m-pm-4} and \eqref{m-pm-5}, we then deduce that the remainder term in \eqref{B0-eq-2-int}  is of order $\mathcal{O}(\lambda^2)$. 
 \end{proof}

\section*{\bf Acknowledgements}
\noindent 
I'm indebted to thank Rupert Frank for useful comments on a preliminary version of the paper. 

\end{document}